\documentclass[a4paper,reqno]{amsart}
\usepackage[english]{babel}
\newcommand\version{April 21, 2022}

\usepackage{amsmath,amsfonts,amsthm,amssymb,amsxtra}
\usepackage{hyperref}
\usepackage{bbm} 
\usepackage{bm} 
\usepackage{color}
\usepackage{enumerate}

\newtheorem{theorem}{Theorem}[section]
\newtheorem{proposition}[theorem]{Proposition}

\theoremstyle{definition}

\theoremstyle{remark}

\newtheorem{remark}[theorem]{Remark}
\newtheorem{remarks}[theorem]{Remarks}

\numberwithin{equation}{section}

\newcommand{\C}{\mathbb{C}}
\newcommand{\const}{\mathrm{const}\ }
\newcommand{\D}{\mathcal{D}}

\renewcommand{\epsilon}{\varepsilon}

\newcommand{\N}{\mathbb{N}}

\renewcommand{\phi}{\varphi}
\newcommand{\R}{\mathbb{R}}

\newcommand{\Z}{\mathbb{Z}}

\DeclareMathOperator{\arsinh}{arsinh}

\DeclareMathOperator{\diag}{diag}

\DeclareMathOperator{\dom}{dom}

\DeclareMathOperator{\spec}{spec}

\DeclareMathOperator{\sgn}{sgn}
\DeclareMathOperator{\Tr}{Tr}
\DeclareMathOperator{\tr}{Tr}

\def\bc{\mathbb{C}}
\def\br{\mathbb{R}}

\def\bs{\mathbb{S}}

\def\gz{\mathbb{Z}}
\def\zp{\dot{\gz}}
\def\rz{\mathbb{R}}

\def\gA{\mathfrak{A}}
 
\def\gh{\mathfrak{h}}

\def\ca{\mathcal{A}}
\def\ce{\mathcal{E}}
\def\cg{\mathcal{G}}
\def\ch{\mathcal{H}}
\def\ci{\mathcal{I}}
\def\ck{\mathcal{K}}

\def\cm{\mathcal{M}}
\def\co{\mathcal{O}}
\def\cs{\mathcal{S}}
\def\ct{\mathcal{T}}
\def\cx{\mathcal{X}}

\def\rd{\mathrm{d}}

\def\dk{\mathrm{d}k}
\def\dr{\mathrm{d}r}
\def\ds{\mathrm{d}s}

\def\dx{\mathrm{d}x}
\def\dy{\mathrm{d}y}

\def\gtf{\gamma_\mathrm{TF}}

\def\uA{\underline{A}}
\def\ualpha{\underline{\alpha}}
\def\uB{\underline{B}}

\def\ur{\underline{r}}
\def\uR{\underline{R}}

\def\uvarrho{\underline{\varrho}}
\def\usigma{\underline{\sigma}}
\def\uz{\underline{z}}
\def\uZ{\underline{Z}}

\def\bPhi{\bm{\Phi}}

\newcommand{\me}[1]{\mathrm{e}^{#1}}

\newcommand{\one}{\mathbf{1}}

\makeatletter
\newcommand*{\rom}[1]{\expandafter\@slowromancap\romannumeral #1@}
\makeatother

\begin{document}

\title[Scott conjecture for large Coulomb systems  --- \version]{The Scott conjecture for large Coulomb systems: a review}

\author[R. L. Frank]{Rupert L. Frank}
\address[Rupert L. Frank]{Mathematisches Institut, Ludwig-Maximilans Universit\"at M\"unchen, Theresienstr. 39, 80333 M\"unchen, Germany, and Munich Center for Quantum Science and Technology (MCQST), Schellingstr. 4, 80799 M\"unchen, Germany, and Mathematics 253-37, Caltech, Pasa\-de\-na, CA 91125, USA}
\email{r.frank@lmu.de}

\author[K. Merz]{Konstantin Merz}
\address[Konstantin Merz]{Institut f\"ur Analysis und Algebra, Technische Universit\"at Braunschweig, Universit\"atsplatz 2, 38106 Braunschweig, Germany}
\email{k.merz@tu-bs.de}

\author[H. Siedentop]{Heinz Siedentop}
\address[Heinz Siedentop]{Mathematisches Institut, Ludwig-Maximilans Universit\"at M\"unchen, Theresienstr. 39, 80333 M\"unchen, Germany, and Munich Center for Quantum Science and Technology (MCQST), Schellingstr. 4, 80799 M\"unchen, Germany}
\email{h.s@lmu.de}


\date{\version}

\begin{abstract}
  We review some older and more recent results concerning the energy and particle distribution in ground states of heavy Coulomb systems. The reviewed results are asymptotic in nature: they describe properties of many-particle systems in the limit of a large number of particles. Particular emphasis is put on models that take relativistic kinematics into account. While non-relativistic models are typically rather well understood, this is generally not the case for relativistic ones and leads to a variety of open questions.
\end{abstract}



\maketitle
\tableofcontents

\section{Introduction and historical background}
\label{s:introduction}

\subsection{Many-particle quantum mechanics}

Properties of ground states of large Coulomb systems involving $N$
electrons, such as atoms or molecules, are of fundamental interest
in quantum physics and chemistry.
Notable examples are the \emph{ground state energy} and the electron
distribution in the ground state. The latter may be expressed in terms
of the \emph{one-particle ground state density}, i.e., the probability
density of finding one of the $N$ electrons at a specific location in
$\R^3$.
It is well known that systems on atomic length scales are accurately
described by quantum mechanics \cite{Heisenberg1925,Heisenberg1926}.
This understanding relies on precise investigations of the underlying
Hamilton operator.

We consider a molecule that consists of $K$ point-like nuclei of charges
$\uZ=(Z_1,...,Z_K)\in(0,\infty)^K$, fixed at pairwise different positions
$\uR=(R_1,...,R_K)\in\R^{3K}$, as well as $N$ electrons, all interacting
via Coulomb potentials in the Born--Oppenheimer approximation. The total
nuclear charge is $|\uZ|:=\sum_{\kappa=1}^K Z_\kappa$. The number of spin
degrees of freedom is denoted by $q\in\N$. Although in reality $q=2$,
one may, for notational convenience, choose $q=1$ when the
spin-dependence is trivial.
 
A non-relativistic quantum mechanical description of this system is
provided by the operator
\begin{align}
  \label{eq:manybodySchrodinger}
  H_{N,V} := \sum_{\nu=1}^N\left(-\frac12\Delta_\nu-V(x_\nu)\right)
  + \sum_{1\leq\nu<\mu\leq N}\frac{1}{|x_\nu-x_\mu|}
  + U
  \quad \text{in}\ \bigwedge_{\nu=1}^N L^2(\br^3:\bc^q)
\end{align}
with
\begin{align}
  \label{eq:defv}
  V(x) = \sum_{\kappa=1}^K \frac{Z_\kappa}{|x-R_\kappa|}
\end{align}
and
\begin{align}
  \label{eq:defu}
  U = \sum_{1\leq \kappa<\kappa'\leq K}\frac{Z_\kappa Z_{\kappa'}}{|R_\kappa-R_{\kappa'}|}.
\end{align}
We choose Hartree units, so that
$\hbar=e=m=1$, where $\hbar$, $e$, and $m$, denote the rationalized
Planck constant, the elementary charge, and the electron mass,
respectively. 
In the atomic case ($K=1$, $\uR=0$), we have $U=0$ and write
$H_{N,Z}\equiv H_{N,Z/|x|}$.

Since electrons are fermions, they obey the Pauli exclusion principle,
i.e., the Hilbert space in which the operator \eqref{eq:manybodySchrodinger}
acts is given by $\bigwedge_{\nu=1}^N L^2(\br^3:\bc^q)$, i.e., the subspace of
$L^2(\br^{3N}:\bc^{q^N})$ consisting of all square-integrable,
$\C^{q^N}$-valued functions whose sign changes under the exchange of any
two particle coordinates. 

We write
\begin{align}
  \label{eq:gsenergynonrel}
  E_S(N,\uZ,\uR) := \inf\spec(H_{N,V})
\end{align}
for the lowest spectral point of the Hamiltonian $H_{N,V}$.
%
This number $E_S(N,\uZ,\uR)$ may or may not be an eigenvalue, and, if it
is, it may be degenerate.  While the results in this review concern
\eqref{eq:gsenergynonrel}, there is an important, related quantity,
which has not received the mathematical attention it deserves; see
\eqref{eq:gs} below.

In the atomic case ($K=1$, $\uR=0$, $\uZ=Z$), we write
$E_S(N,Z):=E_S(N,Z,\underline{0})$ and, for neutral atoms,
$E_S(Z):=E_S(Z,Z,0)$. It is well known that $E_S(N,Z)$ is an
eigenvalue when $N<Z+1$, see Zhislin \cite{Zislin1960} or Simon
\cite{Simon1970}.

In addition to ground state energies, we will be interested in
one-particle ground state densities. We recall that the one-particle
density of a general (pure) state
$\psi\in \bigwedge_{\nu=1}^N L^2(\br^3:\bc^q)$ is defined by
\begin{align}
  \label{eq:defrhohnonrel}
  \rho(x) := N \sum_{\sigma=1}^q\int_{\Gamma^{N-1}} |\psi(x,\sigma;y_2,...y_N)|^2\,\dy_2...\dy_N
\end{align}
for $x\in\R^3$.
Here $\Gamma:=\R^3\times\{1,2,...,q\}$. Elements $y\in\Gamma$ are
space-spin variables and the corresponding measure $\dy$ is the
product measure consisting of Lebesgue measure on $\R^3$ and counting
measure on $\{1,...,q\}$.

If $\psi$ in \eqref{eq:defrhohnonrel} is an eigenfunction of $H_{N,V}$
with eigenvalue $E_S(N,Z,R)$, we write
\begin{align}
  \label{eq:defrhogroundstateintro}
  \rho_S
\end{align}
for its density and analogously for other Hamiltonians that we discuss
later. Although this is might be an abuse of notation since the
eigenvalue $E_S(N,Z,R)$ could be degenerate, our statements about $\rho_S$
will be true for any choice of an eigenfunction. The notion of a
one-particle density and, in particular, of a one-particle ground
state density can be generalized to the case of mixed states, but we
do not do this here.  Also, if the lowest point in the spectrum
$E_S(N,Z,R)$ is not an eigenvalue, one can still obtain meaningful
statements for so-called approximate ground states, but we will not
discuss this in this introduction.

The goal of this review is to summarize known results and open
questions concerning the ground state energy and the one-particle
ground state density in the limit of large electron numbers and
nuclear charges for non-relativistic and, especially, for certain
relativistic descriptions of Coulomb systems.  In the rest of this
introduction we will focus on results for non-relativistic atoms.
This and other settings will be treated in more detail in later sections,
see the table of contents and Subsection \ref{ss:organization} for
relevant pointers.

\medskip
\begin{remarks}
  Some remarks on our goals are in order.
  \begin{enumerate}
  \item It is well known that the spectral analysis of $N$-particle systems
    \emph{for fixed $N$} is prohibitively difficult already when $N\geq2$,
    since the $\co(N^2)$ many inter-particle interactions prohibit
    a reduction to a three-dimensional (possibly) soluble
    one-particle problem. (For instance, if the electron-electron repulsion
    was absent and $K=1$ in $H_{N,V}$, then one could separate variables
    to end up with the direct sum of Schr\"odinger operators describing
    hydrogen.)
    Instead, one often considers the properties of a system for a large
    number of particles.
    This leads to the study of \emph{asymptotic properties}.
    In this review we entirely focus on results in the limit
    $Z_1,...,Z_K,N\to\infty$. The precise way of carrying out this
    limit when $K>1$ is explained later.
  
  \item Studying asymptotics clearly leads to less quantitative
    mathematical statements and is also questionable from a physical
    point of view since experimentally observed values of $Z$
    are bounded, e.g., by $92$ for stable atoms.
    However, the mathematical analysis is drastically simpler and,
    interestingly, leads to theorems that coincide astonishingly well
    with experimentally measured data.
    (This observation has been made repeatedly in different contexts
    in mathematical physics. Stell \cite[p.~48]{Stell1977} calls it the
    \emph{principle of unreasonable utility of asymptotic expansions}
    and makes some interesting philosophical remarks.)

  \item There are some notable exceptions, however.
    \begin{enumerate}
    \item For instance, recently much progress has been made in the
      investigation of smoothness properties of single eigenfunctions
      and sums of squares of eigenfunctions of many-particle Coulomb
      Hamiltonians, such as $H_{N,Z}$, for fixed $N$. In this regard
      see, e.g., the works
      \cite{Fournaisetal2005,Fournaisetal2008,Fournaisetal2009,Fournaisetal2009A,FournaisSorensen2021}
      for non-relativistic and \cite{FournaisSorensen2010} for
      (pseudorelativistic) Chandrasekhar atoms.
      Such a priori estimates for many-particle eigenfunctions are
      important, e.g., for the derivation of eigenvalue asymptotics
      for the associated one-particle density matrix \cite{Sobolev2021}.

    \item Another example concerns the maximal ionization of atoms
      (and mol\-e\-cules).  Experiments indicate that doubly or higher
      charged anions do not exist (Massey
      \cite{Massey1976,Massey1979}), i.e., one expects at most $Z+1$
      many electrons to be bound to the nucleus, while any further
      electrons are located infinitely far away with vanishingly small
      kinetic energy.  Proving this claim is a notoriously difficult
      problem in mathematical physics, see, e.g., Nam
      \cite{Nam2020,Nam2022} for recent reviews.  A slightly weaker
      formulation, the so-called \emph{ionization conjecture}, states
      that there is a number $Q<\infty$ such that, if $E_S(N,Z)$ is an
      eigenvalue, then $N\leq Z+Q$.  Two well known results in this
      direction are due to Lieb \cite{Lieb1984} and Fefferman and Seco
      \cite{FeffermanSeco1989,FeffermanSeco1990}, who proved that
      $N<2Z+1$ and $N\leq Z+ C Z^{\frac{47}{56}}$, respectively, are
      necessary conditions for $E_S(N,Z)$ to be an eigenvalue.
      Recently, Nam \cite{Nam2012N} improved Lieb's result and showed
      $N<1.22Z+3Z^{\frac 13}$, which leads to a sharper result when
      $Z\geq6$.  Lieb's result implies the fact that doubly negatively
      charged hydrogen atoms do not exist.
    \end{enumerate}
  \end{enumerate}
\end{remarks}

\subsection{Glimpse at Thomas--Fermi density functional theories}
\label{ss:glimpsetf}

The $N$ particle quantum Coulomb problem of computing $E_S(N,\uZ,\uR)$ and
the associated eigenspace is -- like its classical analogue, the Kepler
problem -- 
prohibitively difficult to solve (even numerically) already for $N\geq2$
because of the $\co(N^2)$ many interactions between the $N$ electrons.
This necessitates the derivation of so-called ``effective theories'', i.e.,
energy functionals or equations, which depend only on a fixed, but small
number of variables, like three or six, and describe at least the macroscopically
observed properties of the given system ``sufficiently accurately''.
Although these theories are usually more accessible to numerical analysis,
they also pose some interesting mathematical challenges in view of the
presence of non-linearities, which simulate the interparticle
interactions. Here we focus on so-called \emph{density functional theories},
i.e., energy functionals, that only depend on the one-particle density.

\begin{remark}
  We chose to bypass density \emph{matrix} functionals
  (e.g., due to Hartree \cite{Hartree1928PartI,Hartree1928PartII},
  Fock \cite{Fock1930}, Slater \cite{Slater1930}, M\"uller \cite{Muller1984},
  and Sharma et al.~\cite{Sharmaetal2008}),
  as it would go far beyond the scope of this review. In addition to
  referring to \cite{Siedentop2020M}, we highlight pioneering works by
  Lieb and Simon \cite{LiebSimon1974,LiebSimon1977T}, Bach
  \cite{Bach1992,Bach1993}, Graf and Solovej \cite{GrafSolovej1994},
  as well as the works \cite{Franketal2007},
  \cite{Siedentop2009,Siedentop2014}, and Kehle
  \cite{Kehle2017}.
\end{remark}

For simplicity, assume from now on the neutral, atomic case ($K=1$,
$N=Z$).  The breakthrough in the description of ground state
properties of $H_{N,V}$ came with the help of a particularly simple
density functional theory, the so-called Thomas--Fermi theory
\cite{Thomas1927,Fermi1927,Fermi1928}, which will be reviewed in
Subsection \ref{ss:tf}.  In their seminal work \cite{LiebSimon1977},
Lieb and Simon connected Thomas--Fermi theory to the quantum problem
of finding $E_S(Z)$ and showed that the Thomas--Fermi energy
$E^{\rm TF}(Z)$, i.e., the infimum of the Thomas--Fermi functional
$\ce_Z^{\rm TF}$ (Lenz \cite{Lenz1932}), is the leading term of the
asymptotic expansion of $E_S(Z)$ when $Z\to\infty$.  The Thomas--Fermi
energy scales like $E^{\rm TF}(Z)=E^{\rm TF}(1) \cdot Z^{7/3}$, which
is a consequence of
$$
\ce_Z^{\rm TF}[Z^2\rho(Z^{1/3}\cdot)] = Z^{7/3} \ce_Z^{\rm TF}[\rho].
$$
Thus, the result of Lieb and Simon for the ground state energy reads
\begin{align}
  \label{eq:lsintro}
  E_S(Z) = E^{\rm TF}(1) \cdot Z^{7/3} + o(Z^{7/3})
  \quad \text{as}\ Z\to\infty,
\end{align}
see also Theorem \ref{quantumtfconvnonrelatom}.
A numerical computation shows that
$E^{\rm TF}(1)\approx -0.484\,29\cdot q^{2/3}$,
cf.~Gomb{\'a}s~\cite[p.~60]{Gombas1949}.

Figuratively speaking, the leading order in \eqref{eq:lsintro} is generated
by the bulk of the electrons, which are located on distances $\co(Z^{-1/3})$
from the nucleus, and are described semiclassically.
It should not come as a surprise that this energetic result is
accompanied by a result connecting the quantum ground state density $\rho_S$
with the minimizer of $\ce_Z^{\rm TF}$, the Thomas--Fermi density
$\rho_Z^{\rm TF}$. Indeed, Lieb and Simon \cite{LiebSimon1977}, and
Baumgartner \cite{Baumgartner1976} showed that the suitably rescaled
ground state density $\rho_S$ converges to the minimizer of the
Thomas--Fermi theory for hydrogen. More precisely, one has, due to the
scaling properties of Thomas--Fermi theory, the convergence
\begin{align}
  \lim_{Z\to\infty}Z^{-2}\rho_S(Z^{-1/3}\,\cdot \,) = \rho_{Z=1}^{\rm TF}
\end{align}
when both sides are integrated against characteristic functions of
bounded, measurable subsets of $\R^3$. In the context of the ionization
conjecture, Fefferman and Seco \cite{FeffermanSeco1989} obtained (as a
corollary) the convergence in a stronger topology, namely in the so-called
Coulomb norm; see \eqref{eq:defcoulombnorm}. The precise result is
contained in Theorem \ref{quantumtfconvnonrelatom}.

\subsection{Quantum effects close to the nucleus}

Although Thomas--Fermi theory correctly predicts the leading order of
$E_S(Z)$, it turned out that, as Scott \cite[p.~859]{Scott1952} wrote
in 1952, the Thomas--Fermi energy gives values for the binding energy
$(-1)\cdot E_S(Z)$ ``\emph{which are too high by roughly $20 \%$. The
  actual binding energies increase quite smoothly with increasing $Z$,
  which suggests the existence of a more appropriate formula.}''
Naturally, this defect of simple Thomas--Fermi theory triggered some
discussions. One year before Scott's publication, Foldy
\cite{Foldy1951} had proposed the formula
$E_S(Z)=c_1 \cdot Z^{12/5}+c_2(Z)$.  Here $c_2(Z)$ depends on $E_S(2)$
and the sum of the ionization potentials of all atoms with atomic
number greater or equal three and less or equal $Z$. (Foldy does not
give a bound on $c_2(Z)$ but seems to assume that
$c_2(Z)=o(Z^{12/5})$.)  More importantly for us, $c_1$ is a constant
that only depends on the chosen units and obeys $c_1>E^{\rm TF}(1)$.
The exponent $12/5$ was derived from numerical values of the
electrostatic potential close to the nucleus as a function of $Z$
(Dickinson \cite{Dickinson1950}). Since these values were only
available for $Z\leq 80$, Foldy's formula was not expected to hold
asymptotically as $Z\to\infty$.  In the discussion of his formula
Foldy \cite[p.~398]{Foldy1951} points out that Thomas--Fermi theory
does not correctly take into account the following two effects close
to the nucleus.  On the one hand, such electrons are bound stronger to
the nucleus, but, on the other hand, they screen the bulk of the
electrons at larger distances to the nucleus.  Foldy suspected the
screening to dominate, which explains the inequality
$c_1>E^{\rm TF}(1)$ despite the fact that $Z^{12/5}\gg Z^{7/3}$ for
$Z\gg1$.

Scott \cite[p.~867]{Scott1952} made Foldy's observations more precise
and suggested a different formula for $E_S(Z)$.
He believed that Thomas--Fermi theory does correctly describe the leading
order of the ground state energy expansion, but that the failure of
Thomas--Fermi theory
``\emph{is due partly to the shortcomings of the statistical model in
  the region nearest the nucleus, and partly to the effect of exchange}''.
Like Foldy, Scott suggested that the few, but high-energy electrons that
are located close to the nucleus should generate this correction. Due to
their proximity to the nucleus, these electrons should be described
quantum mechanically. Since the correction would be generated only by
``finitely many'' electrons, the electron-electron repulsion should
be irrelevant and the order of the correction should be $\co(Z^2)$,
i.e., in agreement with the magnitude of the eigenvalues of the
hydrogenic Hamiltonian
\begin{align}
  \label{eq:defhydrogenoperator}
  S_Z^{H}:=-\frac12\Delta - \frac{Z}{|x|}
  \quad \text{in}\ L^2(\R^3:\C)
\end{align}
with nuclear charge $Z$.
By a simple calculation
(see Subsection \ref{sss:scottliebarguments} for March's derivation),
Scott was led to
\begin{align}
  \label{eq:scottnonrelinitial}
  E_S(Z) = E^{\rm TF}(1) \cdot Z^{7/3} + \frac{q}{4} \cdot Z^2 + o(Z^2)
  \quad \text{as}\ Z\to\infty.
\end{align}

If one drops the electron-electron repulsion in the Hamiltonian
\eqref{eq:manybodySchrodinger}, the corresponding ground state energy
will also behave to leading order like a constant times $Z^{7/3}$ as
$Z\to\infty$, but with a constant different from $E^{\rm TF}(1)$.
There will also be a subleading correction, given by a constant times
$Z^2$, and, remarkably, the constant here is the same $q/4$ as in
\eqref{eq:scottnonrelinitial}. This is not a coincidence and will
become clear in the discussion below.

\subsubsection{Scott correction}

About thirty years later, Lieb \cite[Problem 6]{Lieb1980} and
\cite[pp.~623--624]{Lieb1981} and Simon \cite[Problem 10b]{Simon1984}
revisited the problem of finding the
second term in the asymptotic expansion of $E_S(Z)$.
Because of Scott's compelling arguments,
Formula \eqref{eq:scottnonrelinitial} was coined
\emph{Scott correction/conjecture}.

In the same decade, Hughes \cite{Hughes1986,Hughes1990} (lower bound)
and the authors of
\cite{SiedentopWeikard1987O,SiedentopWeikard1989,SiedentopWeikard1991}
(upper and lower bound) proved this conjecture. That is, they rigorously
derived the expansion \eqref{eq:scottnonrelinitial};
see Theorem~\ref{swscottnonrel}.
The proof in
\cite{SiedentopWeikard1987O,SiedentopWeikard1989,SiedentopWeikard1991}
relied in part on the mathematical and physical intuition gained in the
precursor \cite{SiedentopWeikard1986}, where the Scott correction is
proved in the absence of electron-electron repulsion. We will present
this motivating result and its short proof in
Subsection~\ref{sss:sweasyproof}.

\begin{remark}
  As has been observed, e.g., by Conlon \cite{Conlon1983A},
  Huxtable \cite{Huxtable1988}, or Sobolev \cite{Sobolev1996},
  the $Z^2$-correction is a consequence of the singularity of the
  Coulomb potential and cannot be explained semiclassically.
  For instance, Huxtable's result \cite[Theorem~7]{Huxtable1988} states
  \begin{align}
    \inf_{\substack{\psi\in\bigwedge_{\nu=1}^Z L^2(\R^3:\C^q)\\ \|\psi\|_2=1}}\left\langle\psi,\left[\sum_{\nu=1}^Z \left(-\frac12 \Delta_\nu - Z^{\frac43} W(Z^{\frac13}x_\nu)\right)\right]\psi\right\rangle
    = c_{\rm TF} Z^{\frac73} + \co(Z^{\frac53})
  \end{align}
  as $Z\to\infty$ for any potential $W\in C^\infty(\R^3)$ satisfying
  $c|x|^2\leq -W(x)\leq C|x|^2$ and $|\nabla W(x)|\leq c'|x|$ for some
  $c,C,c'>0$. Here $c_{\rm TF}$ is related to the non-relativistic
  Thomas--Fermi theory with potential $W$.
\end{remark}

\subsubsection{Strong Scott conjecture}
The Scott conjecture has a close relative that concerns the
ground state density, the so-called ``strong form of the Scott
correction'' (in short ``strong Scott conjecture'' from now on).
It was formulated by Lieb \cite[pp.~623--624]{Lieb1981}; see also
Heilmann and Lieb \cite[p.~3629]{HeilmannLieb1995}.
The strong Scott conjecture states that the suitably rescaled
ground state density $\rho_S$ on distances of order $Z^{-1}$
from the nucleus converges to $q$ times the three-dimensional
\emph{hydrogenic density}, i.e.,
\begin{align}
  \label{eq:defrhohs}
  \rho_S^H(x) := q \cdot \sum_{\ell=0}^\infty \sum_{m=-\ell}^\ell \sum_{n=0}^\infty|\psi_{n,\ell,m}^S(x)|^2,
  \quad x\in\R^3.
\end{align}
The latter is the sum of squares of the $L^2(\R^3:\C)$-normalized
eigenfunctions $\psi_{n,\ell,m}^S$ of the hydrogen Hamiltonian
\begin{align}
  \label{eq:defhydrogenoperator2}
  S^H = -\frac12\Delta-\frac{1}{|x|}
  \quad \text{in}\ L^2(\R^3:\C).
\end{align}
The hydrogenic density $\rho_S^H$ is rather well understood;
see Theorem~\ref{heilmannlieb}. In particular, the right side of
\eqref{eq:defrhohs} converges and is spherically symmetric. (The
labeling of the eigenfunctions $\psi_{n,\ell,m}^S$ uses the
decomposition into angular momentum channels and will be further
explained in Subsection \ref{sec:scottnonrel}.)

Note that $S^H$ is unitarily equivalent to $Z^{-2}S_Z^H$ by scaling
$x\mapsto x/Z$, where $S_Z^H$ is defined in \eqref{eq:defhydrogenoperator}.
Any eigenfunction $\phi_Z$ of $S_Z^H$ scales like $\phi_Z(x)=Z^{3/2}\phi_1(Zx)$,
where $\phi_1$ denotes the corresponding eigenfunction of $S_1^H=S^H$.

The 1996 work of Iantchenko et al.~\cite{Iantchenkoetal1996} showed,
among other things,
\begin{align}
  \label{eq:strongscottinitial}
  \lim_{Z\to\infty} \frac1{4\pi} \int_{\bs^2} Z^{-3}\rho_S(Z^{-1}r\omega)\,\rd\omega
  = \rho_S^H(r) \quad \text{for each}\ r>0,
\end{align}
see also Theorem \ref{ilsdensitynonrel}.
It follows from the convergence results there that the one-particle ground
state density $\rho_S$ is approximately spherically symmetric in the limit
$Z\to\infty$ on distances $Z^{-1}$ from the nucleus.

\subsection{Dirac--Schwinger correction}
\label{ss:ds}

As mentioned in the previous subsection, Scott anticipated that
the subleading terms in the expansion of the ground state energy $E_S(Z)$
should take into account the extreme quantum effects close to the nucleus,
but also the \emph{exchange energy of the electrons},
as proposed by Dirac \cite{Dirac1930}.

On a formal level, Schwinger \cite{Schwinger1981}, as well as
Englert and Schwinger
\cite{EnglertSchwinger1984StatisticalAtom:H,EnglertSchwinger1984StatisticalAtom:S,EnglertSchwinger1985A}
(see also Englert \cite{Englert1988} for a textbook treatment)
derived the third term in the asymptotic expansion of the ground
state energy, which grows like $Z^{5/3}$. In fact, this term is not
only generated by the exchange energy of the electrons,
but is also due to the semiclassical asymptotics of the eigenvalue sum of
the operator
$-\frac12\Delta-\Phi^{\rm TF}$ with the semiclassical parameter $Z^{-1/3}$ and
the $Z$-dependent Thomas--Fermi potential $\Phi^{\rm TF}$ (see \eqref{eq:tfeq3}).
In view of the results in Subsection \ref{ss:glimpsetf}, the
occurrence of $-\frac12\Delta-\Phi^{\rm TF}$ in the analysis of $E_S(Z)$
is not unexpected. 

A decade later, Schwinger's and Englert's derivation was made
mathematically rigorous in the monumental work of Fefferman and Seco
\cite{FeffermanSeco1990O,FeffermanSeco1992,FeffermanSeco1993,FeffermanSeco1994,FeffermanSeco1994T,FeffermanSeco1994Th,FeffermanSeco1995};
see also Bach \cite{Bach1992,Bach1993} and Graf and Solovej
\cite{GrafSolovej1994} for simplifications and improvements
of parts of Fefferman's and Seco's arguments and see \cite{Fefferman1992}
for a review of Fefferman's and Seco's proof.
They proved the existence of a constant $C_{\rm DS}>0$, which can be
computed in terms of the Thomas--Fermi density
$\rho_{Z=1}^{\rm TF}$, see \cite[p.~528]{FeffermanSeco1990O},
such that
\begin{align}
  E_S(Z) = E_1^{\rm TF} Z^{7/3} + \frac{q}{4}Z^2 - C_{\rm DS}Z^{5/3} + o(Z^{5/3}).  
\end{align}

In \cite[pp.~6, 9--10]{Fefferman1992} and
\cite[pp.~13--14]{FeffermanSeco1994} Fefferman and Seco make a
conjecture concerning a fourth, possibly oscillating term in the expansion
of $E_S(Z)$.
C\'ordoba et al.~\cite{Cordobaetal1994,Cordobaetal1995,Cordobaetal1996}
analyzed this term in detail and showed, in particular,
that it is bounded from below and above by constants times $Z^{3/2}$.

\subsection{The necessity of a relativistic description}

From a physical point of view it is questionable whether one can describe atoms with
large nuclear charges non-re\-la\-ti\-vis\-ti\-cal\-ly since already the bulk
of the electrons is localized in orbitals whose distance to the nucleus
is roughly $Z^{-1/3}$ or less. As $Z$ increases, the electrons become localized
closer to the nucleus and, by Heisenberg's uncertainty principle, one expects
that at least the velocities of the innermost electrons are a substantial fraction
of the speed of light $c$. In fact, the non-relativistic energy for electrons
on the length scale $Z^{-1}$ in the field of a nucleus of charge $Z$ is already
$-Z^2/2$. Thus, the virial theorem implies that its kinetic energy is $Z^2/2$.
In classical mechanics, this would show that the velocity of the electron is $Z$.
Since the velocity of light is 137 in our units, a single electron in the field
of a uranium nucleus ($Z=92$) would therefore move with a speed of
$\approx\frac{92}{137}\cdot c$, which indeed is a substantial fraction of the
speed of light.
For this reason, a relativistic description is mandatory, in particular on
the short length scale $Z^{-1}$.
Meanwhile, at distances $Z^{-1/3}$, electrons are expected to move
with velocities $\lesssim 10\%$ of the speed of light and, indeed, as we
will see momentarily, relativistic effects are negligible to the leading,
i.e., Thomas--Fermi order of $E_S(Z)$. 

Before turning to details, we point out that already Scott \cite[p.~866]{Scott1952}
anticipated possible shortcomings of his non-relativistic formula for large $Z$:
\begin{quote}
  ``\emph{Relativity effects of all kinds have been disregarded so far.
    Though this simplification has no serious consequences for $Z<30$,
    these effects are quite important for heavy elements. It would be a
    difficult task to calculate them accurately. A straightforward
    extension of Thomas' statistical method (Vallarta and Rosen
    \cite{VallartaRosen1932}) is inapplicable to our present problem, because
    most of the correction originates in the region close to the nucleus where
    the statistical method is vitiated by the boundary effect, and, in fact,
    such methods would give an infinite binding energy. Moreover, the
    interaction between the electrons is not wholly electrostatic.}''
\end{quote}
Concerning an extension of ``\emph{Thomas' statistical method}'' we refer to
Subsection \ref{ss:tfwrel} for recent developments in this direction.

\medskip
From a fundamental physical point of view, heavy atoms and molecules should be
treated by relativistic quantum electrodynamics (QED) and the corresponding field
theory. Unfortunately, many fundamental mathematical elements, e.g., the state
space and the Hamiltonian, lack mathematical understanding. As a consequence, one
is thrown back to approximate models. Here we will review three such approximate
Hamiltonian models that have been derived by physical arguments from QED, and
proven useful in applications. Moreover, we consider a mathematical
simplification thereof and a density functional obtained in this vein.

The first model can be traced back at least to
Chandrasekhar~\cite{Chandrasekhar1931} in the context of stability of
neutron stars. In it the single-particle kinetic energy $-\Delta/2$ is replaced
by $\sqrt{-c^2\Delta+c^4}-c^2$ with $c$ being the velocity of light.
Despite its mathematical simplicity, the resulting operator features
many physical defects, such as the violation of the principle of locality.
More crucially for us, it leads to ground state energies, that are much
too low compared to experimental data and can only be applied to atoms
with nuclear charge $Z<88$.

Physically and chemically more accurate models
are based on projected Cou\-lomb--Dirac \cite{Dirac1928,Dirac1928II} operators,
such as the Brown--Ravenhall \cite{BrownRavenhall1951} or the Furry
\cite{FurryOppenheimer1934} operator, which are applicable to atoms with
nuclear charge $Z<125$ and $Z<138$, respectively.
The latter is used in quantum chemistry to compute the ground state
energy of large atoms or molecules to chemical accuracy, see, e.g.,
Reiher and Wolf \cite{ReiherWolf2009}. 

A common property of relativistic operators is the fact that, at least
for large momenta, the kinetic energy scales like the Coulomb potential.
On a heuristic level, it is clear that the sole limit $Z\to\infty$
is meaningless since the potential energy cannot be controlled by the kinetic
energy anymore. Consequently, the total energy will not be bounded from
below, and the atom becomes ``unstable'' for fractions $Z/c$ beyond a critical
model-dependent coupling constant.
To make mathematically meaningful statements about asymptotics,
one considers the limit when both $Z$ and $c$ tend to infinity
simultaneously with a fixed ratio $Z/c=:\gamma$. (Of course, like the limit
$Z\to\infty$, the limit $c\to\infty$ is questionable since $c$ has a fixed
value). The idea to introduce $\gamma$ as a separate parameter
goes back at least to Schwinger \cite{Schwinger1980}.

%
%
%
%
For $\gamma\leq2/\pi$, S\o rensen \cite{Sorensen2005} proved that in the
above-described limit, the leading order of the ground state energy in the
Chandrasekhar model is given by the Thomas--Fermi energy.
Moreover, the ground state density on the
Thomas--Fermi length scale converges weakly and in the so-called Coulomb norm
(see \eqref{eq:defcoulombnorm}) to the hydrogenic Thomas--Fermi density
\cite{Merz2019D,MerzSiedentop2019}.
This indicates that the bulk of the electrons on the length scale $Z^{-1/3}$
does not behave relativistically.

%
%
%
On the other hand, electrons on the hydrogenic length scale $Z^{-1}$ are
located much closer to the nucleus and, as described before, are
expected to  lead to relativistic corrections of the Scott correction.
In fact, Schwinger \cite{Schwinger1980} derived a relativistic
$Z^2$-correction, which is lower than Scott's. This lowering was proved,
by two 
different approaches, in
Solovej et al.~\cite{Solovejetal2008} and in
\cite{Franketal2008}.
Later, a relativistic correction of the Scott correction was also proved
for the Brown--Ravenhall \cite{Franketal2009} and the Furry operator
\cite{HandrekSiedentop2015}.

The relativistic generalization of the strong Scott conjecture was
proved recently in \cite{Franketal2020P} (see
also \cite{Franketal2020R}), i.e., the convergence of the suitably rescaled
one-particle ground state density of Chandrasekhar atoms on the
hydrogenic length scale $1/Z$ to the sum of the squares of the
eigenfunctions of the one-particle Chandrasekhar operator.  Shortly
thereafter, the corresponding statement for the physically and chemically
accurate Furry operator was proved \cite{MerzSiedentop2020}.
These results underscores the fact that electrons close to the nucleus behave
relativistically and that self-interactions of the innermost electrons are
negligible.

\subsection{Organization}
\label{ss:organization}

We briefly summarize the contents of the present review.

%
%
In Section \ref{s:dft} we review three examples of density
functional theories.
For the first two, we refer to March's and Lieb's reviews \cite{March1957T,Lieb1981};
see also \cite{Siedentop2020M} for a recent review.
First, and most important for us, we discuss Thomas--Fermi theory.
%
%
Secondly, we review Weizs\"acker's extension of Thomas--Fermi theory, which
is physically and mathematically richer than Thomas--Fermi theory.
Qualitatively, this extension correctly accounts for quantum effects of
electrons close to the nucleus.
%
%
Thirdly, we investigate the Hellmann--Weizs\"acker functional, which
served as the basis for Siedentop's and Weikard's proof of the
Scott correction.
%
%
Finally, we consider a density functional that reduces to the
Thomas--Fermi--Weizs\"acker functional in the non-relativistic limit.
It was derived by Engel and Dreizler and was recently investigated from a
mathematical point of view.

%
%
\medskip
In Section \ref{s:nonrelativistic} we consider
non-relativistic atoms, ions, and molecules, both in the presence and
absence of a self-generated field, and summarize theorems concerning the
energy asymptotics and the convergence of the quantum density on both
the Thomas--Fermi and the Scott length scales.
Emphasis will be put on Scott's original derivation of the energy
correction, as well as
the initial motivating results in \cite{SiedentopWeikard1986}.

%
%
\medskip
Section \ref{s:relativistic} is concerned with
relativistic descriptions. We summarize results concerning
the energetic asymptotics as well as the convergence of the density for
all the three different relativistic models discussed in the introduction
-- the Chandrasekhar, the Brown--Ravenhall, and the Furry model.

\medskip
In Section \ref{s:questions} we discuss some open questions.

\subsection*{Notation}
We write $A\lesssim B$ for two non-negative quantities $A,B\geq 0$ to
indicate that there is a constant $C>0$ such that $A\leq C B$.
If $C=C_\tau$ depends on a parameter $\tau$, we sometimes write $A\lesssim_\tau B$.
The notation $A\sim B$ means $A\lesssim B\lesssim A$.
The indicator function of a set $\Omega$ is denoted by
$\one_\Omega$.
The negative part of a real number or a self-adjoint operator $A$
is defined by $A_-:=\max\{0,-A\}\geq0$.

\section{Density functional theories}
\label{s:dft}


In this section, we briefly review three examples of effective theories
that are known to describe correctly at least the leading order of the
ground state energy of large atoms and molecules.
They are known as \emph{density functionals}, i.e., energy functionals
that only depend on the one-particle density of a given many-particle
system.  We refer to Lieb's detailed review \cite{Lieb1981} on
Thomas--Fermi-type theories and to \cite{Siedentop2020M} for a recent
survey of density (matrix) functional theories.  In particular,
\cite{Lieb1981} also treats extensions of Thomas--Fermi theory like
Weizs\"acker's inhomogeneity \cite{Weizsacker1935} and Dirac's exchange
\cite{Dirac1930} correction.  The first one will be of some interest
for us since it generates a Scott correction, whereas the second one
will be discussed in passing only.

\subsection{Thomas--Fermi theory}
\label{ss:tf}

We begin with the simplest non-relativistic ``statistical model of the atom''
(Fermi \cite{Fermi1927,Fermi1928}, Gomb{\'a}s \cite{Gombas1949}), which was
formulated in the late 1920's independently by Thomas \cite{Thomas1927} and Fermi
\cite{Fermi1927,Fermi1928}.
In the molecular case with $K$ nuclei of charges
$\uZ=(Z_1,...,Z_K)\in(0,\infty)^{K}$ situated at
positions $\uR=(R_1,...,R_K)\in\R^{3K}$, the so-called
Thomas--Fermi (TF from now on) functional (Lenz \cite{Lenz1932}) is given by
\begin{align}
  \label{eq:deftf}
  \ce_V^{\rm TF}(\rho)
  := \int_{\br^3}\left(\frac{3}{5}\gamma_{\mathrm{TF}}\rho^{5/3}(x)-V(x)\rho(x)\right)\dx
  + D(\rho,\rho) + U
\end{align}
with $V$ and $U$ as in \eqref{eq:defv} and \eqref{eq:defu}, respectively.
In the atomic case $K=1$, we write
$\ce_Z^{\rm TF}\equiv \ce_{Z/|x|}^{\rm TF}$.

The first term of $\ce_V^{\rm TF}(\rho)$ represents the kinetic
energy and is
derived via the following argument based on a semiclassical phase space integration.
The TF model views the $N$ non-relativistic quantum particles in a potential
$W$ as a classical gas in phase space.
Since Planck's constant is $h=2\pi\hbar=2\pi$ in our units, the density of the
semiclassical gas is
\begin{align}
  \label{eq:semiclassicalderivation1}
  \rho(x)
  = q\int \one_{\{p^2/2\leq W(x)\}}\,\frac{\rd p}{(2\pi)^3}
  = \frac{q\cdot4\pi}{3\cdot(2\pi)^3}(2W(x))^{3/2}.
\end{align}
Thus, the semiclassical kinetic energy is
\begin{align}
  \label{eq:semiclassicalderivation2}
  q\int \frac{p^2}{2}\one_{\{p^2/2\leq W(x)\}}\frac{\rd p}{(2\pi)^3}
  = \frac{q}{(2\pi)^3}\cdot\frac12\cdot\frac{4\pi}{5}\cdot(2W(x))^{5/2}
  = \frac35 \gtf \cdot \rho(x)^{5/3}
\end{align}
with the \emph{Thomas--Fermi constant}
\begin{align}
  \label{eq:semiclassicalderivation3}
  \gtf := (6\pi^2)^{2/3}\hbar^2(2mq^{2/3})^{-1} = (6\pi^2/q)^{2/3}/2.  
\end{align}

The second term in \eqref{eq:deftf} represents the interaction energy
between the electrons and the nuclei.

The third term in \eqref{eq:deftf} is the electrostatic self-energy of
the charge density $\rho$. It is defined, more generally, for $\rho$
and $\sigma$ by
\begin{align}
  \label{eq:defD}
  D(\rho,\sigma) = \frac12\int_{\br^3}\int_{\br^3}\frac{\overline{\rho(x)}\sigma(y)}{|x-y|}\,\dx\,\dy.
\end{align}
Note that (by Plancherel and the convolution theorem), $D(\rho,\sigma)$
is sesquilinear and positive, and (by Cauchy--Schwarz),
$D(\rho,\sigma)\leq \sqrt{D(\rho,\rho)}\cdot \sqrt{D(\sigma,\sigma)}$.
In fact, $D(\cdot,\cdot)$ defines a scalar product on the set $\ci$
defined in \eqref{eq:defseti} below.
Thus, the right side of
\begin{align}
  \label{eq:defcoulombnorm}
  \|\rho\|_C := D(\rho,\rho)^{1/2}  
\end{align}
defines a norm on that space. This norm is sometimes called \emph{Coulomb norm}.

The TF functional \eqref{eq:deftf} is defined on its natural domain
(Simon \cite{Simon1979})
\begin{align}
  \label{eq:defseti}
  \ci = \{\rho\in L^{5/3}(\br^3):\ D(\rho,\rho)<\infty,\ \rho\geq0\},
\end{align}
i.e., for nonnegative densities with finite kinetic energy and finite
electron-electron repulsion. These conditions automatically guarantee the
finiteness of the electron-nucleus-interaction. (The local singularities at
the nuclei are controlled by the kinetic energy, whereas the long-range
part is controlled by the electron-electron repulsion.)

To describe a system of $N$ electrons, we restrict the TF functional to the set
\begin{align}
  \ci_N & = \{ \rho\in\ci:\ \int\rho = N \} \notag \\
        & = \{\rho\in L^{5/3}(\br^3):\ D(\rho,\rho)<\infty,\ \rho\geq0,\ \int\rho=N\}.
\end{align}
Here, mathematically speaking, $N$ need not be an integer.

In their seminal work \cite{LiebSimon1977}, Lieb and Simon were the first to
analyze this functional with mathematical rigor.
(See also \cite[Section~II]{Lieb1981} and \cite[Subsection~4.1]{Siedentop2020M}
for more detailed reviews.)
The following theorem asserts the existence and uniqueness of minimizers
of the TF functionals on $\ci$ and $\ci_N$.

\begin{theorem}
  \label{tftheorem}
  Let $\uZ=(Z_1,...,Z_K)\in(0,\infty)^K$ and $\uR=(R_1,...,R_K)\in\R^{3K}$.
  Then the following statements hold.
  \begin{enumerate}
  \item (Unconstrained problem):
    There exists a unique $0\leq\rho^{\rm TF}(\uZ,\uR,x)$ such that
    $\int_{\R^3}\rho^{\rm TF}(\uZ,\uR,x)\,\dx=|\uZ|=\sum_{\kappa=1}^K Z_\kappa$
    and
    \begin{align}
      \label{eq:tfunconstrained}
      \ce_V^{\rm{TF}}[\rho^{\rm TF}] = \inf\{\ce_V^{\rm TF}[\rho]: \rho\in\ci\}
      \equiv E^{\rm TF}(\uZ,\uR).
    \end{align}

  \item (Constrained problem): If $N\leq|\uZ|$, then there exists a unique,
    non-negative $\rho^{\rm TF}(N,\uZ,\uR,x)$ such that
    $\int_{\R^3}\rho^{\rm TF}(N,\uZ,\uR,x)\,\dx=N$ and 
    \begin{align}
      \label{eq:tfconstrained}
      \ce_V^{\rm{TF}}[\rho^{\rm TF}]
      = \inf\{\ce_V^{\rm TF}[\rho]: \rho\in\ci_N\}
      \equiv E^{\rm TF}(N,\uZ,\uR).
    \end{align}
    In particular, $N\mapsto E^{\rm TF}(N,\uZ,\uR)$ is
    strictly decreasing.
    
    If $N>\sum_{\kappa=1}^K Z_\kappa$, then
    $E^{\rm TF}(N,\uZ,\uR)$ is not a minimum, i.e., there are
    no negatively charged ions in TF theory.

  \item (Unconstrained Thomas--Fermi equation) In the
    unconstrained problem, the minimizer $\rho^{\rm TF}\in\ci$ obeys
    $\int\rho^{\rm TF}=|\uZ|$ and
    \begin{align}
      \label{eq:tfeq}
      \gtf(\rho^{\rm TF})^{\frac23}
      = V - \left(\rho^{\rm TF}(\uZ,\uR,\cdot)\ast\frac{1}{|\cdot|}\right).
    \end{align}
    Moreover, if $\rho\in\ci$ satisfies
    \eqref{eq:tfeq}, then it minimizes $\ce_V^{\rm TF}$ on $\ci$.
    If $K=1$, then $\rho^{\rm TF}$ is spherically symmetric and decreasing.

  \item (Thomas--Fermi equation in constrained problem) In the
    constrained problem with $0<N\leq|\uZ|$,
    the minimizer $\rho^{\rm TF}\in\ci_N$ satisfies
    \begin{align}
      \label{eq:tfeq2}
      \gtf\rho^{\rm TF}(x)^{\frac23}
      = \left(V(x)-\left(\rho^{\rm TF}(N,\uZ,\uR,\cdot)\ast\frac{1}{|\cdot|}\right)(x)-\mu\right)_+
    \end{align}
    for some (unique) $\mu=\mu(N)\geq0$.
    Moreover, there is no solution $\rho\in\ci_N$ to
    \eqref{eq:tfeq2} for any $\mu$ other than $\rho^{\rm TF}$.

    When $N=Z$, then $\mu=0$, and otherwise
    $\mu>0$, i.e., $E^{\rm TF}(N,\uZ,\uR)$ is strictly decreasing
    in $N$. As $N$ varies from $0$ to $|\uZ|$, $\mu$ varies
    continuously from $\infty$ to $0$. Moreover, $\mu$ is a convex,
    decreasing function of $N$.

  \item (Scaling) For any $a>0$, the scaling relations
    \begin{subequations}
      \label{eq:scalingtf}
      \begin{align}
        \rho^{\rm TF}(N,\uZ,\uR,x)
        & = a^{-2}\rho^{\rm TF}(a N, a \uZ,a^{-1/3}\uR,a^{-1/3}x), \\
        E^{\rm TF}(N,\uZ,\uR)
        & = a^{-7/3} E^{\rm TF}(a N, a\uZ,a^{-1/3} \uR)
      \end{align}
    \end{subequations}
    hold.
  \end{enumerate}
\end{theorem}

\begin{remarks}
  (1) The number $E^{\rm TF}(N,\uZ,\uR)$ is called the
  \emph{Thomas--Fermi energy} and the minimizer $\rho^{\rm TF}$
  is called the \emph{Thomas--Fermi density}.

  (2) Although $E^{\rm TF}(N,\uZ,\uR)$ is not a minimum
  on $\ci_N$ and \eqref{eq:tfeq2} has no solution with
  $\int\rho=N$ if $N>\sum_{\kappa=1}^K Z_\kappa$, the number
  $E^{\rm TF}(N,\uZ,\uR)$ still exists and we have
  $E^{\rm TF}(N,\uZ,\uR)=E^{\rm TF}(|\uZ|,\uZ,\uR)$ in that
  case.

  (3) If $K=1$ in the unconstrained problem, then we write
  $\rho^{\rm TF}(Z,0,x)\equiv\rho_Z^{\rm TF}(x)$
  and $E^{\rm TF}(Z,0)\equiv E^{\rm TF}(Z)$.
  Similarly, in the constrained problem we shall write
  $\rho^{\rm TF}(N,Z,0,x)\equiv\rho_Z^{\rm TF}(N,x)$
  and $E^{\rm TF}(N,Z,0)\equiv E^{\rm TF}(N,Z)$.

  (4) The scaling relations for $K=1$ show, in particular, that
  TF theory has ``natural'' length and energy scales, the
  \emph{Thomas--Fermi length scale} $Z^{-1/3}$ and the
  \emph{Thomas--Fermi energy scale} $Z^{7/3}$, respectively.
  For $Z=1$ the minimizer $\rho^{\mathrm{TF}}$ (either in the
  unconstrained or constrained problem) is called \emph{hydrogenic
    Thomas--Fermi density}.
  The numerical value of the associated infimum is
  $E^{\rm TF}(Z=1)\approx -3.678\,74\cdot\gtf^{-1}$,
  cf.~Gomb{\'a}s \cite[p.~60]{Gombas1949}.
\end{remarks}

The following theorem due to Lieb and Simon
\cite[Theorem IV.5]{LiebSimon1977} (see also
\cite[Theorem 2.8]{Lieb1981}) summarizes some important properties
of the TF density. 

\begin{theorem}[Properties of $\rho^{\rm TF}$]
  \label{tfdensityproperties}
  Let $\uZ=(Z_1,...,Z_K)\in(0,\infty)^K$ and $\uR=(R_1,...,R_K)\in\R^{3K}$
  and let $\rho^{\rm TF}$ denote the solution to the constrained
  Thomas--Fermi equation \eqref{eq:tfeq2}
  with $\int\rho^{\rm TF}(N,\uZ,\uR,x)=N$.
  Then the following statements hold:
  \begin{enumerate}
  \item Let $\kappa\in\{1,...,K\}$ be arbitrary. Then, as $x\to R_\kappa$, one has
    \begin{align}
      \rho^{\rm TF}(N,\uZ,\uR,x)
      = \left(\frac{Z_\kappa}{\gtf}\right)^{\frac32}|x-R_\kappa|^{-3/2} + o(|x-R_\kappa|^{-1/2}).
    \end{align}
  \item $\rho^{\rm TF}(N,\uZ,\uR,x)\to0$ as $|x|\to\infty$.
  \item $\rho^{\rm TF}$ is real analytic on
    $\{x\in\R^3:\,x\neq R_\kappa\ \forall \kappa,\rho^{\rm TF}(x)>0\}$.
  \item In the neutral case ($N=|\uZ|$, $\mu=0$), one has
    $\rho^{\rm TF}(x)>0$ for all $x\in\R^3$.
  \item In the ionic case ($N<|\uZ|$, $\mu>0$), $\rho^{\rm TF}$
    is compactly supported and $C^1$ away from the $R_\kappa$.

  \item (Thomas--Fermi equation) Let the \emph{Thomas--Fermi potential}
    $\Phi^{\rm TF}$ be defined by
    \begin{align}
      \label{eq:tfeq3}
      \Phi^{\rm TF}(N,\uZ,\uR,x) := V(x) - (\rho^{\rm TF}(N,\uZ,\uR,\cdot)\ast|\cdot|^{-1})(x).
    \end{align}
    Then $\Phi^{\rm TF}$ obeys the \emph{Thomas--Fermi differential
      equation}
    \begin{align}
      \label{eq:tfeq4}
      -\frac{1}{4\pi}(\Delta\Phi^{\rm TF})(N,\uZ,\uR,x)
      = \sum_{\kappa=1}^K Z_\kappa\delta(x-R_\kappa) - \gtf^{-3/2}(\Phi^{\rm TF}-\mu)_+^{3/2}.
    \end{align}

  \item (Sommerfeld)
    In the neutral case ($\mu=0$) the \emph{Sommerfeld solution}
    \begin{align}
      \label{eq:sommerfeld}
      \psi(x) = \psi(|x|) = \gtf^3\cdot (3/\pi)^2\cdot |x|^{-4}
    \end{align}
    solves the TF differential equation \eqref{eq:tfeq4} for $|x|>0$
    and $x\neq R_\kappa$ and it is the only power law that does so.
    Moreover,
    \begin{align}
      \lim_{s\to\infty}\frac{\max_{|x|=s}\Phi^{\rm TF}(N,\uZ,\uR,x)}{\psi(s)}
      = \lim_{s\to\infty}\frac{\min_{|x|=s}\Phi^{\rm TF}(N,\uZ,\uR,x)}{\psi(s)}
      = 1.
    \end{align}
    In the atomic case ($K=1$), the TF density $\rho^{\rm TF}$ obeys
    \begin{align}
      \rho^{\rm TF}(N,Z,x)
      = \left(\frac{3\gtf}{\pi}\right)^3 |x|^{-6} + o(|x|^{-6})
    \end{align}
    as $|x|\to\infty$.
  \end{enumerate}
\end{theorem}

\begin{remark}
  Observe that Sommerfeld's solution \eqref{eq:sommerfeld} (see
  \cite{Sommerfeld1932}) is \emph{independent of $Z$} and even solves
  the molecular TF equation; due to the scaling of the TF density, the
  Sommerfeld asymptotics $|x|^{-6}$ of the TF density still has
  magnitude $\mathcal{O}(Z^2)$ for $|x|\lesssim Z^{-1/3}$.  Sommerfeld
  type estimates are contained in Solovej's proof of the ionization
  conjecture in Hartree--Fock theory \cite[Theorems~4.6, 5.2,
  5.4]{Solovej2003}; see also \cite[Theorem~2.10]{Lieb1981}
  and \cite[Section~1]{IvriiSigal1993} for further estimates for the
  TF density and potential.
\end{remark}

In absence of electron repulsion, the Thomas--Fermi energy can be
computed easily. This is important for the heuristic derivation of
the Scott correction (Subsection \ref{sss:scottliebarguments}) and is
done in the following remark.

\begin{remark}[Thomas--Fermi energy for the Bohr atom]
  \label{bohrtfatom}
  Let $K=1$ and $Z>0$, and consider TF theory for an atom in absence
  of electron repulsion, i.e.,
  \begin{align}
    \ce_{Z,{\rm Bohr}}^{\rm TF}(\rho)
    = \int_{\R^3}\left(\frac{3}{5}\gtf\rho(x)^{5/3}-\frac{Z}{|x|}\rho(x)\right)\dx.
  \end{align}
  For $N>0$, let
  $$
  E^{\rm TF}_{\rm Bohr}(N,Z)
  = \inf\left\{\ce_{Z,{\rm Bohr}}^{\rm TF}(\rho) :\, 0\leq \rho\in L^{5/3}(\R^3),\, \int \rho = N \right\}.
  $$
  It is elementary to see that there is a unique minimizer
  $\rho_{{\rm Bohr}}^{\rm TF}(N,Z)$ and that this minimizer
  satisfies the Euler--Lagrange equation
  \begin{align}
    \gtf\cdot\rho_{{\rm Bohr}}^{\rm TF}(N,Z,x)^{2/3}
    = \left(\frac{Z}{|x|}-\mu\right)_+
  \end{align}
  with some $\mu>0$.
  Integrating the $3/2$-th power of this identity leads to the relation
  \begin{align*}
    \mu= \left( \frac{\pi^2}{4} \right)^\frac23 \frac{1}{\gtf} \frac{Z^2}{N^\frac23},
  \end{align*}
  and then to the formula for the energy
  \begin{align}
    \label{eq:tfenergyneutralbohr}
    E^{\rm TF}_{\rm Bohr}(N,Z) = \ce_{Z,{\rm Bohr}}^{\rm TF}(\rho_{{\rm Bohr}}^{\rm TF}(N,Z)) = - \frac{3}{\gtf} \left( \frac{\pi^2}{4} \right)^\frac23 Z^2 N^\frac13.
  \end{align}
\end{remark}

\subsection{Thomas--Fermi--Weizs\"acker theory}
\label{ss:tfw}

The semiclassical derivation of TF theory assumes that the density
is locally constant. In this regard, we recall Scott's observations
\cite[p.~859, p.~867]{Scott1952}:

\begin{quote}
  ``\emph{The Thomas--Fermi statistical model of the atom leads to the
    formula $20\cdot 92\, Z^{7/3}\,\mathrm{ev}$ for the total binding energy
    of an atom with atomic number $Z$, but this formula gives values
    which are too high by roughly $20\%$. The actual binding energies
    increase quite smoothly with increasing $Z$, which suggests the
    existence of a more appropriate formula. [...] The failure of the
    currently-quoted formula is due partly to the shortcomings of the
    statistical model in the region nearest the nucleus, and partly to
    the effect of exchange.}''
\end{quote}

\medskip
In 1935 Weizs\"acker \cite{Weizsacker1935} proposed a correction of
Thomas--Fermi theory that penalizes rapid changes of the density, which
are expected to occur close to the nucleus.

\begin{remark}
  Some words on the history:
  Weizs\"acker introduced this correction to explain the rise of the mass
  defect per nucleon in a nucleus from very heavy (say uranium) to
  semi-heavy nuclei (like iron). To that end he consulted Gamow's liquid
  drop model for nuclei and argued that, as a consequence of the
  uncertainty principle, the ``surface of the nucleus'' must be smeared out.
  For otherwise, an instantaneous drop of the density with infinite slope
  would lead to an infinite kinetic energy, which is unreasonable.
  This smearing of the surface could be accounted for by replacing the
  eigenfunctions that are used in the derivation of the Thomas--Fermi
  functional, namely plane waves, by waves with linearly varying amplitude.
  This gives rise to Weizs\"acker's term $\rho^{-1}(\nabla\rho)^2$.
\end{remark}

We consider the Thomas--Fermi--Weizs\"acker (TFW from now on) functional
\begin{align}
  \label{eq:deftfw}
  \ce_V^{\rm TFW}(\rho) := \frac{A}{2}\int_{\R^3} |\nabla\sqrt{\rho}|^2 + \ce_V^{\rm TF}(\rho)
\end{align}
with nuclei of charges $\uZ=(Z_1,...,Z_K)\in(0,\infty)^K$
situated at positions $\uR=(R_1,...,R_K)\in\R^{3K}$.
It is naturally defined on the set
\begin{align}
  \ca := \{\rho\in L_{\rm loc}^1(\R^3) :\, \rho\geq0,\, \nabla\sqrt\rho\in L^2(\R^3),\, \|\rho\|_C<\infty\},
\end{align}
where the gradient is understood in the sense of distributions.
For fixed particle number $N\in(0,\infty)$, the functional is defined on
\begin{align}
  \label{eq:defcalambda}
  \ca_N := \{\rho \in \ca:\, \int\rho=N\}.
\end{align}
Weizs\"acker introduced \eqref{eq:deftfw} with $A=1$. However, it is convenient
to have $A>0$ as an adjustable parameter, as we shall see soon.

The mathematical analysis of $\ce_V^{\rm TFW}$ started with the works
of Benguria \cite{Benguria1979} and Benguria et al.~\cite{Benguriaetal1981}.
Besides its mathematical richness, it turned out TFW theory describes --
at least qualitatively -- the physics of real atoms more closely than
Thomas--Fermi theory.
For instance, the TFW minimizer is finite at the nuclei and decays
exponentially at infinity. Moreover, binding is possible and anions
can be stable in TFW theory.  For a concise summary of TFW theory we
encourage the reader to consult \cite{Benguria1979,Benguriaetal1981},
as well as \cite[Sect.~VII]{Lieb1981}.  Here we restrict ourselves to
a summary of the energy expansion and the minimizing density as
$|\uZ|\to\infty$.  Our presentation follows closely Lieb
\cite{Lieb1981} and Lieb and Liberman \cite{LiebLiberman1982}.  We
start with the following result on existence and uniqueness of
minimizers of the TFW functional.

\begin{theorem} 
  \label{tfwtheorem}
  Let $A>0$, $\uZ=(Z_1,...,Z_K)\in(0,\infty)^K$, and
  $\uR=(R_1,...,R_K)\in\R^{3K}$. Then the following statements hold.

  \begin{enumerate}
  \item (Unconstrained problem) There is $N_c\in(|\uZ|,2|\uZ|)$
    such that the TFW functional $\ce_V^{\rm TFW}$ has a unique minimizer
    $\rho^{\rm TFW}(\uZ,\uR,x)$ on $\ca$ with particle number
    $\int\rho^{\rm TFW}(\uZ,\uR,x)\,\dx=N_c$.
    This minimizer satisfies the TFW equation
    \begin{align}
      \label{eq:tfweulerlagrange1}
      \left(-\frac{A}{2}\Delta+W\right)\sqrt{\rho^{\rm TFW}} = 0
    \end{align}
    with
    \begin{align}
      \label{eq:defw1}
      W(x) = \gtf\rho^{\rm TFW}(\uZ,\uR,x)^{2/3} - V(x) + \int_{\R^3}\frac{\rho^{\rm TFW}(\uZ,\uR,y)}{|x-y|}\,\dy.
    \end{align}
    The infimum is denoted by
    \begin{align}
      E^{\rm TFW}(\uZ,\uR) = \inf\{\ce_V^{\rm TFW}[\rho]:\,\rho\in\ca\}.
    \end{align}

  \item (Constrained problem) If $N\leq N_c$, then
    the TFW functional has a unique minimizer
    $\rho^{\rm TFW}(N,\uZ,\uR,x)$
    on $\ca_N$. This minimizer satisfies the TFW equation
    \begin{align}
      \label{eq:tfweulerlagrange2}
      \left(-\frac{A}{2}\Delta+W\right)\sqrt{\rho^{\rm TFW}} = -\mu\sqrt{\rho^{\rm TFW}}
    \end{align}
    with $W(x)$ as in \eqref{eq:defw1}, $\mu\geq0$, and $\mu=0$
    for $N=N_c$. The infimum is denoted by
    \begin{align}
      \label{eq:deftfwenergy}
      E^{\rm TFW}(N,\uZ,\uR) = \inf\{\ce_V^{\rm TFW}[\rho]:\,\rho\in\ca_N\}.
    \end{align}
    
    If $N>N_c$, there is no minimizer on $\ca_N$.
  \end{enumerate}
\end{theorem}

\begin{remark}
  Benguria and Lieb \cite{BenguriaLieb1985} proved the previously-mentioned
  ionization conjecture for TFW molecules and showed
  $0<N_c-|\uZ|\leq 270.74 \cdot (\frac{A}{2\gtf})^{3/2}\cdot K$.
  As we shall see below, the value $A=0.1859$ is in some sense natural.
  Together with the value of $\gtf=(6\pi^2/2)^{2/3}/2$ this leads to the bound
  $N_c-|\uZ|< 0.7335 \cdot K$.
\end{remark}

Theorem~\ref{tfwtheorem} shows, in particular, that anions can be stable in
TFW theory.
The next theorem says that $\rho^{\rm TFW}$ on the TF length
scale is described by $\rho^{\rm TF}$; see also Solovej \cite{Solovej1990}
for results when only some of the nuclear charges tend to infinity.

\begin{theorem}[{\cite[Theorem 7.30]{Lieb1981},\cite[(2.25)]{LiebLiberman1982}}]
  Let $A>0$, $\uZ=(Z_1,...,Z_K)\in(0,\infty)^K$,
  $\uR=(R_1,...,R_K)\in\R^{3K}$, and $N>0$ so that
  $\lambda:=N/|\uZ|$ is fixed. Define $\uz$ and $\ur$ by
  $\uZ=|\uZ|\uz$ and $\uR=|\uZ|^{-1/3}\ur$, respectively.
  Then
  \begin{align}
    \lim_{|\uZ|\to\infty}|\uZ|^{-2}\rho^{\rm TFW}(N,\uZ,\uR,|\uZ|^{-1/3}x)
    = \rho^{\rm TF}(\lambda,\uz,\ur,x)
  \end{align}
  weakly in $L^1$ if $\lambda\leq|\uZ|$ and weakly in
  $L_{\rm loc}^1$ if $\lambda>|\uZ|$.
\end{theorem}

Naturally, the question arises of how close the two infima
$E^{\rm TFW}(N,\uZ,\uR)$ and $E^{\rm TF}(N,\uZ,\uR)$ are, i.e.,
one seeks an upper bound on the right side of
\begin{align}
  0\leq E^{\rm TFW}(N,\uZ,\uR) - E^{\rm TF}(N,\uZ,\uR).
\end{align}
Already in the neutral, atomic case ($K=1$, $N=Z$) one might be tempted
to say that the difference is $\mathcal{O}(Z^{5/3})$ by plugging in
the TF density $\rho_Z^{\rm TF}$ and using the scaling relation
$\rho_Z^{\rm TF}(x)=Z^2\rho_1^{\rm TF}(Z^{1/3}x)$.
However, this is not correct, as can be seen heuristically as follows.
By Theorem \ref{tfdensityproperties}, one has
$\rho_1^{\rm TF}(x)\sim \const |x|^{-3/2}$ as $|x|\to0$, which makes it
plausible that $|\nabla\sqrt{\rho_1^{\rm TF}}| \sim \const |x|^{-7/4}$,
but this is not square-integrable and leads to an infinite Weizs\"acker
term. Instead, as the following theorem shows, the difference
$E^{\rm TFW}(Z,Z,0) - E^{\rm TF}(Z,Z,0)$ is given, to leading order as
$Z\to\infty$, by a constant times $Z^2$. This is the Scott correction
in TFW theory.
As in the quantum problem, the $Z^2$-term originates
from effects on the hydrogenic length scale $Z^{-1}$ rather than the
TF length scale $Z^{-1/3}$, see \cite[p.~635]{Lieb1981}.
Lieb \cite[Theorem~7.30]{Lieb1981} also shows that the correction
is independent of the electron number, i.e., it also holds when comparing
the TF and TFW energies for ions with fixed ratio $N/Z$.

Let us return to the general, multi-center case.  To describe TFW
theory on the length scale $Z^{-1}$ more precisely, we consider the
\emph{atomic TFW functional without electron repulsion}.  After a
`renormalization' (that is, formally subtracting the integral of
$$
\frac25 \gtf \left( \frac{Z}{\gtf |x|} \right)^{5/2}
$$
from the right side of \eqref{eq:deftfw}) one can show that the resulting
functional has a unique minimizer and that this minimizer solves the
Euler--Lagrange equation
\begin{align}
  \label{eq:tfweulerlagrange}
  \left(-\frac{A}{2}\Delta+\gtf \rho^{2/3}- \frac Z{|x|} \right)\sqrt{\rho}
  = 0;
\end{align}
see \cite[Theorem 7.29]{Lieb1981}. By scaling, one has
$\rho(x) = (2Z^2/(A\gtf))^{3/2} \rho_\infty(2Zx/A)$, where $\rho_\infty$
is the solution corresponding to $Z=A/2=\gtf$.

Then we have the following results on the hydrogenic
energy and length scales.

\begin{theorem}[{\cite[Theorem 7.30]{Lieb1981}, \cite[(2.26)]{LiebLiberman1982}}]
  \label{tfwz2theorem}
  Let $A>0$, $\uZ=(Z_1,...,Z_K)\in(0,\infty)^K$,
  $\uR=(R_1,...,R_K)\in\R^{3K}$, and $N>0$ so that
  $\lambda:=N/|\uZ|$ is fixed.

  \begin{enumerate}
  \item (Energy) We have
    \begin{align}
      \label{eq:scotttfw}
      E^{\rm TFW}(N,\uZ,\uR) = E^{\rm TF}(N,\uZ,\uR) + D^{\rm TFW}\sum_{\kappa=1}^K Z_\kappa^2 + o(|\uZ|^2)
    \end{align}
    with $D^{\rm TFW} :=2^{1/2}A^{1/2}\gtf^{-3/2}\cdot \int_{\R^3} |\nabla \sqrt{\rho_\infty}|^2\,\dx$.

  \item (Density) Define $\uz$ and $\ur$ by $\uZ=|\uZ|\uz$ and
    $\uR=|\uZ|^{-1/3}\ur$, respectively. Then the solution $\rho^{\rm TFW}$
    of the problem \eqref{eq:deftfwenergy} converges to that of
    \eqref{eq:tfweulerlagrange} in the sense that for each
    $\kappa\in\{1,...,K\}$, 
    \begin{align}
      \label{eq:convtfwdensitywithoutrepulsion}
      \lim_{N\to\infty}Z_\kappa^{-3}\rho^{\rm TFW}(N,\uZ,\uR,R_\kappa+Z_\kappa^{-1}x)
      = (A\gtf/2)^{-3/2}\rho_\infty(2x/A)
    \end{align}
    both pointwise and in $L_{\rm loc}^1$.
  \end{enumerate}
\end{theorem}

As emphasized by Lieb and Liberman in
\cite[Section~2.C]{LiebLiberman1982}, the second term in
\eqref{eq:scotttfw} has the following properties, which allows one to
think of it as a ``core effect'':
\begin{itemize}
\item It is independent of $\lambda = N/|\uZ|$.
\item It is additive in the nuclei, that is, it is a sum of terms
  corresponding to each atom in the molecule.
\item The constant $D^{\rm TFW}$ does not change if the electron-electron
  repulsion is removed.	
\end{itemize}
By the last point, we mean that \eqref{eq:scotttfw} remains true, with
the same constant $D^{\rm TFW}$, if in the definitions of both
$E^{\rm TFW}(N,\uZ,\uR)$ and $E^{\rm TF}(N,\uZ,\uR)$ the term $D(\rho,\rho)$
is dropped. This is proved in \cite{Lieb1981}.

\medskip
The asymptotics \eqref{eq:scotttfw} and the convergence in
\eqref{eq:convtfwdensitywithoutrepulsion} suggest a discussion
of the parameter $A$ in \eqref{eq:deftfw}.
By Theorem \ref{tfwz2theorem}, it suffices to discuss the atomic case $K=1$.
While Weizs\"acker initially chose $A$ to be one, other
values have been suggested.
For instance, Kirzhnits \cite{Kirzhnits1957} suggested $A\approx1/9$
based on the gradient expansion of the Hohenberg--Kohn functional, assuming
the Coulomb potential was replaced by a ``weak perturbing potential''
(\cite[p.~12]{LiebLiberman1982}).
However, due to the local singularity, the Coulomb potential cannot
be regarded as such a weak perturbing potential.

More then 15 years before Theorem \ref{tfwz2theorem} was proved,
Yonei and Tomishima \cite{YoneiTomishima1965} analyzed \eqref{eq:tfweulerlagrange}
with $\mu$ such that the solution $\rho_\infty$ obeys $\int\rho_\infty=Z$.
From a numerical analysis they concluded $A\approx1/5$ (especially when $Z>25$)
leads to good agreement with the energy obtained from summing up the first $Z$
eigenvalues of the hydrogen operator \eqref{eq:defhydrogenoperator} (Bohr atom,
cf.~Remark \ref{bohrtfatom} and Subsection \ref{sss:scottliebarguments}).

Possibly inspired by Yonei's and Tomishima's work, Lieb and Liberman
\cite[(2.32)]{LiebLiberman1982} chose $A$ such that the $Z^2$-correction
in the TFW model agrees with that of the quantum model, i.e.,
$D^{\rm TFW}=q/4$. This choice leads to $A=0.1859$.

Another choice for $A$ is motivated by comparing the densities
$\rho_Z^{\rm TFW}$ and the one-particle ground state density $\rho_S$
in \eqref{eq:defrhohnonrel} on the length scale $Z^{-1}$.  As
indicated in \eqref{eq:strongscottinitial}, the spherical average over
$\rho_S$ tends to the hydrogenic density $\rho_S^H$
(cf.~\eqref{eq:defrhohs}) on the length scale $Z^{-1}$ \emph{pointwise} as
$Z\to\infty$. Recall that all hydrogenic eigenfunctions of
\eqref{eq:defhydrogenoperator2} are finite at the origin with only
eigenfunctions with $\ell=0$ being non-zero. (For a detailed analysis of
$\rho_S^H$, we refer to Theorem~\ref{heilmannlieb} by Heilmann and Lieb
\cite{HeilmannLieb1995}.)
Thus, the limiting value of $Z^{-3}\rho_S(Z^{-1}\cdot(0+))$ as
$Z\to\infty$ is well defined and can be computed explicitly thanks to
the explicit knowledge of hydrogen eigenfunctions.  At the same time,
the convergence in \eqref{eq:convtfwdensitywithoutrepulsion} is
pointwise as well, so $Z^{-3}\rho_Z^{\rm TFW}(Z^{-1}\cdot)$ is
also accessible and can be computed numerically. Thus, to have
agreement of the quantum density $\rho_S$ and the TFW density on the
scale $Z^{-1}$, one may choose Weizs\"acker's parameter $A$ so that
one has the equality
\begin{align}
  \rho_S^H(0) = (A\gtf/2)^{-3/2}\rho_\infty(0).
\end{align}
This led Lieb and Liberman to the numerical value $A\approx 0.4798$,
cf.~\cite[(2.33)]{LiebLiberman1982}.

The following table summarizes plausible choices for
Weizs\"acker's coefficient $A$.
\begin{center}
  \begin{tabular}{ l c }
    & $A$ \\
    \hline\\
    Weizs\"acker (mass defect theory) & $=1$ \\ 
    Kirzhnits (gradient expansion) & $\approx 0.11$ \\  
    Yonei--Tomishima (numerical computations, Bohr atom) & $\approx 0.2$ \\
    Lieb--Liberman (energy agreement) & $\approx 0.1859$ \\
    Lieb--Liberman (density agreement) & $\approx 0.4798$
  \end{tabular}
\end{center}

\medskip
In conclusion, one may regard the proof of the Scott conjecture in TFW
theory as a warm-up problem for its proof in the full quantum problem.
(One may wonder whether Scott was aware of Weizs\"acker's extension
\cite{Weizsacker1935} at the time of writing his work \cite{Scott1952}.)
%

\subsection{Hellmann--Weizs\"acker functional}
\label{sss:sweasyproof}

In this subsection we discuss the so-called
\emph{Hellmann--Weizs\"acker functional}, which
plays an important role in the proof of the Scott conjecture by
Siedentop and Weikard, as we will discuss in
Subsubsection~\ref{sss:scottliebarguments} below.

We first introduce the Hellmann--Weizs\"acker functional for fixed
angular momentum $\ell\in\N_0$. We work on $\R_+$ with the measure $\dr$.
We set
$$
\cg^{\rm W} := \left\{ \varrho \in L^3(\R_+) :\ \varrho\geq 0,\, \sqrt{\varrho}' \in L^2(\R_+),\, \varrho(0)=0 \right\}.
$$
Here, the derivative of $\sqrt{\varrho}'$ is understood in the sense
of distributions and we recall that the square integrability of this
derivative implies that $\sqrt{\varrho}$ is continuous on $\R_+$ and has
a boundary value $\sqrt{\varrho}(0)$. In particular, the last condition
in the definition of $\cg^{\rm W}$ is well defined. Let
$$
\alpha_\ell := \left(\frac{\pi}{q(2\ell+1)}\right)^2
$$
and define, for $Z>0$ and $\varrho\in\cg^{\rm W}$,
\begin{align}
  \label{eq:defhwl}
  \ce_{\ell,Z}^{\rm HW}(\varrho)
  & := \frac12\int_0^\infty \left(((\sqrt{\varrho})'(r))^2 + \frac{\ell(\ell+1)}{r^2}\varrho(r) + \frac{\alpha_\ell}{3}\varrho(r)^3\right)\dr - \int_0^\infty \frac Zr \varrho(r)\,\dr.
\end{align}
The second term is finite by Hardy's inequality and the last term is,
since, for any $R>0$,
$$
\int_0^\infty \frac{\varrho_\ell(r)}{r}\,\dr \leq R \int_0^R \frac{\varrho_\ell(r)}{r^2}\,\dr + \left( \int_R^\infty \varrho_\ell(r)^3\,\dr \right)^{1/3} \left( \int_R^\infty r^{-3/2}\,\dr \right)^{2/3}.
$$
Thus, $\ce_{\ell,Z}^{\rm HW}$ is well defined on $\cg^{\rm W}$.

Next, we introduce the full Hellmann--Weizs\"acker functional.
It is defined on sequences $\uvarrho:=(\varrho_0,\varrho_1,...)$
with $\varrho_\ell\in\cg^{\rm W}$ for all $\ell\in\N_0$.
For such a sequence, we set
$$
\tilde D(\uvarrho,\uvarrho) := \frac12\sum_{\ell,\ell'\geq0}\iint_{\R_+\times\R_+} \dr\,\dr'\, \frac{\varrho_\ell(r)\varrho_{\ell'}(r')}{\max\{r,r'\}}.
$$ 
Let
\begin{align}
\cm^{\rm W} := \{\uvarrho \in (\cg^{\rm W})^{\N_0}:\,  & \sum_{\ell\geq 0} \int_0^\infty \!\!
\left(((\sqrt{\varrho_\ell})'(r))^2 + \frac{\ell(\ell+1)}{r^2}\varrho_\ell(r) + \frac{\alpha_\ell}{3}\varrho_\ell(r)^3\right) \dr <\infty, \notag \\
& \tilde D(\uvarrho,\uvarrho)<\infty \}
\end{align}
The \emph{Hellmann--Weizs\"acker functional} \cite{Hellmann1936}
is defined, for $\uvarrho\in\cm^{\rm W}$, by
\begin{align}
	\label{eq:defhw}
	\ce_Z^{\rm HW}(\uvarrho)
	& := \sum_{\ell\geq0}\ce_{\ell,Z}^{\rm HW}(\varrho_\ell) + \tilde D(\uvarrho,\uvarrho).
\end{align}
One can prove that $\ce_Z^{\rm HW}$ is well defined on
$\uvarrho\in\cm^{\rm W}$. This functional is studied in detail in
\cite{SiedentopWeikard1986O}; see also Hoops \cite{Hoops1993}.
Finally, for $N>0$, we set
$$
\cm_N^{\rm W} := \Big\{ \uvarrho\in\cm^{\rm W}:\ \sum_{\ell\geq 0} \int_0^\infty \varrho_\ell(r)\,\dr = N \Big\}.
$$

\medskip
The following theorem shows how the terms $(\alpha_\ell/3) \varrho_\ell^3$
in the Hellmann--Weiz\-s\"acker functional are related to the term
$(3/5)\gtf \rho^{5/3}$ in the TF functional.

\begin{theorem}
  \label{hwz2}
  Let $E_{\rm Bohr}^{\rm TF}(N,Z)$
  be the Thomas--Fermi energy of the constrained problem without
  electron-electron repulsion; see Remark \ref{bohrtfatom}.
  Then, if $N=\alpha Z$ and $Z\to\infty$,
  \begin{align}
    \label{eq:hwz2}
    & \inf\big\{\ce_{Z}^{\rm HW}(\uvarrho)- \tilde D(\uvarrho,\uvarrho) :\, \uvarrho\in\cm_N^{\rm W}\big\} 
    = E_{\rm Bohr}^{\rm TF}(N,Z) + \co_\alpha(Z^2).
  \end{align}
\end{theorem}

\begin{remarks}
  \label{remarkshw}
  (1) If one replaces Weizs\"acker's gradient term in the definition
  of $\ce_{Z}^{\rm HW}$ by $(4r^2)^{-1}\varrho_\ell$ (which is a lower
  bound by Hardy's inequality), one is led to the \emph{Hellmann functional}
  \begin{align}
    \ce_Z^{\rm H}(\uvarrho) & := \sum_{\ell\geq0}\ce_{\ell,Z}^{\rm H}(\varrho_\ell), \\
    \label{eq:hellmannl}
    \ce_{\ell,Z}^{\rm H}(\varrho_\ell) & := \frac12\int_0^\infty \left(\frac{\alpha_\ell}{3}\varrho_\ell(r)^3 + \frac{(\ell+1/2)^2}{r^2}\varrho_\ell(r)\right)\,\dr - \int_0^\infty \frac{Z}{r}\varrho_\ell(r)\,\dr.
  \end{align}
  These functionals are well defined on sets $\cg$ and $\cm$ that are
  defined in a similar manner as $\cg^{\rm W}$ and $\cm^{\rm W}$.
  A straightforward computation \cite[Theorem~1]{SiedentopWeikard1986}
  shows
  \begin{align}
    \label{eq:hellmanntf}
    \inf\{\ce_Z^{\rm H}(\uvarrho):\, \uvarrho\in\cm_N\}
    = E_{\rm Bohr}^{\rm TF}(N,Z) + \co(Z^2N^{-1/3}).
  \end{align}
  This is one step in the proof of Theorem \ref{hwz2}.

  (2) Hoops \cite[Theorem~4.5]{Hoops1993} showed that the electron-electron
  repulsion does not alter \eqref{eq:hwz2} significantly.
  Moreover, he computed the coefficient of the $Z^2$-term in this case.
  If $Z-\alpha Z^\beta \leq N\leq Z+Q_c$, where $\alpha>0$, $0\leq\beta<2/3$,
  and $Q_c\geq0$ is the $Z$-independent number specified in
  \cite[Theorem~3.3]{Hoops1993}, then one has
  \begin{align}
    \inf\{\ce_Z^{\rm HW}(\uvarrho):\, \uvarrho\in\cm_N^{\rm W}\}
    = E^{\rm TF}(Z) + qGZ^2 + o(Z^2),
  \end{align}
  where $G$ is the infimum of another explicit functional defined in
  \cite[(4.15)-(4.16)]{Hoops1993} and obeys the numerical bounds
  $2\cdot 0.388 \leq G \leq 2\cdot 0.417$, see \cite[p.~58]{Hoops1993}.
  The coefficient $G$ is about three times bigger than Scott's
  coefficient $1/4$. As Hoops puts it \cite[p.~58]{Hoops1993}:
  \emph{``To have such a big discrepancy suggests that the Hellmann--Weizs\"acker
    functional does not treat the innermost electrons sufficiently accurate
    to get the same behavior as the quantum mechanical ground state. That
    means that Weizs\"acker's gradient term is a major correction (it
    creates a $Z^2$-order term) at places where we have strong varying
    potentials but it does not suffice to give the right coefficient.''}
  %
  %
  %
  %
\end{remarks}

\subsection{Relativistic TFW functional by Engel and Dreizler}
\label{ss:tfwrel}

As discussed in the introduction, a relativistic description of large
Coulomb systems is mandatory.
This suggests to consider relativistic density functionals.
A particularly simple one can be traced back at least to Vallarta and Rosen
\cite{VallartaRosen1932} and Jensen \cite{Jensen1933}, who mimicked the steps
\eqref{eq:semiclassicalderivation1}-\eqref{eq:semiclassicalderivation3}
with the kinetic energy $p^2/2$ replaced by $\sqrt{c^2p^2+c^4}-c^2$
to derive a relativistic Thomas--Fermi theory.
For $q=2$ and nuclear configuration $\uZ=(Z_1,...,Z_K)\in(0,\infty)^K$,
$\uR=(R_1,...,R_K)\in\R^{3K}$, the resulting functional is
\begin{align}
  \ce_{c,V}^{\rm rTF}(\rho)
  := \ct^{\rm rTF}(\rho) - \int_{\R^3}V(x)\rho(x)\,\dx + D(\rho,\rho) + U
\end{align}
with
\begin{align}
  \ct^{\rm rTF}(\rho)
  := \int_{\R^3}\frac{c^5}{8\pi^2} T^{\rm rTF}\left(\frac{p(x)}{c}\right)\,\dx,
\end{align}
$T^{\rm rTF}(t):=t(t^2+1)^{3/2}+t^3(t^2+1)^{1/2}-\arsinh(t)-\frac{8}{3}t^3$,
and the Fermi momentum
\begin{align}
  p(x):=(3\pi^2\rho(x))^{1/3}.
\end{align}
This functional is unbounded from below since the relativistic kinetic
energy cannot control the Coulomb singularity.
This was already anticipated by Jensen, see also
Gomb{\'a}s \cite[\S 14]{Gombas1949}, \cite[Chapter~III, Section~16]{Gombas1956}
for a review of these facts.
Gomb{\'a}s also suggested that Weizs\"acker's (non-relativistic) inhomogeneity
correction would prevent the unboundedness from below.
His suggestion was first carried out by Tomishima \cite{Tomishima1969},
who showed, among other things, the finiteness of the energy and the electron
density at the nucleus.

While Gomb{\'a}s introduced the Weizs\"acker term ad hoc,
Engel and Dreizler \cite{EngelDreizler1987} offered a (formal) derivation
from quantum electrodynamics.
The Engel--Dreizler derivation also yields an exchange term.
In total, their functional reads
\begin{align}
  \ce_{c,V}^{\rm rTFWD}(\rho) := \ce_{c,V}^{\rm rTF}(\rho) + \ct^{\rm W}(\rho) - \cx(\rho).
\end{align}
The Weizs\"acker term is
\begin{align}
  \ct^{\rm W}(\rho) := \int_{\R^3}\frac{3A}{8\pi^2}(\nabla p)^2(x)\cdot c\cdot f\left(\frac{p(x)}{c}\right)^2\,\dx
\end{align}
with $f(t)^2:=t(1+t^2)^{-1/2}+2t^2 (1+t^2)^{-1}\arsinh(t)$ and an
adjustable parameter $A>0$.  The exchange term is
\begin{align}
  \cx(\rho) := \int_{\R^3}\frac{c^4}{8\pi^3}X\left(\frac{p(x)}{c}\right)\,\dx
\end{align}
with $X(t):=2t^4-3[t(1+t^2)^{1/2} - \arsinh(t)]^2$.

The analysis of $\ce_{c,V}^{\rm rTFWD}$ started with
Chen \cite{Chen2019} and was continued in the works
\cite{ChenSiedentop2020,Chenetal2020A}.
In the ultrarelativistic limit, i.e., in absence of the $\arsinh$ function
in the Weizs\"acker term, it had been investigated earlier
in \cite{Benguriaetal2008}.
The functional $\ce_{c,V}^{\rm rTFWD}$ is naturally defined on
\begin{align}
  P &:= \{\rho\in L^{4/3}(\R^3):\, \rho\geq0,\, D(\rho,\rho)<\infty,\, \nabla (F\circ p) \in L^2(\R^3)\},\\
  P_N&:= \{\rho\in P:\, \int_{\br^3}\rho\leq N\}
\end{align}
where $F(t):=\int_0^t f(s)\,\ds$.
In absence of the exchange term $\cx(\rho)$, Chen \cite[p.~39]{Chen2019}
proved the existence of minimizers of $\ce_{c,V}^{\rm rTFWD}$.

As we shall discuss in Subsection~\ref{ss:1poperators},
\emph{non-renormalized} relativistic quantum
models for Coulomb systems are not expected to be well defined for arbitrary
large nuclear charges, since the Coulomb potential and the kinetic energy have
the same scaling behavior, at least for high momenta.
The renormalization in Engel's and Dreizler's derivation leads to the
$\arsinh$ function in Weizs\"acker's term which ensures the lower boundedness
of $\ce_{c,V}^{\rm rTFWD}(\rho)$ for all nuclear charges $\uZ$.
The necessity of renormalization was realized early, see, e.g.,
Heisenberg and Euler \cite{HeisenbergEuler1936}.

\begin{theorem}[{\cite[Theorem~1]{Chenetal2020A}}]
  For $K=1$ and given $c,Z,A>0$ let $\kappa:=Z/(c\sqrt A)$.
  Let $\xi:=(4\pi)^{-1}\max\{X(t)/t^3:t>0\}$ and
  $s_0:\R_+\to\R_+$ be the explicit function given in
  \cite[(16)]{Chenetal2020A}, which is strictly monotone increasing and
  satisfies $s_0(0)=0$ and $\lim_{\kappa\to\infty}s_0(\kappa)=\infty$.
  Then for all $\rho\in P$ with $\int\rho=N$ one has
  \begin{align}
    \ce_{c,V}^{\rm rTFWD}(\rho) \geq -\frac{4s_0(\kappa)^5}{5 T^{\rm rTF}(s_0(\kappa))}E_{z=1}^{\rm TF}(1)Z^{7/3} - \xi c N.
  \end{align}  
\end{theorem}
In absence of the Dirac term, an analogous result was proved
in \cite[Theorem~1]{ChenSiedentop2020}.
In fact, their result also holds in the molecular case.

The existence of minimizers and bounds for the excess charge (in absence or
presence of the Dirac term) were proved by Chen \cite[p.~39]{Chen2019} and in
\cite[Theorem~2]{ChenSiedentop2020} and \cite[Theorem~2]{Chenetal2020A}.
In passing, we mention that bounds on the excess charge are
available in many non-relativistic models; see, e.g., Lieb
\cite{Lieb1981}, as well as Benguria and Lieb \cite{BenguriaLieb1985},
Solovej \cite{Solovej1990}, and
the more recent work \cite{Franketal2018T2}.

The following result concerns the energy asymptotics in the atomic case.
\begin{theorem}[{\cite[Theorem~1]{Siedentop2021}}]
  \label{edtf}
  Let $K=1$ and $Z/c>0$ be fixed. Then
  \begin{align}
    \label{eq:edtf}
    \inf\ce_{c,Z/|x|}^{\rm rTFWD}(P_Z) = E^{\rm TF}(Z) + \co(Z^2)
    \quad \text{as}\ Z\to\infty.
  \end{align}
\end{theorem}
The core elements of the proof are the facts that the relativistic
kinetic energy is dominated by the non-relativistic one, that
relativistically described electrons far away from the nucleus behave
non-relativistically, and that the TFW functional provides a
$Z^2$-correction to the Thomas--Fermi energy.
One may argue that Theorem \ref{edtf} is expected in view of the
heuristics explained in Subsection \ref{sss:relatomtf}. Still, it is
quite surprising that a formally derived functional correctly yields a
fundamental feature like the ground state energy of large atoms.

\section{Quantum mechanics of non-relativistic Coulomb systems}
\label{s:nonrelativistic}

In this section we review results for the energy and density of
non-relativistic systems with one or several nuclei. These are
described by the Hamilton operator $H_{N,V}$ in
\eqref{eq:manybodySchrodinger}.
Recall that the ground state energy $E_S(N,\uZ,\uR)$ is defined in
\eqref{eq:gsenergynonrel}.

\subsection{Atoms without magnetic fields}
\label{ss:nonrelatom}

\subsubsection{Thomas--Fermi scale}

One of the main results of the work \cite{LiebSimon1977} of Lieb and Simon
is that TF theory correctly describes both the leading order of the
quantum mechanical ground state energy of large atoms, and the one-particle
electron distribution on the length scale $Z^{-1/3}$ in the limit $Z\to\infty$.
This is summarized in the following theorem.

\begin{theorem}
  \label{quantumtfconvnonrelatom}
  Let $E_S(N,Z)=\inf\spec(H_{N,Z})$ be the ground state energy of $H_{N,Z}$.
  Fix $\alpha=N/Z$. Then
  \begin{align}
    \label{eq:quantumtf}
    \lim_{Z\to\infty}\frac{E_S(N,Z)}{Z^{7/3}} = E^{\mathrm{TF}}(\alpha,1).
  \end{align}
  If $E_S(N,Z)$ is an eigenvalue of $H_{N,Z}$ and $\rho_S$ is the one-particle density of
  any of its associated normalized eigenfunctions, then one has
  \begin{align}
    \label{eq:quantumtfdensitycoulomb}
    \lim_{Z\to\infty}\|Z^{-2}\rho_S(Z^{-1/3}\cdot)-\rho_{z=1}^{\rm TF}(\alpha,\cdot)\|_C=0,\\
    \label{eq:quantumtfdensityl53}
    \lim_{Z\to\infty}\int_{\R^3} U(x)Z^{-2}\rho_S(Z^{-1/3}x)\,\dx
    = \int_{\R^3} U(x)\rho_{z=1}^{\mathrm{TF}}(\alpha,x)\,\dx
  \end{align}
  for
  all $U\in L^{5/2}(\R^3)$.
  Finally, if $N\leq Z$, then 
  the convergence
  $Z^{-2}\rho_S(Z^{-1/3}\cdot) \to \rho_{z=1}^{\mathrm{TF}}(\alpha,\cdot)$
  also holds in the weak $L^1$-sense.
\end{theorem}

\begin{remark}
  \label{remnoeigenvalue}
  As mentioned before, $E_S(N,Z)$ is indeed an eigenvalue of $H_{N,Z}$
  in the most relevant case $N=Z$.  However, there are generalizations
  of the convergence statements
  \eqref{eq:quantumtfdensitycoulomb}--\eqref{eq:quantumtfdensityl53}
  that remain valid even if $E_S(N,Z)$ is not an eigenvalue. Namely,
  one can take $\rho_S$ to be the one-particle density of elements in
  any sequence of normalized functions $\psi_N$, $N=1,2,...$ in the
  form domain of $H_{N,Z}$ satisfying
  \begin{align}
    \lim_{N\to\infty}\frac{\langle \psi_N,H_{N,Z}\psi_N\rangle - E_S(N,Z)}{Z^{7/3}} = 0.
  \end{align}
  Such a sequence $(\psi_N)_{N\in\N}$ is sometimes called approximate
  ground state of $H_{N,Z}$ on the Thomas--Fermi scale.
\end{remark}

Asymptotics \eqref{eq:quantumtf} and \eqref{eq:quantumtfdensityl53}
(the latter for characteristic functions of bounded, measurable sets)
are due to Lieb and Simon \cite[Theorems~III.1, III.3]{LiebSimon1977}.
In the context of a proof of stability of matter, Thirring \cite{Thirring1981}
found a substantially simpler proof of the lower bound in \eqref{eq:quantumtf}
that used coherent states involving Gaussians.
Lieb slightly generalized these coherent states and found a shorter proof for
the upper bound in \eqref{eq:quantumtf},
see \cite[Theorem~5.1]{Lieb1981} (where also an adaption of Thirring's proof
of the lower bound is reported).
Formula \eqref{eq:quantumtfdensityl53} (for certain characteristic
functions) is also due to Baumgartner \cite{Baumgartner1976},
while the convergence \eqref{eq:quantumtfdensitycoulomb} of $\rho_S$ in
Coulomb norm was proved by Fefferman and Seco \cite{FeffermanSeco1989}.
(Note that the Scott correction with error bound $\mathcal{O}(Z^{47/24})$
(see \cite{SiedentopWeikard1987O,SiedentopWeikard1989,Bach1989E,Bach1989})
actually implies the quantitative bound
$\|Z^{-2}\rho_S(Z^{-1/3}\cdot)-\rho_{z=1}^{\rm TF}(\alpha,\cdot)\|_C\lesssim Z^{-3/16}$.)
Asymptotics \eqref{eq:quantumtfdensityl53} for general $U\in L^{5/2}(\R^3)$
follow from the fact that finite sums of characteristic functions of
bounded measurable sets are dense in $L^{5/2}$, together with the uniform
boundedness $\int_{\R^3} \left(Z^{-2}\rho_S(Z^{-1/3}x)\right)^{5/3}=\co(1)$,
which in turn follows from the kinetic Lieb--Thirring inequality.

%

\begin{remark}[The Bohr atom]
  Let $E_{\rm Bohr}(N,Z)$ be the ground state energy of $H_{N,V}$
  in the atomic case ($K=1$, $R=0$) and without electron-electron
  repulsion, that is, without the double sum in
  \eqref{eq:manybodySchrodinger}. Let $E_{\rm Bohr}^{\rm TF}(N,Z)$
  be the corresponding quantity in TF theory, defined in Remark
  \ref{bohrtfatom}, and note that, by scaling,
  $E_{\rm Bohr}^{\rm TF}(N,Z)= E_{\rm Bohr}^{\rm TF}(N/Z,1)\cdot Z^{7/3}$.
  We claim that, for any fixed $\alpha=N/Z$,
  $$
  \lim_{Z\to\infty} \frac{E_{\rm Bohr}(N,Z)}{Z^{7/3}} = E_{\rm Bohr}^{\rm TF}(\alpha,1).
  $$
  Indeed, this can either be proved by following the proof of
  Theorem \ref{quantumtfconvnonrelatom} or, much more directly, by
  using the explicit formula for the eigenvalue $E_{\rm Bohr}(N,Z)$
  and the explicit formula for $E_{\rm Bohr}^{\rm TF}(N,Z)$ in
  \eqref{eq:tfenergyneutralbohr}.
\end{remark}

Let us return to the situation of Theorem \ref{quantumtfconvnonrelatom}.
Recall that the TF minimizer satisfies
$\rho_{z=1}^{\rm TF}(x) \sim \const |x|^{-3/2}$ as $x\to0$
(see Theorem \ref{tfdensityproperties}).
Thus, Theorem \ref{quantumtfconvnonrelatom} implies that the ground state
density develops a singularity on the scale $Z^{-1/3}$.
This is not surprising, since the electrons in a Bohr atom have a density
proportional to $Z^3$,
which tends to infinity relative to the TF magnitude $Z^2$.
Thus, TF theory is not expected to provide a complete description of large
atoms, especially on scales $Z^{-1}$ close to the nucleus.
%
For this reason we now look closer at the electrons on the
shorter, hydrogenic length scale $Z^{-1}$.

\subsubsection{Scott scale}\label{sec:scottnonrel}

As explained in the introduction, Scott's idea is that the leading correction
to TF theory stems from the few innermost high-energy electrons on distances
$Z^{-1}$ from the nucleus, i.e., on the natural length scale of the hydrogen
operator $S_Z^H$. Recall also that the eigenvalues of this operator have
magnitude $\co(Z^2)$.

The following theorem states that 
\emph{Scott's conjecture is indeed true}.
Hughes \cite{Hughes1986,Hughes1990} (lower bound) and the works
\cite{SiedentopWeikard1987O,SiedentopWeikard1989} (upper and lower bound)
proved this conjecture for neutrally charged atoms.
Since the Scott correction should not depend on ``electrons on the outermost
shells'', it was believed that the conjecture also holds for ions. 
Indeed, Bach \cite{Bach1989,Bach1989E} showed that also this intuition is
correct, and proved Scott's conjecture for (positive and negative) ions.

\begin{theorem}[{\cite{Hughes1986,Hughes1990,SiedentopWeikard1987O,SiedentopWeikard1989,Bach1989,Bach1989E}}]
  \label{swscottnonrel}
  Let $\alpha>0$ and $N=\alpha Z$ with $N\in\N$.
  Then the ground state energy $E_S(N,Z)$ of $H_{N,Z}$ in
  \eqref{eq:manybodySchrodinger} satisfies
  \begin{align}
    \label{eq:swscottnonrel}
    E_S(N,Z) = E^{\rm TF}(\alpha,1)\cdot Z^{7/3} + \frac{q}{4} \cdot Z^2 + \mathcal{O}(Z^{47/24}).
  \end{align}
\end{theorem}

We will give a heuristic argument in favor of Theorem
\ref{swscottnonrel} later in Subsection \ref{sss:scottliebarguments}.
At this point we would also like to mention two alternative and
interesting proofs of Theorem \ref{swscottnonrel} by
Ivrii and Sigal \cite{IvriiSigal1993} and by
Solovej and Spitzer \cite{SolovejSpitzer2003}, which do not use the
spherical symmetry of the atomic case (see Theorem \ref{scottnonrelmolecules}).
Both used them to prove the Scott conjecture for clamped nuclei whose
distances are scaled by $Z^{-1/3}$.

\medskip

Next, we turn to the strong Scott conjecture. 
Scott's heuristics led Lieb \cite[p.~623]{Lieb1981} to the belief that
for large $Z$ the suitably rescaled electron density $\rho_S$ on distances
$Z^{-1}$ from the nucleus
converges to the hydrogenic density $\rho_S^H$ (see \eqref{eq:defrhohs}).
Since the magnitude of the total electronic density of a Bohr atom is
proportional to $Z^3$,
one suspects that the correct object to study is the scaled density
$Z^{-3}\rho_S(Z^{-1}x)$.
A step towards the proof of the strong Scott conjecture in the full
$N$-particle setting was taken in \cite{Siedentop1994A},
where the following upper bound was shown:
\begin{align}
  \label{eq:rakowskysiedentoppre}
  \limsup_{Z\to\infty} \frac{\rho_S(0)}{Z^3} \leq \frac{\pi}{24}q;
\end{align}
see also Theorem \ref{rakowskysiedentop} for a more general statement.
Interestingly, the numerical value $\pi/24$ is quite close to
$q^{-1}\cdot\rho_S^H(0)=\pi^{-1}\zeta(3)$ (see Theorem \ref{heilmannlieb}).
Motivated by this strong numerical evidence,
Iantchenko et al.~\cite{Iantchenkoetal1996} eventually
proved \emph{Lieb's strong Scott conjecture}.

In fact, in \cite{Iantchenkoetal1996} also an angular-momentum-resolved
version of this conjecture is proved. To formulate this result, for
$\ell\in\N_0$, we let $Y_{\ell,m}$, $m\in\{-\ell,...,\ell\}$, be
$L^2(\bs^{2})$-orthonormal spherical harmonics of degree $\ell$. We then
define one-dimensional, radial, angular-momentum-resolved densities of
$\psi\in \bigwedge_{\nu=1}^N L^2(\br^3:\bc^q)$ by
\begin{align}
  \label{eq:defrhononrell}
  \varrho_{\ell}(r)
  := Nr^2 \sum_{m=-\ell}^\ell \sum_{\sigma=1}^q \int_{\Gamma^{N-1}}\left|\int_{\bs^2}\rd\omega\,\overline{Y_{\ell,m}(\omega)}\psi(r\omega,\sigma;y_2,...,y_N)\right|^2\,\dy_2...\dy_N
\end{align}
for $\ell\in\N_0$ and $r>0$. (In the following we use the letter
$\varrho:\R_+\to\R_+$ to denote one-dimensional, radial densities with
particle number $\int_0^\infty\varrho(r)\dr$, i.e., we integrate with
respect to $\dr$ and not $r^2\dr$.  Three-dimensional densities are
denoted by the letter $\rho:\R^3\to\R_+$.) The densities
$\varrho_{\ell}$ in \eqref{eq:defrhononrell} are related to the total
density $\rho$ in \eqref{eq:defrhohnonrel} by
\begin{align}
  \int_{\bs^2}\rho(r\omega)\,\rd\omega = r^{-2}\sum_{\ell=0}^\infty\varrho_{\ell}(r)
\end{align}
for $r>0$.
If $\psi$ in \eqref{eq:defrhononrell} is an eigenfunction of $H_{N,V}$
with eigenvalue $E_S(N,Z,R)$, we write
\begin{align}
	\label{eq:defrhogroundstateintrol}
	\varrho_{\ell,S}
\end{align}
for \eqref{eq:defrhononrell} and analogously for other Hamiltonians
that we discuss later.

Recall that the (spinless) hydrogen Hamiltonian $S^H$ was defined in
\eqref{eq:defhydrogenoperator2}. Due to the spherical symmetry of $S^H$,
one can consider its parts in a fixed angular momentum channel
$\ell\in\N_0$, i.e.,
\begin{align}
  \label{eq:defhydrogenoperatorl}
  \frac12\left(-\frac{\rd^2}{\dr^2} + \frac{\ell(\ell+1)}{r^2}\right) - \frac{1}{r}
  \quad \text{in } L^2(\R_+,\dr: \C).
\end{align}
The eigenfunctions $\psi_{n,\ell,m}^S$ of $S^H$ are then of the form
\begin{align}
  \psi_{n,\ell,m}^S(x) = \frac{\psi_{n,\ell}^S(|x|)}{|x|}Y_{\ell,m}\left(\frac{x}{|x|}\right),
\end{align}
where $\psi_{n,\ell}^S\in L^2(\R_+,\dr)$ are normalized eigenfunctions
of \eqref{eq:defhydrogenoperatorl}.
Uns\"old's theorem (cf.~\cite[p.~377]{Unsold1927},
\cite[Section~18.4]{WhittakerWatson1927}) states that
$\sum_{m=-\ell}^\ell Y_{\ell,m}(\omega)\overline{Y_{\ell,m}(\sigma)}=\frac{2\ell+1}{4\pi}P_\ell(\omega\cdot\sigma)$
for $\omega,\sigma\in\bs^2$ and the $\ell$-th Legendre polynomial $P_\ell$,
which obeys $P_\ell(1)=1$. Using this, one can carry out the $m$-summation
in the definition of the hydrogenic density $\rho^H_S$ in \eqref{eq:defrhohs}
and obtains
\begin{align}
  \label{eq:defrhohs2}
  \rho_S^H(x) = \frac{1}{4\pi |x|^2}\sum_{\ell=0}^\infty \varrho_{\ell,S}^H(|x|),
  \quad x\in\R^3
\end{align}
with the one-dimensional, angular-momentum-resolved hydrogenic densities
\begin{align}
  \label{eq:defrhohsl}
  \varrho_{\ell,S}^H(r) := q(2\ell+1)\sum_{n=0}^\infty |\psi_{n,\ell}^S(r)|^2, \quad r>0.
\end{align}
Formula \eqref{eq:defrhohs2} implies, in particular, that $\rho_S^H$ is
spherically symmetric, as mentioned in the introduction. The right sides
of \eqref{eq:defrhohs2} and \eqref{eq:defrhohsl} converge;
see Theorem~\ref{heilmannlieb} for details.

The following theorem states the validity of the strong Scott conjecture.

\begin{theorem}[{\cite{Iantchenkoetal1996}}]
  \label{ilsdensitynonrel}
  Let $\ell\in\N_0$, $E_S(N,Z)$ be an eigenvalue of $H_{N,Z}$, and $\rho_S$
  and $\varrho_{\ell,S}$ be the total and angular momentum resolved
  one-particle densities associated to any of its eigenfunctions, respectively.
  Then the following statements hold.
  
  \begin{enumerate}
  \item (Convergence of angular momentum density). For all $r>0$ one has
    pointwise convergence
    \begin{align}
      \lim_{Z\to\infty}Z^{-3}\varrho_{\ell,S}(r/Z) = \varrho_{\ell,S}^H(r).
    \end{align}
    Moreover, for $v\in L^1(\R_+,\dr)$ one has
    \begin{align}
      \lim_{Z\to\infty}\int_0^\infty r^{-1}v(r) Z^{-3}\varrho_{\ell,S}(r/Z)\,\dr
      = \int_0^\infty r^{-1}v(r) \varrho_{\ell,S}^H(r)\,\dr.
    \end{align}

  \item (Convergence of total density) Let $W\in L^\infty(\bs^2)$ and $r>0$.
    Then the total density, when spherically averaged, converges pointwise
    to the hydrogenic density, i.e.,
    \begin{align}
      \lim_{Z\to\infty}\int_{\bs^2}W(\omega)Z^{-3}\rho_S(r\omega/Z)\,\rd\omega
      = \int_{\bs^2}W(\omega)\rho_S^H(r\omega)\,\rd\omega.
    \end{align}
    Moreover, for any locally bounded $v \in L^1(\R^3)$, one has
    \begin{align}
      \lim_{Z\to\infty}\int_{\R^3}|x|v(x)Z^{-3}\rho_S(x/Z)\,\dx
      = \int_{\R^3}|x|v(x)\rho_S^H(x)\,\dx.
    \end{align}
  \end{enumerate}
\end{theorem}

\begin{remark}
  \label{remnoeigenvaluescott}
  Theorem \ref{ilsdensitynonrel} remains valid for the total and
  angular momentum resolved one-particle densities $\rho_S$ and
  $\varrho_{\ell,S}$ of elements in any sequence of normalized
  functions $\psi_N$, $N=1,2,...$ in the form domain of $H_{N,Z}$
  satisfying
  \begin{align}
    \lim_{N\to\infty}\frac{\langle \psi_N,H_{N,Z}\psi_N\rangle - E_S(N,Z)}{Z^{2}} = 0.
  \end{align}
  Such a sequence $(\psi_N)_{N\in\N}$ is sometimes called approximate
  ground state of $H_{N,Z}$ on the Scott scale.
\end{remark}

Later, we shall see that the hydrogenic density obeys
\begin{align}
  \label{eq:rhohlargedistances}
  \rho_S^H(x) = \left(\frac{1}{\gtf}\right)^{3/2}|x|^{-3/2} + o(|x|^{-3/2})
\end{align}
as $|x|\to\infty$. Since Theorem \ref{tfdensityproperties} asserts
that the TF density has the exact same asymptotic behavior,
but for small $|x|$ (on the scale $Z^{-\frac13}$),
Theorem \ref{ilsdensitynonrel} shows that there is a smooth transition
of $\rho_S$ between the length scales $Z^{-1}$ and $Z^{-\frac13}$.

\begin{remarks}[More results on the quantum density]
  We summarize some further results for the many-particle ground state
  density.
  \begin{enumerate}
  \item Recall that Theorem \ref{ilsdensitynonrel} does not say anything yet about
    the quantum density at the origin $r=0$. Therefore, it is of interest to find
    quantitative upper bounds on $\rho_S(0)$. We record the following theorem.
    \begin{theorem}[{\cite{Siedentop1994A,Rakowsky1995,RakowskySiedentop1995}}]
      \label{rakowskysiedentop}
      Let $(\psi_N)_{N\in\N}$ be an approximate ground state of $H_{N,Z}$ on the
      Scott scale in the sense of Remark \ref{remnoeigenvaluescott} and
      $\rho_S$ be the one-particle density of the element $\psi_N$ of that sequence.
      Fix any $c>0$ and assume that $N/Z>0$ is fixed or that $N>Z-cZ^{1/2}$.
      Then one has
      \begin{align}
        \label{eq:rakowskysiedentop}
        \limsup_{Z\to\infty} \frac{\rho_S(0)}{Z^3} \leq \frac{\pi}{24}q.
      \end{align}
    \end{theorem}
    For $N=Z$ this was shown in \cite{Siedentop1994A} before
    the strong Scott conjecture \cite{Iantchenkoetal1996} was proved.
    The result for ions is due to Rakowsky \cite{Rakowsky1995} and
    \cite{RakowskySiedentop1995}.

    At the core of the proof of Theorem \ref{rakowskysiedentop} lies
    a linear response argument
    and the inequality
    \begin{align}
      \rho_S(0) \leq \frac{Z}{2\pi}\int_{\R^3}\frac{\rho_S(x)}{|x|^2}\,\dx
    \end{align}
    due to Hoffmann--Ostenhof et al.~\cite{Hoffmann-Ostenhofetal1978}.
    Theorem \ref{rakowskysiedentop} supports the belief that the strong
    Scott conjecture holds for ions as well.
    
  \item The work \cite{IantchenkoSiedentop2001}
    proved the convergence of the one-particle ground state density
    matrix on Scott's scale. Consider an approximate ground state
    $(\psi_N)_{N\in\N}$ of $H_{N,Z}$ on the Scott scale.
    For $x,y\in\R^3$ let
    \begin{align}
      \begin{split}
        \gamma_S(x,y)
        & := \sum_{\sigma_1,...,\sigma_N=1}^q\sum_{\nu=1}^N \int_{\R^{3(N-1)}} \dx_1...\dx_{\nu-1}\,\dx_{\nu+1}...\dx_N\\
        & \quad \psi_N(x_1,...,x_{\nu-1},x,x_{\nu+1},...,x_{N};\sigma_1,...,\sigma_N)\\
        & \quad \times \overline{\psi_N(x_1,...,x_{\nu-1},y,x_{\nu+1},...,x_{N};\sigma_1,...,\sigma_N)}.
      \end{split}
    \end{align}
    be the kernel of the associated approximate one-particle ground state density
    matrix $\gamma_S$.
    Observe that $0\leq\gamma_S\leq q$ is a trace class operator.
    Recall the $L^2(\R^3)$-normalized eigenfunctions $\psi_{n,\ell,m}^S$
    of the hydrogen operator \eqref{eq:defhydrogenoperator2}. Define
    the orthogonal projection onto the negative (purely discrete)
    spectral subspace of this operator by
    \begin{align}
      \gamma_S^H := q \cdot \sum_{n\geq0}\sum_{\ell\geq0}\sum_{m=-\ell}^\ell|\psi_{n,\ell,m}^S\rangle\langle\psi_{n,\ell,m}^S|,
    \end{align}
    where the series is pointwise convergent (see Theorem
    \ref{heilmannlieb} by Heilmann and Lieb \cite{HeilmannLieb1995}).
    Then the following result holds.

    \begin{theorem}[{\cite[Theorem 1]{IantchenkoSiedentop2001}}]
      \label{convdensitymatrixscott}
      Let $K$ be a trace-class operator and $\gamma_S$ be an approximate
      one-particle ground state density matrix of the Schr\"odinger
      operator $H_{N,Z}$ in \eqref{eq:manybodySchrodinger}.
      Then one has the weak convergence
      \begin{align}
        \lim_{Z\to\infty}\tr\left(K \cdot Z^{-3}\gamma_S\left(\frac{\cdot}{Z},\frac{\cdot}{Z}\right)\right)
        = \tr(K\gamma_S^H).
      \end{align}
    \end{theorem}
    
  \item As already indicated, there is a smooth transition of
    $\rho_S$ between the scales $Z^{-1}$ and $Z^{-1/3}$
    (in view of Theorems \ref{tfdensityproperties} and
    \ref{ilsdensitynonrel}, and Formula \eqref{eq:rhohlargedistances}).
    This may have led to Lieb's belief \cite{Lieb1981} that
    $Z^{-3/2-3/(2\delta)}\rho_S(x/Z^\delta)$ should also behave like
    $(1/\gtf)^{3/2}|x|^{-3/2}$ when $\delta\in(1/3,1)$.

    In \cite{Iantchenko1997} Iantchenko succeeded in proving this conjecture
    using some of Ivrii's and Sigal's methods \cite{IvriiSigal1993}
    adapted to the case of (neutral) atoms.
    
    \begin{theorem}[{\cite[Theorem 2]{Iantchenko1997}}]
      \label{iantchenkointermediate}
      Let $(\psi_N)_{N\in\N}$ be an approximate ground state of $H_{N,Z}$
      (in \eqref{eq:manybodySchrodinger}) on the Scott scale in the
      sense of Remark \ref{remnoeigenvaluescott} and let $\rho_S$ be the
      one-particle density associated to $\psi_N$.
      Let $U\in C^\infty(\R^3\setminus\{0\})\cap L^\infty(\R^3)$, and assume
      \begin{align}
        |\partial^\nu U(x)| & \lesssim_\nu |x|^{-|\nu|-1}\langle x\rangle^{-3}, \quad |x|\geq2, |\nu|\geq0,\\
        |\partial^\nu U(x)| & \lesssim_\nu |x|^{-|\nu|}, \quad |x|\leq2, |\nu|\geq0.
      \end{align}
      (Here $\nu=(\nu_1,\nu_2,\nu_3)\in(\Z\setminus(-\infty,0))^3$ and
      $|\nu|=\nu_1+\nu_2+\nu_3$. We also use the notation
      $\langle x\rangle:=\sqrt{1+|x|^2}$ for $x\in\R^d$.)
      Then for any $\delta\in(1/3,1)$, one has
      \begin{align}
        \lim_{Z\to\infty}\int_{\R^3}U(x) Z^{-3/2-3/(2\delta)}\rho_S(x/Z^\delta)\,\dx
        = \int_{\R^3}U(x) \left(\frac{1}{\gtf}\right)^{3/2}|x|^{-3/2}\,\dx.
      \end{align}
    \end{theorem}

    It may very well be true that Theorem \ref{iantchenkointermediate} also
    holds under less severe assumptions on the test function $U$.
    Theorem \ref{iantchenkointermediate} can also be generalized to
    the molecular case, see \cite[Theorem~7]{Iantchenko1997} for details.
    Further extensions were recently outlined by Ivrii
    \cite{Ivrii2019,Ivrii2019S,Ivrii2019T}.
  \end{enumerate}
\end{remarks}

Some elements of the proof of Theorem \ref{ilsdensitynonrel} will
be given in the relativistic case, which is in spirit the same but
technically more elaborate (Subsection \ref{sss:elementsfmss}).
Here we just mention that all known proofs for the convergence of the
density (as of this writing) rely on energetic results.
It would be interesting to see if this scheme could be reversed. 

Before we come to the arguments in favor of Theorem \ref{swscottnonrel},
we collect some properties of the hydrogenic density $\rho_S^H$.

\subsubsection{Hydrogenic density}

Recall the normalized eigenfunctions $\psi_{n,\ell,m}^S$ of the
hydrogen operator \eqref{eq:defhydrogenoperator2} acting in
$L^2(\R^3:\C)$.  Recall also the definition of the hydrogenic
densities $\rho_S^H$ and, for $\ell\in\N_0$, the angular momentum
resolved version $\varrho_{\ell,S}^H$ in
\eqref{eq:defrhohs} and \eqref{eq:defrhohs2}--\eqref{eq:defrhohsl},
respectively. The specific knowledge of the $\psi_{n,\ell,m}^S$
allowed Heilmann and Lieb \cite{HeilmannLieb1995} to analyze
$\rho_S^H$ and $\varrho_{\ell,S}^H$ in great detail.
For $n\in\N_0$, $\ell\in\N_0$, and $m\in\{-\ell,...,\ell\}$ they are
\begin{align}
  \psi_{n,\ell,m}^S(x)
  = R_{n,\ell}(r)Y_{\ell,m}\left(\frac{x}{|x|}\right),
\end{align}
where the $R_{n,\ell}$ are explicitly known \cite[p.~3630]{HeilmannLieb1995}.
Thus, by Uns\"old's theorem,
\begin{align}
  \label{eq:defrhohs3}
  \rho_S^H(x) = \frac{q}{4\pi}\sum_{n\geq0}\sum_{\ell\geq0}(2\ell+1)R_{n,\ell}(r)^2.
\end{align}
Making heavy use of properties of special functions, they carried
out the $\ell$- and $n$-summations (in that order), and proved the
following theorem.

\begin{theorem}[{\cite[Theorem~1, p.~3631]{HeilmannLieb1995}}]
  \label{heilmannlieb}~
  \begin{enumerate}
  \item The series on the right side of \eqref{eq:defrhohs3}
    converges pointwise.

  \item One has the asymptotic expansion, as $r\to\infty$,
    \begin{align}
      \label{eq:heilmannlieb1}
      \begin{split}
        \rho_S^H(r) = \frac{q}{\sqrt2 \pi^2}r^{-3/2} & \left[\sum_{j\geq0}a_j(8r)^{-j} - \sin(\sqrt{32r})\sum_{j\geq1}b_j(8r)^{-j} \right. \\
        & \qquad\qquad\qquad \left. + \cos(\sqrt{32r})\sum_{j\geq1}c_j(8r)^{-j-1/2}\right]
      \end{split}
    \end{align}
	The first few coefficients are
    \begin{align*}
      a_0 = 2/3 \quad a_1 = -1/12 \quad a_2 = 79/960, \\
      b_1 = 3/2 \quad b_2 = -140\,589/11\,200, \\
      c_1 = 141/40 \quad c_2 = -2\,028\,627/44\,800.
    \end{align*}

  \item Let $\lambda\in(0,1]$ and $\rho^{\rm TF}(\lambda,1,x)$ be the
    TF minimizer of $\ce_{z=1}^{\rm TF}$ on $\ci_\lambda$.
    Then for any $\lambda\in(0,1]$,
    \begin{align}
      \label{eq:heilmannlieb2}
      \lim_{|x|\to\infty} |x|^{3/2}\rho_S^H(|x|)
      = \lim_{|x|\to0} |x|^{3/2}\rho_{z=1}^{\rm TF}(\lambda,x).
    \end{align}
    
  \item $\rho_S^H(r)$ is monotone decreasing and achieves its
    maximum at $r=0$ with
    \begin{align}
      \label{eq:heilmannlieb3}
      \rho_S^H(r) \leq \rho_S^H(0)
      = \frac{q}{\pi}\sum_{n\geq0}(n+1)^{-3} \approx q\cdot 0.383.
    \end{align}
  \end{enumerate}
\end{theorem}

\begin{remark}
  The ``shell structure'', i.e., the oscillations of $\rho_S^H(r)$
  are barely visible.
  As Heilmann and Lieb put it \cite[p.~3633]{HeilmannLieb1995}:
  ``\emph{In fact, it is necessary to take two derivatives with
    respect to $r$ in order to make the oscillatory terms as large
    as the nonoscillatory ones. In short, shell structure is not a
    prominent property of this universal atomic function.}''
  For graphical illustrations, see, e.g., \cite[p.~3633]{HeilmannLieb1995}.
\end{remark}

\subsubsection{Scott's derivation of the $Z^2$-correction and ideas of
  Siedentop and Weikard}
\label{sss:scottliebarguments}

We first present Scott's \cite{Scott1952} heuristic derivation of the
$Z^2$-correction to TF theory. For simplicity, we restrict ourselves
to the neutral case $N=Z$.
Our exposition follows March \cite[pp.~8--11]{March1957T}.
Afterwards we present the ideas in \cite{SiedentopWeikard1987O}
to prove the upper bound \eqref{eq:swscottnonrel}.
\medskip

Scott's idea is that the $Z^2$-correction stems from the innermost
electrons, which live roughly a distance $Z^{-1}$ away from the nucleus.
Since there are only ``few'' electrons on this scale, one expects the
electron-electron repulsion to be irrelevant there, i.e., the electrons
should be described by the hydrogen operator $S_Z^H$.

These heuristics suggest that the first correction of the Thomas--Fermi energy
for a large atom should be the difference between the ``Bohr energy'' (i.e.,
the sum over all hydrogen eigenvalues) and the Thomas--Fermi Bohr energy.
We computed the latter in \eqref{eq:tfenergyneutralbohr}. In the neutral case,
we obtain
\begin{align}
  E^{\rm TF}_{\rm Bohr}(Z) = -(3^{1/3}/2)q^{2/3}Z^{7/3}.
\end{align}
On the other hand, in the Bohr atom, each shell (indexed by $n\in\N$)
of energy $-Z^2/(2n^2)$ has $qn^2$ states. Thus, the $N=Z$ many electrons
occupy $k$ shells, where $k$ is determined by
\begin{align}
  Z
  = N
  = q\left(\sum_{n=1}^k n^2+(k+1)^2\epsilon\right)
  = q\left(\frac{k^3}{3} + \frac{k^2}{2} + \frac{k}{6} + (k+1)^2\epsilon\right)
\end{align}
and where $0\leq\epsilon<1$ is the fraction of the $(k+1)$-st shell that
is filled. One finds $k=(3Z/q)^{1/3}-\frac12-\epsilon+o(1)$ as $Z\to\infty$.
Thus, the energy of the Bohr atom is
\begin{align}
  E_{\rm Bohr}(Z) = -\frac{Z^2}{2}q\left(\epsilon+\sum_{n=1}^k 1\right)
  = E^{\rm TF}_{\rm Bohr}(Z) + \frac{q}{4}Z^2 + o(Z^2).
\end{align}
The second term here is exactly the Scott correction.

\medskip

We now present some ideas from \cite{SiedentopWeikard1987O} that enter
into the proof of the upper bound in \eqref{eq:swscottnonrel}. As
usual, an upper bound will be derived using a suitable trial state. As
we will describe later in more detail, this trial state contains, in
addition to hydrogen orbitals, also so-called Macke orbitals. It is in
connection with these Macke orbitals that the Hellmann--Weizs\"acker
functional defined in Subsection \ref{sss:sweasyproof} comes into
play \cite{Siedentop1981}.

To motivate why this functional is relevant, we recall that the TFW
energy has an asymptotic expansion whose first term coincides with the
TF energy and whose second term behaves like $Z^2$.  If one could show
that the TFW functional (with some choice of $A$) provides an upper
bound for the ground state energy $E_S(N,Z)$, then there might be a
chance that this functional can be used to prove the upper bound in
\eqref{eq:swscottnonrel}.
%
Apart from a factor that decreases like $N^{-2}$ and is irrelevant in the
limit $N\to\infty$, this was indeed proved in \emph{one spatial dimension}
by March \cite{March1957} and March and Young \cite{MarchYoung1958}.
They showed that the sum of the first $N$ eigenvalues of $-\frac{d^2}{dx^2}- v$
in $L^2(\R)$ is bounded from above by
\begin{align}
  \begin{split}
    \inf\left\{\int_\R\left[\frac12((\sqrt\rho)')^2+\frac12\cdot\frac{\pi^2}{3}\left(1-\frac{1}{N^2}\right)\rho(x)^3 - v(x)\rho(x)\right]\dx\,:\, \rho\in\ca_N\right\},
  \end{split}
\end{align}
where $\ca_N$ is the one-dimensional analog of \eqref{eq:defcalambda};
see also \cite[Section~3]{Siedentop2020M} or M\"uller
\cite[pp.~114--139]{Muller1978}.  To prove this, March and Young used
cleverly constructed trial functions by Macke
\cite{Macke1955,Macke1955W}.  These so-called Macke orbitals may be
seen as a (finer) analogue of coherent states.  March and Young also
proposed $d$-dimensional versions of Macke's orbitals, which, however,
lead to a contradiction, as was pointed out by M\"uller
\cite[pp.~131--134]{Muller1978}, Lieb \cite[p.~96]{Lieb1980},
and the author of
\cite[p.~213]{Siedentop1981}; see also Dietze \cite{Dietze2018} for a
closer inspection.

Let us return to the three-dimensional case.  In
\cite[Theorem~5.2]{SiedentopWeikard1986O} and
\cite[(4.1)]{SiedentopWeikard1986} (see also Lad\'anyi
\cite{Ladanyi1958} and \cite{Siedentop1981} for precursors) the
following upper bound was proved in the atomic case:
\begin{align}
  \label{eq:upperboundhw}
  E_S(N,Z) \leq \inf\left\{\ce_Z^{\rm HW}(\uvarrho) + \mathcal R(\uvarrho) :\, \uvarrho\in\cm_N^{\rm W}\right\}
\end{align}
with an (unimportant) remainder term
\begin{align}
  \mathcal R(\uvarrho) := \sum_{\ell\geq0, n_\ell\neq0} \frac{\alpha_\ell}{3}\left(\frac{-
  1+6\epsilon_\ell-3\epsilon_\ell^2}{n_\ell^2} +
  \frac{2\epsilon_\ell^3-
  6\epsilon_\ell^2+4\epsilon_\ell}{n_\ell^3}\right) \int_0^\infty
  \varrho_\ell(r)^3\,\dr.
\end{align}
Here $n_\ell=(q(2\ell+1))^{-1}\int_0^\infty\varrho_\ell(r)\,\dr$
and $\epsilon_\ell := n_\ell - [ n_\ell]\in[0,1)$.
Physically, $n_\ell$ has the interpretation as the number of electrons in
the sector of angular momentum $\ell$ corresponding to fixed magnetic
quantum number $m$ and spin $s$.
(In \cite{SiedentopWeikard1986O,SiedentopWeikard1986,SiedentopWeikard1987O},
this number is denoted by $N_{\ell,m,s}$.)
The total number of electrons in the angular momentum channel $\ell$ is therefore
$q(2\ell+1)n_\ell$. 
Note that, if $n_\ell$ is integer, then $\mathcal R(\uvarrho)\leq 0$
and the term can be dropped.
For general $n_\ell$, this term has the effect of slightly modifying
the prefactor of the term involving $\varrho_\ell^3$ in the
Hellmann--Weizs\"acker functional.
For properly chosen $\uvarrho$, it is $\mathcal{O}(Z^{5/3})$,
cf.~\cite[pp.~472--473, Proposition~3.6]{SiedentopWeikard1987O}.
In this way, one can obtain from \eqref{eq:upperboundhw} the bound
\begin{align}
  \label{eq:hwinteracting}
  E_S(N,Z) \leq E^{\rm TF}(N,Z) + \const Z^2.
\end{align}
This contains an error term of the correct order and is the
`interacting analogue' of the upper bound in Theorem \ref{hwz2}.
The proof also reveals that Weizs\"acker's correction is $o(Z^2)$
when the $\ell$-summation starts at $L=[Z^{1/12}]$,
cf.~\cite[Proposition~3.3]{SiedentopWeikard1987O}. Thus, we have
obtained an upper bound of the correct order, but not yet with the
correct coefficient. In order to obtain the correct coefficient,
we need one more modification of the above strategy.

\medskip

After this preparation we can give a brief outline of the proof in
\cite{SiedentopWeikard1987O} of the upper bound in \eqref{eq:swscottnonrel}.
The basic idea is to describe electrons close to the nucleus using
hydrogen eigenfunctions, and electrons far from the nucleus using the
Hellmann--Weizs\"acker functional.  The distinction between ``close''
and ``far'' electrons is implemented by an angular momentum cut-off
$L=[Z^\delta]$ for an appropriate $\delta\in(0,1)$.  This is suggested
by the solution of Kepler's problem, where the perihelion of a planet
grows like the square of its angular momentum.  For the hydrogen atom,
this is reflected by the explicitly computable expectation values of
powers of $|x|$ in their eigenfunctions, see, e.g., Bethe
\cite[(3.19)-(3.27)]{Bethe1933}.

In the following outline, we disregard the electron-electron repulsion.
Treating this term requires some effort, but it does not affect the Scott
correction and, therefore, ignoring this term helps to clarify the basic
steps in the proof.

\begin{enumerate}
\item Use the variational principle with a trial state that is made up of
  hydrogen eigenfunctions for angular momenta $0\leq\ell<[Z^{\frac{1}{12}}]$ and
  Macke orbitals for $\ell\geq[Z^{\frac{1}{12}}]$.
  For the precise form of the resulting upper bound, see
  \cite[(2.3)]{SiedentopWeikard1987O}.
	
\item The Macke orbitals for large $\ell$ lead to the Hellmann--Weizs\"acker
  functionals $\ce_{\ell,Z}^{\rm HW}$ (see \eqref{eq:defhwl}).
  These functionals are summed starting from $\ell\geq[Z^{1/12}]$.
  As remarked after \eqref{eq:hwinteracting}, the Weizs\"acker term is $o(Z^2)$;
  see \cite[Proposition~3.3]{SiedentopWeikard1987O}.
  Thus,
  $\sum_{\ell\geq[Z^{1/12}]}\ce_{\ell,Z}^{\rm HW}[\varrho_\ell] \leq \sum_{\ell\geq[Z^{1/12}]} \ce_{\ell,Z}^{\rm H}[\varrho_\ell^{\rm H}] + o(Z^2)$
  in terms of the Hellmann functionals $\ce_{\ell,Z}^{\rm H}$ in
  \eqref{eq:hellmannl} and their minimizers $\varrho_\ell^{\rm H}$.
  
\item The summation of the hydrogen eigenvalues up to angular momentum
  $\ell<[Z^{1/12}]$ leads to
  $\sum_{\ell<[Z^{1/12}]} \ce_{\ell,Z}^{\rm H}[\varrho_\ell^{\rm H}]$
  plus Scott's correction $\frac{q}{4}Z^2$ plus $\mathcal{O}(Z^{23/12})$;
  see \cite[Propositions~3.1 and 3.2]{SiedentopWeikard1987O}.
  
\item By \eqref{eq:hellmanntf}, the leading order of the Hellmann functional is
  the TF energy modulo $\co(Z^{5/3})$-errors;
  see \cite[Lemma~4.1]{SiedentopWeikard1987O}.

\end{enumerate}

This concludes our sketch of the upper bound in
\eqref{eq:swscottnonrel}.  We mention that Macke orbitals and the
Hellmann--Weizs\"acker functional can also be used to give a proof of
the lower bound of \eqref{eq:scottnonrelinitial}; see
\cite{SiedentopWeikard1991} and, for an exposition of the basic ideas,
\cite{SiedentopWeikard1994}. Both in the proof of the upper and lower
bound, the use of Macke orbitals is reminiscent of the use of coherent
states by Lieb in his proof of the asymptotic exactness of TF theory.

\subsection{Molecules without magnetic fields}
\label{ss:nonrelmolecule}

\subsubsection{Thomas--Fermi scale}

As in the atomic case, TF theory predicts the leading term of
$E_S(N,\uZ,\uR)$ and the density $\rho_S$ on the TF length scale
correctly.  More precisely, fix $\uR=(R_1,...,R_K)\in \R^{3K}$,
$\uZ=(Z_1,...,Z_K)\in(0,\infty)^K$,
$|\uZ|=\sum_{\kappa=1}^K Z_\kappa$, and the TF electron number
$\lambda$. It is not necessary to assume $\lambda\leq|\uZ|$. For each
$N=1,2,\ldots$, define $a_N=N/\lambda$.  In $H_{N,V}$ in
\eqref{eq:manybodySchrodinger} replace each $Z_\kappa$ by
$a_N\cdot Z_\kappa$ and each $R_\kappa$ by $a_N^{-1/3}R_\kappa$.  This
means that the nuclei come together at the rate
$a_N^{-1/3}\sim|\uZ|^{-1/3}$ as $N\to\infty$ when $N/|\uZ|$ is kept
constant.
At this point we want to emphasize that this scaling of $R_\kappa$,
which has become customary in the mathematical literature, is motivated
by mathematical considerations rather than by physical reality.
However, neither the energy nor the density is expected to be
close to the molecular ground state energy or ground state density. In
fact the two energies are expected to differ already to leading order,
since the minimal positions of the nuclei do not scale in this way.

In case $E_S(N,a_N\uZ,a_N^{-1/3}\uR)$ is an eigenvalue, let $\psi_N$ be
an associated eigenfunction with corresponding one-particle density $\rho_S$.
If not, $\rho_S$ shall denote the one-particle density associated to the
element $\psi_N$ of an approximate ground state $(\psi_N)_{N\in\N}$
in the sense of Remark \ref{remnoeigenvalue}.

By the scaling relations \eqref{eq:scalingtf} for TF theory, we have
\begin{subequations}
  \begin{align}
    \rho^{\rm TF}(\lambda,\uZ,\uR,x)
    & = a_N^{-2}\rho^{\rm TF}(N,a_N \uZ,a_N^{-1/3}\uR,a_N^{-1/3}x)\\
    E^{\rm TF}(\lambda,\uZ,\uR)
    & = a_N^{-7/3} E^{\rm TF}(N,a_N \uZ,a_N^{-1/3} \uR).
  \end{align}
\end{subequations}

We now make the connection between the full quantum problem
\eqref{eq:gsenergynonrel} with frozen nuclear positions associated to
$H_{N,V}$, and TF theory by letting the total nuclear charge $|\uZ|$ and
the electron number $N$ tend to infinity in such a way that the ionization
degree $N/|\uZ|$ is kept constant.
The following theorem due to Lieb and Simon \cite{LiebSimon1977}
generalizes Theorem \ref{quantumtfconvnonrelatom}.

\begin{theorem}[{\cite[Section~III]{LiebSimon1977}, \cite[Theorem~5]{Lieb1976}}]
  \label{quantumtfconvnonrelmolecule}
  Let $a_N=N/\lambda$, let $\uz=(z_1,...,z_K)\in(0,\infty)^K$, and let 
  $\uR=(R_1,...,R_K)\in\R^{3K}$.
  Let $H_{N,V}$ be as in \eqref{eq:manybodySchrodinger} with nuclear
  configuration $\{a_N\uz,a_N^{-1/3}\uR\}$. Then the following statements hold:

  \begin{enumerate}
  \item The quantity $a_N^{-7/3}E_S(N,a_N\uz,a_N^{-1/3}\uR)$ has a limit
    as $N\to\infty$ and this limit coincides with
    $E^{\rm TF}(\lambda,\uz,\uR)$.

  \item The scaled one-particle density $a_N^{-2}\rho_S(a_N^{-1/3}x)$
    associated to a (possibly approximate) ground state on the TF scale
    has a limit as $N\to\infty$. If $\lambda\leq|\uz|$, then the convergence
    is weakly in $L^1$ and the limit is $\rho^{\rm TF}_{|\uz|}(\lambda,x)$.
    If $\lambda>|\uz|$, then the convergence is weakly in $L_{\rm loc}^1$
    and the limit is $\rho^{\rm TF}_{|\uz|}(|\uz|,x)$.
  \end{enumerate}
\end{theorem}

As Lieb \cite[p.~560]{Lieb1976} notices:
``\emph{Note that if $\lambda>|\uz|$, then this result says that the
  surplus charge moves off to infinity  and the result is a neutral
  molecule. This means that large atoms or molecules cannot have a
  negative ionization proportional to the total nuclear charge; at
  best they can have a negative ionization which is a vanishingly
  small fraction of the total charge.}''

\begin{remark}
  Note that, if the $R_\kappa$ are kept \emph{fixed and unscaled},
  one ends up with isolated atoms in the limit $N\to\infty$,
  see Lieb \cite[pp.~559--560]{Lieb1976} for the precise statements.
\end{remark}

\subsubsection{Scott scale}

In addition to being independent of the absence or presence of the
electron-electron repulsion and of the ionization degree, one expects
the $Z^2$-correction to be the sum of the Scott corrections of the
atoms constituting the molecule, as long as the minimal internuclear
distance is not too close to the Scott scale. In particular, this is
expected for ground states of molecules where the interatomic distance
is expected to be of order one in the scaling parameter of the nuclear
charges.

The Scott conjecture for molecules with frozen nuclear positions was
first proved by Ivrii and Sigal \cite{IvriiSigal1993} using a
multiscale analysis and microlocal techniques.
Later, Solovej and Spitzer \cite{SolovejSpitzer2003} found a
proof that is partly similar to the multi-scale
analysis in \cite{IvriiSigal1993} using an interesting new
coherent states method.
Around the same time, Balodis \cite{Balodis2004} proved the Scott correction
uniformly in $K$.
The following theorem is the content of \cite[Theorem~1]{SolovejSpitzer2003}.

\begin{theorem}
  \label{scottnonrelmolecules}
  Let $\uZ=(Z_1,...,Z_K)\in(0,\infty)^K$
  and $\uR=|\uZ|^{-1/3}(r_1,...,r_K)\in\R^{3K}$
  with $\min_{k\neq\ell}|r_k-r_\ell|>r_0$
  for some $r_0>0$. Define $\uz=(z_1,...,z_K):=|\uZ|^{-1}\uZ$.
  Let $E^{\rm TF}(\uz,\ur)$ be the Thomas--Fermi energy of the
  unconstrained problem \eqref{eq:tfunconstrained} and
  $E_{S}(N,\uZ,\uR)$ be the ground state energy of $H_{|\uZ|,V}$
  with nuclear configuration $\uZ$ and $\uR$.
  Then
  \begin{align}
    E_S(|\uZ|,\uZ,\uR) = E^{\rm TF}(\uz,\ur)\cdot|\uZ|^{7/3} + \frac{q}{4}\sum_{\kappa=1}^K Z_\kappa^2 + \mathcal{O}(|\uZ|^{2-1/30})
  \end{align}
  as $|\uZ|\to\infty$, where the error term $\mathcal{O}(|\uZ|^{2-1/30})$
  besides $|\uZ|$ depends only on $Z_1,...,Z_K$ and $r_0$.
\end{theorem}




\subsection{Molecules with self-generated magnetic fields}
\label{ss:nonrelmagnetic}

In this subsection we consider molecules in the presence of a classical
magnetic field. Quoting Erd\H{o}s and Solovej \cite[p.~229]{ErdosSolovej2010}:
``\emph{External magnetic fields were
  taken into account in \cite{Liebetal1994L,Liebetal1994S} (homogeneous)
  and \cite{ErdosSolovej1997} (inhomogeneous), but subject to certain
  regularity conditions. Self-generated magnetic fields, obtained from
  Maxwell’s equation are not known to satisfy these conditions.}''
For this reason we shall consider self-generated magnetic
fields here. (See (2) in Remark \ref{remarksmagnetic} below for an explanation
of the word ``self-generated''.)

To introduce the setting precisely, let
\begin{align}
  \gA := \left\{ \uA\in L^2_{\rm loc}(\R^3:\R^3):\, \nabla\otimes \uA\in L^2(\R^3,\R^{3\times 3}),\, \mathrm{div}(\uA) = 0 \right\},
\end{align}
where all derivatives are understood in the sense of distributions and
where $\nabla\otimes \uA$ denotes the $3\times 3$ matrix of all derivatives
$\partial_i A_j$.
We set $|\nabla\otimes \uA|^2=\sum_{i,j=1}^3|\partial_i A_j|^2$ and,
with $c>0$ denoting the velocity of light, consider the magnetic field energy
\begin{align}
  \frac{c^2}{8\pi}\int_{\R^3} |(\nabla\times \uA)(x)|^2\,\dx
  = \frac{c^2}{8\pi}\int_{\R^3} |(\nabla \otimes \uA)(x)|^2\,\dx.
\end{align}
The identity here is a consequence of the Coulomb gauge.

In the non-relativistic approximation, the one-particle kinetic energy
is the \emph{magnetic Schr\"odinger} or \emph{Pauli operator}
\begin{align}
  \label{eq:magnetickinetic}
  T(\uA) = \frac12\left(-i\nabla + \uA(x)\right)^2
  \quad \text{or} \quad
  T(\uA) = \frac12\left[\usigma\cdot\left(-i\nabla + \uA(x)\right)\right]^2
\end{align}
in $L^2(\R^3:\C^2)$ depending on whether the particles are (effectively)
considered spinless, or have spin-$\frac12$.
Here $\usigma=(\sigma_1,\sigma_2,\sigma_3)$ is the vector of Pauli matrices.

The total energy of a non-relativistic molecule with charges
$\uZ=(Z_1,...,Z_K)\in(0,\infty)^K$ fixed at positions
$\uR=(R_1,...,R_K)\in\R^{3K}$ in a classical magnetic vector potential
$\uA\in\gA$ is then described by the Hamilton operator
\begin{align}
  \label{eq:manybodySchrodingerMagnetic}
  H_{N,V,\uA} = \sum_{\nu=1}^N\left( T^{(\nu)}(\uA) - V(x_\nu)\right)
  + \sum_{1\leq\nu<\mu\leq N}\frac{1}{|x_\nu-x_\mu|} + U
  \quad \text{in}\ \bigwedge_{\nu=1}^N(L^2(\R^3:\C^2)),
\end{align}
where $V$ is the electron-nucleus interaction as in \eqref{eq:defv}.

In case the kinetic energy is described by the Pauli operator, we will need
to impose an additional restriction on the quotient $|\uZ|/c^2$. Indeed,
if $|\uZ|/c^2$ is sufficiently small, then the quadratic form associated
to \eqref{eq:manybodySchrodingerMagnetic} is bounded from below uniformly
in $\uA$, see, e.g.,
\cite{Frohlichetal1986,Fefferman1995,Liebetal1995,LiebLoss2001}.
Crucially, however, stability fails, if $|\uZ|/c^2$ is too large, see
\cite{LossYau1986,ErdosSolovej2001}.

For both choices of the magnetic kinetic energy $T(A)$ and each fixed
$\uA\in\gA$, the operator $H_{N,V,\uA}$ is defined as the Friedrichs extension
of the corresponding quadratic form defined on $\cs(\R^3:\C^2)$.

For fixed magnetic potential $\uA\in\gA$ and fixed nuclear positions,
the electronic ground state energy is
\begin{align}
  \label{eq:gsmagneticnonrel}
  E(N,\uZ,\uR,\uA) & := \inf\spec(H_{N,V,\uA}).
\end{align}
The total ground state density with fixed nuclear positions
arises from minimizing this energy with respect to $\uA\in\gA$, i.e.,
\begin{align}
  \label{eq:totalgsmagneticnonrel}
  E_{S,{\rm mag}}(N,\uZ,\uR,c) & := \inf_{\uA\in\gA} \left(E(N,\uZ,\uR,\uA) + \frac{c^2}{8\pi}\int_{\R^3} |(\nabla \times \uA)(x)|^2\,\dx\right).
\end{align}

\begin{remarks}
  \label{remarksmagnetic}
  (1) In \eqref{eq:totalgsmagneticnonrel} it suffices to minimize over
  all compactly supported $\uA\in\gA$.

  (2) As remarked in \cite[p.~231]{ErdosSolovej2010}, the Euler--Lagrange
  equation that arises from minimizing 
  $\left\langle\Psi,\left(H_{N,V,\uA}+\frac{c^2}{8\pi}\int_{\R^3} |(\nabla \times \uA)(x)|^2\,\dx \right)\Psi\right\rangle$
  over $\Psi$ and $\uA$, corresponds to the stationary version of the
  coupled Maxwell--Pauli system, i.e., the eigenvalue problem
  $H_{N,V,\uA}\Psi=E_{S,{\rm mag}}(N,\uZ,\uR,c)\Psi$ together with the
  Maxwell equation for the magnetic field, i.e.,
  $\nabla\times \uB=4\pi c^{-2}\underline{J}_\Psi$;
  here $\underline{J}_\Psi$ is the current of the wave function $\Psi$.
  This explains why $\uB$ is called a self-generated magnetic field.
\end{remarks}

\subsubsection{Thomas--Fermi scale}

On the Thomas--Fermi scale it turns out -- as the semiclassical picture suggests --
that the magnetic field does not change the leading order of the energetic expansion
as $Z\to\infty$.
To make this precise, define
\begin{align}
  \label{eq:totalgsmagneticnonreltf}
  E_{S,{\rm mag}}(\uZ,\uR,c) & := \inf_{N\in\N}E_{S,{\rm mag}}(N,\uZ,\uR,c),\\
  E_{S,{\rm nonmag}}(\uZ,\uR,c) & := \inf_{N\in\N} \inf\spec(H_{N,V,0}).
\end{align}
Erd\H{o}s and Solovej \cite{ErdosSolovej2010} proved the following
theorem.

\begin{theorem}[{\cite[Theorem~1.1]{ErdosSolovej2010}}]
  \label{tfnonrelmagn}
  Suppose that $T(\uA)$ is either the Pauli or the magnetic Schr\"odinger
  operator and assume (for simplicity) $Z:=Z_1=Z_2=...=Z_K$ and
  $|R_i-R_j|\geq c_1 Z^{-1/3}$ for all $i\neq j$. Then there is a
  positive constant $\kappa_0$ such that if $Z/c^2\leq\kappa_0$, then
  \begin{align}
    E_{S,{\rm nonmag}}(\uZ,\uR,c)
    \geq E_{S,{\rm mag}}(\uZ,\uR,c)
    \geq E_{S,{\rm nonmag}}(\uZ,\uR,c)- c_2 Z^{\frac73-\frac{1}{63}}.
  \end{align}
\end{theorem}

\subsubsection{Scott scale}

The following theorem due to Erd\H os et
al.~\cite{Erdosetal2012Sc,Erdosetal2012S} characterizes the Scott
correction in the presence of a self-generated magnetic field.

\begin{theorem}[{\cite[Theorem~1.1]{Erdosetal2012Sc}}]
  \label{scottnonrelmagn}
  Suppose that $T(\uA)$ is either the Pauli or the magnetic Schr\"odinger
  operator.
  Let $\uz=(z_1,...,z_K)\in(0,\infty)^K$ with $\sum_{\kappa=1}^K z_k=1$
  and $\ur=(r_1,...,r_K)\in\R^{3K}$ with
  $\min_{k\neq\ell}|r_k-r_\ell|\geq r_0$ for some $r_0>0$ be given.
  Let $\uZ=(Z_1,...,Z_K)=|\uZ|(z_1,...z_K)\in(0,\infty)^K$ for some $|\uZ|>0$
  and $\uR=|\uZ|^{-1/3}\ur$ be the charges and the positions of the nuclei
  of the operator $H_{N,V,\uA}$ in \eqref{eq:manybodySchrodingerMagnetic}.
  Then there is a universal (independent of $\uz,\ur,K$), continuous,
  monotone non-increasing function $S:(0,\kappa_0]\to\R$
  with some universal $\kappa_0>0$ and with
  $\lim_{\kappa\searrow0}S(\kappa)=\frac14$ such that, as
  $|\uZ|=\sum_{\kappa=1}^K Z_\kappa\to\infty$ and $c\to\infty$ with
  $\max_\kappa 8\pi Z_\kappa/c^2\leq\kappa_0$, one has
  \begin{align}
    E_{S,{\rm mag}}(|\uZ|,\uZ,\uR,c) = E^{\rm TF}(\uz,\ur)|\uZ|^{7/3} + 2\cdot |\uZ|^2 \sum_{\kappa=1}^K z_\kappa^2 \cdot S(8\pi Z_\kappa/c^2) + o(|\uZ|^2).
  \end{align}
\end{theorem}


\begin{remarks}
  (1) The theorem is independent of the existence and uniqueness
  (modulo gauge freedom) of the minimizer $\uA$. In fact, it is not
  clear whether the infimum is attained at $\uA=0$ or at a non-trivial
  magnetic field.

  (2) The threshold $\kappa_0$ for which the assertion of Theorem
  \ref{scottnonrelmagn} is shown to hold is less than the number
  $\kappa_{\rm cr}$ above which $H_{N,V,\uA}$ fails to be bounded from
  below.

  (3) The theorem does not assert that
  $S(\kappa)$ is \emph{strictly} decreasing, although this is believed
  to be the case. In fact, it is conceivable that $S(\kappa)$ is constant
  equal to $1/4$ for all $\kappa$ up to the critical value $\kappa_{\rm cr}$
  beyond which it is minus infinity.
\end{remarks}

\begin{remark}
  Theorem \ref{tfnonrelmagn} concerning the leading order holds also
  in the case where the $\uA$-field is quantized
  \cite{ErdosSolovej2010}. So far, the Scott correction for quantized
  $\uA$-fields -- both in the non-relativistic setting of Theorem
  \ref{scottnonrelmagn} and in the relativistic one of Theorem
  \ref{scottchandrasekharmoleculemagnetic} below -- can only be proved
  with a low ultraviolet cutoff of the magnetic field which corresponds
  to a length scale that is longer than the Scott scale, i.e., would
  be of limited physical meaning.
\end{remark}

We emphasize that all results in this subsection concern the
energy. We are not aware of results concerning the density.

\section{Relativistic Coulomb systems}
\label{s:relativistic}

In this section, we discuss relativistic models of large Coulomb systems
and begin with an overview of the relevant underlying one-particle operators.

\subsection{One-particle operators}
\label{ss:1poperators}

We first introduce the relativistic one-particle operators that will
later be used to construct the many-particle operators that we are
mostly interested in.  We state conditions on the coupling constants
of the Coulomb potential for which the operators can be defined and
recall some of their spectral properties.  More detailed treatments
are contained, e.g., in the textbooks by Balinsky and Evans
\cite{BalinskyEvans2011} and Thaller \cite{Thaller1992}, as well as in
the paper by Matte and Stockmeyer \cite{MatteStockmeyer2010} and the
references therein.

\subsubsection{Chandrasekhar operator}

\subsubsection*{Definition}

The Chandrasekhar operator is the simplest relativistic operator
discussed here.  In the literature this operator is sometimes referred
to as Herbst operator or pseudorelativistic operator.  Its origins can
be traced back at least to Chandrasekhar in the context of
(in)stability of neutron stars \cite{Chandrasekhar1931} (see also
\cite{LiebThirring1984,LiebYau1987}). The mathematical investigation
of this operator started with the work of Herbst \cite{Herbst1977};
see also Weder \cite{Weder1974} for electric potentials
$V(x)=|x|^{-\beta}$ with $\beta\in(0,1)$.  The operator is defined as
the Friedrichs extension of the quadratic form -- whenever it is
bounded from below (see \eqref{eq:kato}) -- associated to
\begin{align}
  \label{eq:defchandraz}
  C_{c,Z}^{H} := \sqrt{-c^2\Delta+c^4}-c^2-\frac{Z}{|x|}
  \quad \text{in}\ L^2(\br^3:\bc)
\end{align}
with form domain being the Schwartz space $\cs(\br^3:\bc)$.
In the following, we abbreviate the fractional Laplace
by $|p|:=\sqrt{-\Delta}$.

\subsubsection*{Scaling}

The operator has a natural length scale, namely $c^{-1}$.
Indeed, scaling $x\mapsto x/c$ and writing $\gamma:=Z/c$ shows that
$C_{c,Z}^{H}$ is unitarily equivalent to $c^2 C_{1,\gamma}^{H}=:c^2C_\gamma^H$
with
\begin{align}
  \label{eq:defchandra}
  C_\gamma^{H} := \sqrt{-\Delta+1}-1-\frac{\gamma}{|x|}
  \quad \text{in}\ L^2(\br^3:\C).
\end{align}

\subsubsection*{Kato's inequality}

The sharp Hardy--Kato--Herbst inequality -- for short Kato's
inequality -- states that
\begin{align}
  \label{eq:kato}
  \frac2\pi\int_{\R^3} \frac{|u(x)|^2}{|x|}\,\dx
  \leq \int_{\R^3}|\xi||\hat u(\xi)|^2\,\rd\xi, \quad u\in \cs(\R^3),
\end{align}
where $\hat u(\xi):=(2\pi)^{-3/2}\int_{\R^3}\me{-ix\cdot\xi}u(x)\,\dx$;
see Kato \cite[Chapter~5, Formula~(5.33)]{Kato1966} (without proof)
and Herbst \cite[Theorem~2.5]{Herbst1977}. It follows from Kato's
inequality and the inequalities $|p|\geq\sqrt{p^2+1}-1\geq|p|-1$ that
the quadratic form associated to $C_\gamma^H$ is bounded from below
if and only if $\gamma\leq\gamma_C$ with
\begin{align}
  \label{eq:defgammac}
  \gamma_C := \frac2\pi.
\end{align}
In fact, Raynal et al.~\cite{Raynaletal1994}
showed that the form is strictly greater than $-1$, even if $\gamma=2/\pi$.
Numerical evidence for this fact had been provided by Hardekopf and Sucher
\cite{HardekopfSucher1985}.

\subsubsection*{Domain considerations}
The quadratic form domain of $C_\gamma^H$ is $H^{1/2}(\R^3)$ when $\gamma<2/\pi$.
For $\gamma=2/\pi$ the form domain is the closure of $\cs(\R^3)$
with respect to the norm $(\langle u,C_\gamma^H u\rangle+\|u\|^2)^{1/2}$.
In analogy to the local case, we believe that there are functions in the
form domain of $C_{2/\pi}^H$ for which both sides of Kato's inequality are
infinite and, therefore, that this form domain strictly contains $H^{1/2}(\R^3)$.
For domain considerations, see also Le Yaouanc et al.~\cite{LeYaouancetal1997}.

\subsubsection*{Decomposition into angular momenta}

The spherical symmetry allows to decompose $C_\gamma^H$ into angular
momentum channels. Decomposing $L^2(\R^3)$ into the direct sum
$$
L^2(\br^3)=L^2(\br_+,r^2\,\dr)\otimes L^2(\bs^2,\rd\omega)=\bigoplus_{\ell\in\N_0}\ch_\ell,
$$
induces the decomposition
\begin{align}
  C_\gamma^H = \bigoplus_{\ell\in\N_0} \tilde C_{\ell,\gamma}^H \otimes \one_{K_\ell},
\end{align}
where $\tilde C_{\ell,\gamma}^H$ acts in $L^2(\R_+,r^2\,\dr)$.
Here $\ch_\ell=L^2(\br_+,r^2\,\dr)\otimes K_\ell$
with $K_\ell$ being the eigenspace associated to the $\ell$-th
eigenvalue $\ell(\ell+1)$ of the Laplace--Beltrami operator on $\bs^2$.
Defining $U:L^2(\R_+,r^2\,\dr)\to L^2(\R_+,\dr)$ by
$(Uf)(r)=rf(r)$ for $f\in L^2(\R_+,r^2\,\dr)$,
we may introduce
\begin{align}
  \label{eq:defcellh}
  C_{\ell,\gamma}^H := U \tilde C_{\ell,\gamma}^H U^* \equiv \sqrt{-\frac{\rd^2}{\dr^2}+\frac{\ell(\ell+1)}{r^2}+1}-1-\frac{\gamma}{r} \quad \text{in}\ L^2(\R_+,\dr).
\end{align}

\subsubsection*{Fourier--Bessel transform}

The kinetic energy operator
\begin{align}
  \label{eq:defcl}
  C_\ell := \sqrt{-\frac{\rd^2}{dr^2}+\frac{\ell(\ell+1)}{r^2}+1}-1
  \quad \text{in}\ L^2(\R_+,\dr)
  \end{align}
can be diagonalized by the Fourier--Bessel
transform $\Phi_\ell:L^2(\R_+,\dr)\to L^2(\R_+,\dr)$. For
$u\in \cs(\R_+)$ it acts as
\begin{align}
  u\mapsto (\Phi_\ell u)(k) := \int_0^\infty \dr\, \sqrt{kr}J_{\frac{d-2}{2}+\ell}(kr)u(r),
  \quad k\in\R_+.
\end{align}
Note that $\Phi_\ell=\Phi_\ell^*$ is unitary on $L^2(\R_+,\dr)$.
Just as the Fourier transform diagonalizes translation invariant operators,
the Fourier--Bessel transform diagonalizes translation invariant,
spherically symmetric operators when restricted to a specific angular
momentum channel $\ell$.
Define the operator $p_\ell$ in $L^2(\R_+,\dr)$ by the equality
\begin{align}
  \langle u, p_\ell u\rangle_{L^2(\R_+,\dr)}
  = \langle f,(-\Delta)^{1/2}f\rangle_{L^2(\R^3)} 
\end{align}
for any $f(x)=|x|^{-1}u(|x|)Y_{\ell,m}(x/|x|)$ with $u\in \cs$.
Formally, we have
$$
p_\ell = \sqrt{-\frac{\rd^2}{\dr^2}+\frac{\ell(\ell+1)}{r^2}}.
$$
Then, for any spectral multiplier $F\in L_{\rm loc}^1(\R_+)$
of $p_\ell$, one has
\begin{align}
  (\Phi_\ell (F(p_\ell) f))(k) = F(k)\cdot(\Phi_\ell f)(k), \quad k>0,
\end{align}
and (in weak sense) the kernel of $F(p_\ell)$ is
\begin{align}
  \label{eq:fourierbessel1}
  (F(p_\ell))(r,s)
  = \int_0^\infty \dk\, F(k) \sqrt{kr}J_{\frac{d-2}{2}+\ell}(kr)\sqrt{ks}J_{\frac{d-2}{2}+\ell}(ks),
  \quad r,s>0.
\end{align}

\subsubsection*{Hydrogen eigenvectors and density}
Let $\psi_{n,\ell,m}^C$ denote the $L^2(\R^3)$-normalized eigenvectors
of $C_\gamma^H$. Due to the spherical symmetry of $C_\gamma^H$ we have
\begin{align}
  \psi_{n,\ell,m}^C(x) = \frac{\psi_{n,\ell}^C(|x|)}{|x|}Y_{\ell,m}\left(\frac{x}{|x|}\right),
\end{align}
where $\psi_{n,\ell}^C$ are the $L^2(\R_+,\dr)$-normalized
eigenvectors of $C_{\ell,\gamma}^H$. 
The radial, one-dimensional hydrogenic density in angular momentum channel
$\ell\in\N_0$ is
\begin{align}
  \varrho_{\ell,C}^H(r) := q(2\ell+1)\sum_{n\geq0}|\psi_{n,\ell}^C(r)|^2, \quad r>0,
\end{align}
and the (spherically symmetric) total, three-dimensional density is
\begin{align}
  \rho_C^H(x) = \frac{1}{4\pi |x|^2}\sum_{\ell\geq0}\varrho_{\ell,C}^H(|x|),
  \quad x\in\R^3.
\end{align}
These quantities are indeed well defined as the following
theorem due to \cite[Theorem~1.4]{Franketal2020P} shows.
To state it, we define
\begin{align}
  \label{eq:defsigmagammalequalzero}
  \begin{split}
    \Phi: (-1,1] & \to (-\infty,2/\pi]\\
    \sigma & \mapsto \frac{2\Gamma \left(\frac{1}{2} (3-\sigma)\right) \Gamma\left(\frac{1}{2} (1+\sigma)\right)}{\Gamma \left(\frac{\sigma}{2}\right) \Gamma \left(\frac{1}{2} (2-\sigma)\right)}
    = (1-\sigma)\tan\left(\frac{\pi\sigma}{2}\right).
  \end{split}
\end{align}
This is a monotone increasing function, which satisfies
$\lim_{\sigma\searrow-1}\Phi(\sigma)=-\infty$
and $\Phi(0)=0$, and whose maximal value is $2/\pi = \Phi(1)$.
Consequently, for any $\gamma\in[0,2/\pi]$ there is a unique
$\sigma_\gamma\in[0,1]$
such that $\Phi(\sigma_\gamma)=\gamma$ for $\gamma\in[0,2/\pi]$, i.e.,
\begin{align}
  \label{eq:defsigmagamma}
  \sigma_\gamma = \Phi^{-1}(\gamma)\in[0,1] \quad \text{for}\ \gamma\in[0,2/\pi].
\end{align}

\begin{theorem}[{\cite[Theorem~1.4]{Franketal2020P}}]
  \label{existencerhoh}
  Let $1/2<s\leq 3/4$ if $0<\gamma<(1+ \sqrt 2)/4$ and
  $1/2<s<3/2-\sigma_\gamma$ if $(1+\sqrt 2)/4\leq\gamma<2/\pi$.
  Then for all $\ell\in\N_0$ there is a constant $A_{s,\gamma}>0$
  such that for all $r\in\R_+$ one has
  \begin{align}
    \label{eq:existencerhoh1}
    \begin{split}
      \varrho_{\ell,C}^H(r)
      & \leq q\cdot A_{s,\gamma} \left( \ell+\tfrac 12 \right)^{-4s+1}\\
      & \qquad \times \left[\left(\frac{r}{\ell+\tfrac12}\right)^{2s-1}\one_{\{r\leq\ell+\frac12\}}+\left(\frac{r}{\ell+\tfrac12}\right)^{4s-1}\one_{\{\ell+\frac12< r\leq(\ell+\frac12)^2\}}\right.\\
      & \qquad\qquad \left.+ \left(\ell+\tfrac12\right)^{4s-1}\one_{\{r>(\ell+\frac 12)^2\}}\right].
    \end{split}
  \end{align}
  Moreover, for any $\epsilon>0$, there are constants
  $A_\gamma,A_{\gamma,\epsilon}>0$ such that for all $r\in\R_+$ one has
  \begin{align}
    \label{eq:existencerhoh2}
    \rho_C^H(r) \leq 
    \begin{cases}
      q\cdot A_\gamma \, r^{-\frac32} & \text{if}\ 0<\gamma<(1+\sqrt 2)/4 , \\
      q\cdot A_{\gamma,\epsilon} \left( r^{-2\sigma_\gamma-\epsilon} \one_{\{r\leq 1\}} + r^{-\frac32} \one_{\{r>1\}} \right)
      & \text{if}\ (1+\sqrt 2)/4\leq \gamma<2/\pi.
    \end{cases}
  \end{align}
\end{theorem}

\begin{remarks}
  (1) The proof of \eqref{eq:existencerhoh1} uses Bessel kernel bounds for
  $C_{\ell,\gamma}^H$, i.e., bounds for the integral kernel
  \begin{align}
    (C_{\ell,\gamma}^H+a_\ell)^{-t}(r,s)
  \end{align}
  with $a_\ell=a_\gamma(\ell+1/2)^{-2}$, $r,s>0$, and
  $t\in(1,\min\{3-2\sigma_\gamma,3\})$.
  For $\gamma<1/2$ the Hardy potential $\gamma/r$ is an
  \emph{operator perturbation} for $p_{\ell=0}$ by Hardy's inequality.
  Hence, the proof in this case effectively only uses Bessel kernel bounds
  for $C_\ell$ (recall \eqref{eq:defcl}), which can be obtained using the
  Fourier--Bessel transform.
  On the other hand, if $\gamma\geq1/2$, a comparison between powers of
  $C_{\ell,\gamma}^H$ and $p_\ell$ is not straightforward. In this case, the
  proof of \cite[Theorem~1.4]{Franketal2020P} used the comparison result in
  \cite{Franketal2021} (see Theorem \ref{fms} later). This is the reason
  for the assumption $s<3/2-\sigma_\gamma$ and the appearance of $\epsilon>0$
  in \eqref{eq:existencerhoh2}.

  (2) Instead of comparing Bessel kernels of $C_{\ell,\gamma}^H$ and $p_\ell$,
  one can compare the Bessel kernels of $C_{\ell,\gamma}^H$ and $p_\ell-\gamma/r$.
  For $\ell=0$, the latter can immediately be derived using the spectral
  theorem and recent heat kernel bounds for $|p|-\gamma/|x|$ by
  Bogdan et al.~\cite{Bogdanetal2019}.
  Since for $\ell\geq1$ the Hardy potential is again an operator
  perturbation for $\gamma<1/2$ (by Hardy's inequality), Bessel kernel
  bounds for $p_\ell-\gamma/r$ are similar to those for $p_\ell$.
  We expect that this strategy allows to remove the $\epsilon>0$ in
  \eqref{eq:existencerhoh2} when $\gamma>(1+\sqrt2)/4$.
  It is an open problem to prove that the behavior of $\varrho_{\ell,C}^H$
  and $\rho_C^H$ at $r=0$ for $\gamma>(1+\sqrt2)/4$ is optimal.

  (3) We do not expect Bessel kernel bounds for $p_\ell-\gamma/r$ to
  yield precise (probably $\gamma$-dependent) bounds for $\varrho_{\ell,C}^H$
  at the origin in the cases $\ell=0$ and $\gamma\leq1/2$, and $\ell\geq1$
  and $\gamma\leq2/\pi$, because the Bessel kernel bounds for
  $p_\ell-\gamma/r$ are similar to those for $p_\ell$.
  
  (4) The appearance of $\gamma=(1+\sqrt 2)/4$ is technical and comes
  from the restriction $\sigma\leq 3/4$ together with the fact that
  $\sigma_{(1+\sqrt 2)/4}=3/4$.  
\end{remarks}

\subsubsection*{Larger coupling constants for higher angular momenta}

Using a relative to the Fourier--Bessel transform, namely the Mellin
transform, Le Yaouanc et al.~\cite{LeYaouancetal1997}
showed that the largest admissible coupling constant associated to
$C_{\ell,\gamma}^H$ increases as $\ell$ increases.
Independently, Yafaev \cite[(2.4), (2.26)]{Yafaev1999} gave an
alternative proof of this fact and showed
\begin{align}
  \label{eq:yafaev}
  \int_0^\infty \dk\, k|(\Phi_\ell u)(k)|^2
  \geq \frac{2 \Gamma \left(\frac{1}{4} (4+2\ell)\right)^2}{\Gamma \left(\frac{1}{4} (2+2\ell)\right)^2}\int_0^\infty \frac{|u(r)|^2}{r}\,\dr.
\end{align}
For $d=3$ and $\ell=1$, the largest admissible coupling constant
is $\pi/2$, which compares to the critical value $2/\pi$ when $\ell=0$.
We note that similar inequalities hold in dimensions $d$ other then three
and powers $\alpha\in(0,\min\{2,d\})$ of the square root of the Laplacian.

\subsubsection*{Ground state transform}
An alternative representation of the operator $(-\Delta)^{1/2}-\gamma/|x|$
in $L^2(\R^3)$ proceeds via the \emph{ground state transform}.
To that end, recall \eqref{eq:defsigmagammalequalzero}.
The ground state transform makes use of the fact that the radial function
$x\mapsto|x|^{-\sigma}$ is a \emph{(generalized) ground state}
for $|p|-\Phi(\sigma)|x|^{-\alpha}$. It states that
\begin{align}
  \label{eq:orggstransform}
  \langle u,(|p|-\Phi(\sigma)|x|^{-1})u\rangle
  = \frac{1}{2\pi^2}\int_{\R^3}\frac{|v(x)-v(y)|^2}{|x-y|^{4}}\cdot (|x||y|)^{-\sigma}\,\dx\,\dy
\end{align}
where $\sigma\in[0,1]$, $u(x)=|x|^{-\sigma}v(x)$ with
$v\in \cs(\R^3\setminus\{0\})$, cf.~\cite[Proposition~4.1]{Franketal2008H}.
Note that for $\alpha=2$ the ground state transform has been known long
before, see, e.g., \cite[p.~169]{ReedSimon1975} for a textbook treatment
when $d=3$.
For a further study of the ground state transform in the
fractional case, see \cite{FrankSeiringer2008}.

\subsubsection*{Spectrum}
Although the eigenvalues $\lambda_{Z,n,\ell}$ ($n\in\N_0$) of $C_{c,Z}^{H}$
are not explicitly known, the inequality $\sqrt{p^2+1}-1\leq p^2/2$
and the lower bound of \cite[Theorem~2.2]{Franketal2009} imply the
inequalities
\begin{align}
  -\frac{Z^2}{2(n+\ell+1)^2} \geq \lambda_{Z,n,\ell} \geq -\const\cdot\frac{Z^2}{(n+\ell+1)^2},
\end{align}
where the constant in the second inequality can be chosen independently of $\gamma\in[0,2/\pi]$.
The expression on the left side is just the $n$-th eigenvalue of the
hydrogen operator \eqref{eq:defhydrogenoperator} in angular momentum
channel $\ell$. Thus, although the relativistic eigenvalues are smaller
than the non-relativistic ones, their magnitudes in $Z$ are the same and
many of their summability properties with respect to $n$ are similar.

Finally, the spectrum of $C_\gamma^H$ in $[0,\infty)$ is purely absolutely
continuous, the singular continuous spectrum is empty, and there are
no embedded eigenvalues \cite[Theorem~2.3]{Herbst1977}.
In particular, there is no zero eigenvalue, a fact that we will use later.

\subsubsection*{Physical shortfalls}

Although $C_\gamma^H$ is mathematically well understood, it has a number of
physical deficits.
For instance, the restriction $\gamma\leq2/\pi$ implies that only atoms with
nuclear charge $<88$ can be described.
Moreover, the predicted ground state energies for heavy atoms are much too low.
This can already be anticipated by comparing the ground state energies for
hydrogen with $c=1$ and coupling $\gamma$ close to $2/\pi$, which are
$\approx-0.5$ in the Chandrasekhar model \cite[p.~106]{Raynaletal1994}
and $\approx-0.06$ (cf.~\eqref{eq:eigenvalue}) in the Dirac model, respectively.

Although the model is unsuitable for the quantitative description
of systems with strong attractive external Coulomb forces, Chandrasekhar
\cite{Chandrasekhar1931} used it successfully for attractive two-particle
Coulomb forces in his Nobel prize winning estimate on the mass necessary to
collapse a star to a white dwarf.
Despite its mathematical simplicity and its success in correctly describing
some qualitative features of relativistic Coulomb systems,
it is desirable to examine models that also lead to quantitatively correct
predictions for ground state properties.
Such models are, e.g., based on the Coulomb--Dirac operator, which we discuss next.

\subsubsection{Coulomb--Dirac operator}

\subsubsection*{Free Dirac operator}
In 1928 Dirac \cite{Dirac1928,Dirac1928II} derived a Lorentz invariant
equation of motion for quantum mechanical particles with spin moving
in an external electromagnetic field, the so-called Dirac equation.
We refer to \cite{Bethe1933,Thaller1992} for comprehensive treatments.
For a free particle, the equation reads
\begin{align}
  \label{eq:Diracgleichung}
 i\partial_t\psi(t,x)=\left(-ic\ualpha\cdot\nabla+\beta c^2\right)\psi(t,x)
\end{align}
with the Dirac matrices $\ualpha=(\alpha_1,\alpha_2,\alpha_3)$,
\begin{align*}
  \alpha_j=\left(\begin{array}{cc}
             0_{\bc^2} & \sigma_j\\
             \sigma_j & 0_{\bc^2}
           \end{array}\right),
\end{align*}
the Pauli matrices $\usigma:=(\sigma_1,\sigma_2,\sigma_3)$, and
$\beta=\diag(1,1,-1,-1)$.
The operator on the right side of \eqref{eq:Diracgleichung} is called the
free Dirac operator. It acts on states $\psi_t(x)\in\bc^4$, called
Dirac spinors. The underlying Hilbert space is $L^2(\br^3:\bc^4)$.
The domain on which the free Dirac operator can be realized as a self-adjoint
operator is $H^1(\br^3:\bc^4)$.
The Foldy--Wouthuysen transform $U_{\rm FW}$ allows to perform a
block diagonalization, whereby the free Dirac operator takes the form
\begin{align}
  U_{\mathrm{FW}}\left(-ic\ualpha\cdot\nabla+\beta c^2\right)U_{\mathrm{FW}}^*
  =\left(\begin{array}{cc}
           \sqrt{-c^2\Delta+c^4} & 0\\
           0 & -\sqrt{-c^2\Delta+c^4}
         \end{array}\right).
\end{align}
This shows that the spectrum equals $(-\infty,-c^2]\cup[c^2,\infty)$.
Physically, this means that states can possess ``negative energy'' and
that there is an infinitely deep energy reservoir, the so-called Dirac
sea. By adding electromagnetic fields and the charge conjugation operator,
one can interpret states with negative energy as ``antiparticles'',
i.e., particles with same mass but opposite charge. Such particles are
called positrons.

\subsubsection*{Self-adjoint extensions of the Coulomb--Dirac operator}
The one-particle Dirac operator describing the hydrogen atom can initially
be defined on $\cs(\R^3:\C^4)$ and is formally given by the differential
operator
\begin{align}
  D_{c,Z}^H := -ic\ualpha\cdot\nabla + c^2\beta-\frac{Z}{|x|} \quad \text{in}\ L^2(\br^3:\bc^4).
\end{align}
Scaling $x\mapsto x/c$ and writing $\gamma:=Z/c$ shows that
$D_{c,Z}^H$ is unitarily equivalent to
\begin{align}
  \label{eq:defcoulombdirac}
  c^2\left[-i\ualpha\cdot\nabla + \beta-\frac{\gamma}{|x|}\right]
  =: c^2 D_{1,\gamma}^H \equiv c^2 D_\gamma^H.
\end{align}

Weidmann \cite{Weidmann1971} showed that $D_\gamma^H$ is essentially
self-adjoint on $\cs(\R^3\setminus\{0\}:\C^4)$ if and only
if $|\gamma|<\sqrt3/2$, see also \cite[Theorem~4.4]{Thaller1992}.
For $\gamma\in[\sqrt3/2,1]$ there is a ``distinguished'' (sometimes called
``physically relevant'') self-adjoint extension of $D_\gamma^H$.
For $\gamma\in(\sqrt3/2,1)$ it was established by Schmincke \cite{Schmincke1972},
W\"ust \cite{Wuest1975}
(see also Kalf, Schmincke, Walter, and W\"ust \cite{Kalfetal1974} for a
review of these results),
Nenciu \cite{Nenciu1976}, and Klaus and W\"ust \cite{KlausWust1978}.
According to Schmincke and W\"ust, this realization stands out by the property
that all states in the domain of the Coulomb--Dirac operator had finite potential
energy. On the other hand, Nenciu's realization was distinguished by the fact
that states had finite kinetic energy.
Klaus and W\"ust showed that both realizations coincide, and that the essential
spectrum is $(-\infty,1]\cup[1,\infty)$, see \cite{KlausWust1979}
(or \cite[p.~117]{Thaller1992} for a textbook treatment).
In summary, the domain $\dom(D_\gamma^H)$ of the distinguished realization
satisfies $H^1(\br^3:\bc^4)\subset \dom(D_\gamma^H)\subset H^{1/2}(\br^3:\bc^4)$,
and the quadratic form domain is $H^{1/2}(\br^3:\bc^4)$. In particular,
the expectation values of both kinetic and potential energy are finite in
$\dom(D_\gamma^H)$; this motivates the term ``physically relevant extension''.
With the help of the sharp Hardy--Dirac inequality~\cite{Dolbeaultetal2000},
\begin{align}
  \label{eq:hardydirac}
  \int_{\R^3}\frac{|\phi(x)|^2}{|x|}\,\dx
  \leq \int_{\R^3}\left(\frac{|(\usigma\cdot\nabla\phi)(x)|^2}{1+|x|^{-1}}+|\phi(x)|^2\right)\,\dx,
  \quad \phi\in H^1(\R^3:\C^2),
\end{align}
Esteban and Loss \cite{EstebanLoss2007} constructed a distinguished
self-adjoint extension for $\gamma=1$.
States in the domain of this operator need not have finite kinetic and
potential energy separately.
We remark that similar results in two dimensions (where the critical
coupling is $\gamma=1/2$) were proved by Warmt \cite{Warmt2011}.
In this review we will only focus on the three-dimensional case and $\gamma<1$.

\subsubsection*{Partial wave analysis}

Since $D_\gamma^H$ is spherically symmetric, one can, analogously to
the angular momentum decomposition for spherically symmetric scalar
operators, perform a partial wave decomposition, see, e.g.,
\cite{Evansetal1996}, \cite[Section~2.1]{BalinskyEvans2011},
\cite[Sections~4.6.3-4.6.5]{Thaller1992}, and
\cite[Appendix A]{MerzSiedentop2020}.
We begin by observing that those of the spherical $\C^2$-spinors
\begin{equation}
  \label{2.6}
  \Omega_{\ell,m,s}(\omega) :=
  \begin{pmatrix} 2s\sqrt{\frac{\ell+\frac12+2sm}{2\ell+1}}
    Y_{\ell,m-\frac12}(\omega)\\ \sqrt{\frac{\ell+\frac12-2sm}{2\ell+1}}
    Y_{\ell,m+\frac12}(\omega)
  \end{pmatrix} 
\end{equation}
with $\ell=0,1,2,...$ and $m=-\ell-\frac12,...,\ell+\frac12$, that
do not vanish, form an orthonormal basis of $L^2(\mathbb{S}^2:\C^2)$,
see, e.g., \cite[(7)]{Evansetal1996}.
Moreover, they are joint eigenfunctions of $L^2$, $J^2$ ($J=L+S$ being
the total angular momentum), and $J_3$ with respective eigenvalues
$\ell(\ell+1)$, $(\ell+s)(\ell+s+1)$, and $m$.

Introducing the spin-orbit operator $K=\beta(J^2-L^2+1/4)$, there is
an orthonormal basis of eigenvectors $\Phi^\sigma_{\kappa,m}$ of
$L^2(\bs^2:\C^4)$ such that
$J^2\Phi^\sigma_{\kappa,m}=j_\kappa(j_\kappa+1)\Phi^\sigma_{\kappa,m}$,
$J_3\Phi^\sigma_{\kappa,m}=m\Phi^\sigma_{\kappa,m}$, and
$K\Phi^\sigma_{\kappa,m}=\kappa\Phi^\sigma_{\kappa,m}$ with
the total angular momentum $j_\kappa$ and orbital
angular momentum $\ell_\kappa$ defined as
\begin{equation}
  \label{l}
  j_\kappa:=|\kappa|-\tfrac12\ \text{and}\
  \ell_\kappa:= j_\kappa-\tfrac12\sgn(\kappa)= |\kappa|-\theta(\kappa),
\end{equation}
the magnetic quantum numbers $m\in\{-j_\kappa,...,j_\kappa\}$,
the spin-orbit coupling
\begin{equation}
  \label{zp}
  \kappa\in\zp:=\gz\setminus\{0\},
\end{equation}
and $\sigma\in\{+,-\}$.
A standard choice is
\begin{equation}\label{fi}
  \Phi^+_{\kappa,m} :=  \left(\begin{array}{c}
          {\rm i}\sgn(\kappa)\Omega_{\ell_\kappa,m,\frac12\sgn(\kappa)}\\
          0
                            \end{array}
                          \right), \
                          \Phi^-_{\kappa,m}:=\left(\begin{array}{c}
                              0\\
                              -\sgn(\kappa)\Omega_{\ell_\kappa+\sgn(\kappa),m,-\frac12\sgn(\kappa)}
                            \end{array}
                          \right).    
\end{equation}
Using these spinors, we introduce the spaces
\begin{align}
  \label{a3}
  \gh_{\kappa,m} :=& \mathrm{span}\{x\mapsto \tfrac{f^+(|x|)}{|x|} \Phi_{\kappa,m}^+(\tfrac x{|x|}) + \tfrac{f^-(|x|)}{|x|}\Phi_{\kappa,m}^-(\tfrac x{|x|}):f^+,f^-\in L^2(\R_+)\},\\
  \label{a2}
  \gh_\kappa:=& \bigoplus_{m=-j_\kappa}^{j_\kappa}\gh_{\kappa,m}.
\end{align}
These spaces form an orthogonal decomposition of $L^2(\R^3:\C^4)$.
Note that Dirac operators with radial potentials leave the spaces
$\gh_{\kappa,m}$ invariant.
To see this, let $f\in\gh_{\kappa,m}\cap H^{1}(\R^3:\C^4)$ and
$g\in\gh_{\kappa',m'}\cap H^{1}(\R^3:\C^4)$. Then one has
\begin{equation}
  \label{eq:radialdirac}
  \begin{split}
    & \langle f,D_\gamma^H g\rangle_{L^2(\R^3:\C^4)}\\
    & \quad = \left\langle \left(\begin{array}{c}
                      f^+\\ f^-
                    \end{array}
    \right),\left(\begin{array}{cc}
                   1-\frac{\gamma}{r} & -\frac{\rd}{\dr}-\frac{\kappa}{r}\\
                   \frac{\rd}{\dr}-\frac{\kappa}{r} & -1-\frac{\gamma}{r}
                 \end{array}
    \right)\left(\begin{array}{c}
                   g^+\\ g^-
                 \end{array}
    \right)\right\rangle_{L^2(\R_+:\C^2)}\delta_{\kappa\kappa'}\delta_{mm'},
  \end{split}
\end{equation}
see also \cite[(7.105)]{Thaller1992}.

\subsubsection*{Spectrum}
The Coulomb--Dirac operator $D_\gamma^H$ has no embedded eigenvalues
(Kalf \cite{Kalf1976}) and no singular continuous spectrum in $[0,\infty)$
(Vogelsang \cite{Vogelsang1988}, and Richard and Tiedra de Aldecoa
\cite{RichardTiedraDeAldecoa2007}).
The lowest eigenvalue is $\lambda_1=\sqrt{1-\gamma^2}$ and one has
$\lim_{k\to\infty}\lambda_k=1$.
The eigenvalues of $D_\gamma^H$ are explicitly known and given by
\begin{align}
  \label{eq:eigenvalue}
  \lambda_{n,\kappa} = \left(1+\tfrac{\gamma^2}{\left(n+\sqrt{\kappa^2-\gamma^2}\right)^2} \right)^{-1/2}.
\end{align}
They only depend on $(\kappa,n)\in(-\N\times\N)\cup(\N\times\N_0)$.
Sommerfeld \cite{Sommerfeld1916} anticipated these eigenvalues 
in the framework of the old relativistic theory of quanta
even before the Coulomb--Dirac operator was written down.
Darwin \cite{Darwin1928}, Gordon \cite{Gordon1928},
and Pidduck \cite{Pidduck1929} solved the eigenvalue equation
for the Coulomb--Dirac operator only 12 years later;
see also Bethe's \cite{Bethe1933} or
Thaller's \cite[Section~7.4]{Thaller1992} textbooks for comprehensive treatments
and Mawhin and Ronveaux \cite{MawhinRonveaux2010} for interesting
historical comments.
In particular, one has the bounds
$$
-\frac{\gamma^2}{2(n+\ell+1)^2} \geq \lambda_{n,\kappa}-1 \geq -\const\cdot\frac{\gamma^2}{(n+\ell+1)^2},
$$
see also \cite[Lemma~1]{HandrekSiedentop2015}.

\begin{remark}
  If the Coulomb potential $\gamma/|x|$ is replaced by a more general
  measurable Hermitian $4\times 4$-matrix-valued function
  $V:\R^3\to\C^{4\times 4}$ with
  $0\leq V(x)\leq\gamma/|x|\otimes\one_{\C^4}$, then the eigenvalues
  of $-i\ualpha\cdot\nabla + \beta-V$ can be computed using analogues
  of the classical Courant--Fischer min-max principle,
  (cf.~\cite[Theorems~XIII.1-2]{ReedSimon1978}) originally written down
  by Talman \cite{Talman1986} and Datta and Devaiah
  \cite{DattaDeviah1988}, mathematically established by Esteban and
  S{\'e}r{\'e} \cite{EstebanSere1997} and by
  \cite{GriesemerSiedentop1999,Griesemeretal1999}, and further
  developed by Dolbeault et
  al.~\cite{Dolbeaultetal2000O,Dolbeaultetal2000}, Morozov and
  M\"uller \cite{MorozovMueller2015}, M\"uller \cite{Mueller2016},
  Esteban et al.~\cite{Estebanetal2019,Estebanetal2021}, and Schimmer
  et al.~\cite{Schimmeretal2020}.  The latter turned the principle
  around and used it to define the Hamiltonian, since the maximization
  leads to the Hardy type inequality \eqref{eq:hardydirac}.
\end{remark}

\subsubsection*{Hydrogen eigenvectors and density}

Similarly, the eigenvectors $\psi_{n,\kappa,m}^D$ of the eigenvalue equation
$D_{\gamma}^H\psi_{n,\kappa,m}^D=\lambda_{n,\kappa}\psi_{n,\kappa,m}^D$ are also
well known, see, e.g., Pidduck \cite{Pidduck1929} (in terms of Laguerre
polynomials) or Gordon \cite{Gordon1928} and Darwin \cite{Darwin1928} (in
terms of hypergeometric confluent functions). 
For nice pictorial representations, see, e.g., White \cite{White1931}
and, for textbook references, see Bethe \cite[Formula (9.37)]{Bethe1933} and
Thaller \cite[p.~427]{Thaller2005}.
The \emph{three-dimensional} density for a Bohr atom for a given $\gamma\in(0,1)$
in spin-orbit channel $\kappa\in\dot\Z$ is
\begin{align}
  \rho_{\kappa,D}^H(x) := \sum_{n=\theta(-\kappa)}^\infty\sum_{m=-j_\kappa}^{j_\kappa}
  \sum_{\sigma=1}^4|\psi_{n,\kappa,m}^D(x,\sigma)|^2, \quad x\in\R^3,
\end{align}
and the total, three-dimensional hydrogenic density is 
\begin{equation}
  \label{rhoh}
  \rho_D^H(x) := \sum_{\kappa\in\zp}\rho_{\kappa,D}^H(x), \quad x\in\R^3.
\end{equation}
These quantities are indeed well defined, as the following theorem
demonstrates. To state it, let
\begin{align}
  \label{eq:defsigmanu}
  \Sigma_\gamma :=1-\sqrt{1-\gamma^2}\in[0,1] \quad \text{for}\ \gamma\in[0,1].
\end{align}

\begin{theorem}[{\cite[Theorem~2.3]{MerzSiedentop2020}}]
  \label{existencerhohdirac}
  Let $1/2<s\leq3/4$, if $\gamma\in(0,\sqrt{15}/4)$ and
  $1/2<s<3/2-\Sigma_\gamma$, if $\gamma\in[\sqrt{15}/4,1)$.
  Then for all $\kappa\in\zp$ there is a constant $A_{s,\gamma}>0$ such that
  for all $x\in\rz^3\setminus\{0\}$ one has
  \begin{multline*}
    \rho_{\kappa,D}^H(x)
    \leq A_{s,\gamma} \frac{|\kappa|^{1-4s}}{|x|^2}\\
    \times\left[\left(\frac{|x|}{|\kappa|}\right)^{2s-1}\one_{\{|x|\leq |\kappa|\}}+\left(\frac{|x|}{|\kappa|}\right)^{4s-1}\one_{\{|\kappa|\leq |x|\leq |\kappa|^2\}}
      + |\kappa|^{4s-1}\one_{\{|x|\geq |\kappa|^2\}}\right].
  \end{multline*}
  Moreover, for any $\varepsilon>0$
  there are constants $A_{\gamma,\varepsilon},A_\gamma>0$
  such that for all $x\in\rz^3\setminus\{0\}$ one has   
  \begin{align}
    \label{eq:boundrhoh}
    \rho_D^H(x)\leq
    \begin{cases}
      A_\gamma |x|^{-3/2} & \text{if}\ \gamma\in(0,\sqrt{15}/4]\\
      A_{\gamma,\varepsilon}\left(|x|^{-2\Sigma_\gamma-\varepsilon}\one_{\{|x|\leq1\}}
        + |x|^{-3/2}\one_{\{|x|>1\}}\right) & \text{if}\ \gamma\in(\sqrt{15}/4,1)
    \end{cases}.
  \end{align}
\end{theorem}

The proof of this theorem is similar to that of
Theorem~\ref{existencerhoh}.

\subsubsection*{Instability}
In the Chandrasekhar model we call an atom ``unstable'', if the
coupling constant is so large that the operator is unbounded from
below.  Recall that for $\gamma<1$ the Coulomb--Dirac operator
$D_\gamma^H$ has a ``distinguished'' self-adjoint extension with the
property that the expectation values of both kinetic and potential
energy are finite in $\dom(D_\gamma^H)$.  \emph{Instability} for the
Coulomb--Dirac operator refers to the fact that all self-adjoint
extensions of the Coulomb--Dirac operator have infinitely many
eigenfunctions with infinite expectation value of the potential
energy.  This situation occurs when $\gamma>1$, i.e., in the case when
the lowest eigenvalue of $D_\gamma^H$ has hit the lower threshold of
the continuous spectrum.  See, e.g., Hogreve
\cite[Theorem~2.1.(iii)]{Hogreve2013} and the references therein for
details and Thaller \cite[p.~218]{Thaller1992} for an overview.

\subsubsection*{Brown--Ravenhall operator}

According to Dirac, the ``vacuum'', i.e., the situation in which no
(negatively charged) electrons with positive energies are present, is
described by a completely filled negative energy continuum of the Dirac
operator (the so-called ``Dirac sea''), whereas only solutions to the
free Dirac equation with positive kinetic energy should be regarded
as ``physical electrons''.
Dirac proposed the following interpretation \cite[p.~362]{Dirac1930A}:
\begin{quote}
  ``\emph{The most stable states for an electron (the states of lowest energy) are
    those with negative energy and very high velocity. All the electrons in the
    world will tend to fall into these states with emission of radiation. The Pauli
    exclusion principle, however, will come into play and prevent more than one
    electron going into any one state. Let us assume there are so many electrons
    in the world that all the most stable states are occupied, or, more accurately,
    that all the states of negative energy are occupied except perhaps a few of small
    velocity. Any electrons with positive energy will now have very little chance
    of jumping into negative-energy states and will therefore behave like electrons
    are observed to behave in the laboratory. We shall have an infinite number of
    electrons in negative-energy states, and indeed an infinite number per unit
    volume all over the world, but if their distribution is exactly uniform we
    should expect them to be completely unobservable. Only the small departures
    from exact uniformity, brought about by some of the negative-energy states being
    unoccupied, can we hope to observe.}''
\end{quote}

Shortly after Dirac's equation was formulated, Breit
\cite{Breit1929,Breit1930,Breit1932} derived a relativistic wave equation
for helium using the quantum electrodynamics of Heisenberg and Pauli,
which reads
\begin{subequations}
  \label{eq:breit}
  \begin{align}
    \label{eq:breit1}
    i\frac{\partial\psi(x_1,x_2)}{\partial t}
    & = \left[(-i\ualpha\cdot\nabla + \beta)^{(1)} + (-i\ualpha\cdot\nabla + \beta)^{(2)} - \frac{\gamma}{|x_1|} - \frac{\gamma}{|x_2|} + \frac{1}{|x_1-x_2|}\right. \\
    \label{eq:breit2}
    & \hspace{-1em} \left. - \frac{1}{2|x_1-x_2|}\left(\ualpha^{(1)}\cdot\ualpha^{(2)} + \frac{\ualpha^{(1)}\cdot(x_1-x_2)\cdot\ualpha^{(2)}\cdot(x_1-x_2)}{|x_1-x_2|^2}\right)\right]\psi(x_1,x_2).
  \end{align}
\end{subequations}
As is explained in \cite[p.~552]{BrownRavenhall1951}
``the superscripts (1) and (2) indicate operation on the coordinate or
spinor components of the first or second electron respectively in the
sixteen-component wave function $\psi(x_1,x_2)$''.
Up to the single term in \eqref{eq:breit2}, Breit's equation coincides with
the naive extension of Dirac's equation for two particles.
Brown and Ravenhall \cite{BrownRavenhall1951} observed that the energies predicted
by Breit's equation did not match the experimentally measured values very well
and explained their observation as follows \cite[pp.~552--553]{BrownRavenhall1951}:
\begin{quote}
  ``\emph{Because of the negative-energy states, equation \eqref{eq:breit}
    is in fact meaningless.
    This can be seen by constructing a solution of [the Dirac equation for two
    electrons in absence of electron-electron repulsion] and then turning on the
    inter-electron interaction slowly. The system can make real transitions to
    states where one electron has a large negative energy and the other electron
    is in the positive-energy continuum; thus equation \eqref{eq:breit} has no
    stationary solutions if interpreted in this way.}''
\end{quote}
In the language of spectral theory, the spectrum of a two-particle
Coulomb--Dirac operator occupies the whole real axis (already without
electron-electron repulsion), and all eigenvalues are embedded.
This phenomenon is sometimes called \emph{Brown--Ravenhall disease}
\cite{Sucher1987,Pilkuhn2005,ReiherWolf2009}.
As is well known from the non-relativistic theory, these are likely to turn
into resonances when the electron-electron repulsion is turned on, which
leads to unphysical consequences, such as the instability of the atom.
We conclude this discussion by mentioning that despite these physical
deficits, Oelker \cite{Oelker2019} recently showed that the single-particle
Dirac operator may be extended to a self-adjoint multi-particle operator.

To remedy the above serious defects, Brown and Ravenhall proposed to only
allow states with positive energy with respect to the free Dirac operator,
i.e., they restricted the Hilbert space of admissible states to
\begin{align}
  \gh_{c,0} := \Lambda_{c,0}(L^2(\br^3:\bc^4))
  := \one_{(0,\infty)}(-ic\ualpha\cdot\nabla+c^2\beta)(L^2(\br^3:\bc^4)).
\end{align}
The energy of an electron in the Brown--Ravenhall picture is then given by
\begin{align}
  \langle\psi,(D_{c,Z}^H-c^2)\psi\rangle,
  \quad \psi\in\Lambda_{c,0}\cs(\R^3:\C^4).
\end{align}
The work \cite{Evansetal1996} showed that this energy is
bounded from below, if and only if $\gamma=Z/c\leq \gamma_B$ with
\begin{align}
  \label{eq:defgammab}
  \gamma_B:= \frac{2}{\pi/2+2/\pi},
\end{align}
which implies $Z<125$.
For such $\gamma$ the energy form can be extended to a closed quadratic
form in $\gh_{c,0}$ with form domain $\gh_{c,0}\cap \cs(\R^3:\C^4)$.
The resulting self-adjoint operator constructed according to Friedrichs
is called Brown--Ravenhall operator and is denoted by $B_{c,Z}$.
By scaling $x\mapsto x/c$ the Brown--Ravenhall operator is seen to be
unitarily equivalent to $c^2B_{1,\gamma}$ with $\gamma=Z/c$. We write
$B_\gamma:=B_{1,\gamma}$.

Although we will not use it in this review, we record the following convenient
representation of the Brown--Ravenhall operator as a self-adjoint operator in
$L^2(\R^3:\C^2)$. For $E_c(\xi):=\sqrt{c^2|\xi|^2+c^4}$ with $\xi\in\R^3$ and
\begin{align}
  \phi_j(\xi) := \sqrt{\frac{E_1(\xi)+(-1)^j}{2E_1(\xi)}},
  \quad j\in\{0,1\}, \quad \xi\in\R^3,
\end{align}
we define the $\C^{2\times2}$-valued functions
\begin{align}
  \Phi_0(\xi):=\phi_0(|\xi|)\one_{\C^2}, \quad
  \Phi_1(\xi):=\phi_1(|\xi|)\frac{\usigma\cdot\xi}{|\xi|},
  \quad \xi\in\R^3.
\end{align}
Then the map
\begin{align}
  \begin{split}
    \bPhi_c:L^2(\br^3:\bc^2) \to L^2(\R^3:\C^4), \quad
    u\mapsto\left(\begin{array}{c}
                     \Phi_0(-i\nabla/c)u\\
                     \Phi_1(-i\nabla/c)u
                   \end{array}\right)
               \end{split}
\end{align}
maps $L^2(\R^3:\C^2)$ unitarily onto $\gh_{c,0}$, cf.~\cite{Evansetal1996}.
Therefore, $B_\gamma$ in $\gh_{c,0}$ is unitarily equivalent to the operator
\begin{align}
  \label{eq:brunitary}
  E_c(-i\nabla)-c^2-\ct_c\left(\frac{Z}{|x|}\right) \quad \text{in}\ L^2(\R^3:\C^2)
\end{align}
with the ``twisted potential''
\begin{align}
  \begin{split}
    \ct_c\left(V\right)
    :&= \left(\begin{array}{c}\Phi_0(-i\nabla/c)\\ \Phi_1(-i\nabla/c)\end{array}\right)V\left(\begin{array}{c}\Phi_0(-i\nabla/c)\\ \Phi_1(-i\nabla/c)\end{array}\right)\\
    & = \Phi_0\left(\frac{-i\nabla}{c}\right) V \Phi_0\left(\frac{-i\nabla}{c}\right)
    + \Phi_1\left(\frac{-i\nabla}{c}\right) V \Phi_1\left(\frac{-i\nabla}{c}\right)
  \end{split}
\end{align}
for any $V:\R^3\to\C^4$, whenever meaningfully defined.
From a technical point of view the representation \eqref{eq:brunitary}
is important, e.g., in
\cite{Evansetal1996,CassanasSiedentop2006,Franketal2009,Merz2019D}.
(Figuratively speaking, the transformation $\ct_c$ mollifies the Coulomb
singularity and ensures the lower boundedness of $B_\gamma$ for coupling
constants greater than $2/\pi$. This transpires, e.g., in
\cite[Lemma~2.7]{Franketal2009} and \cite[Lemma~5.2]{Frank2009}.)

In \cite{Evansetal1996} the authors showed that the Brown--Ravenhall
operator $B_\gamma$ is bounded from below by $-\gamma(\pi/4-1/\pi)-1$ when
$\gamma\leq\gamma_B$. In fact, Tix \cite{Tix1997,Tix1998} proved the
stronger lower bound $-\gamma_B$.

If $\gamma<\gamma_B$, the essential spectrum of the Brown--Ravenhall
operator $B_\gamma$ is $[0,\infty)$ and the singular continuous spectrum is
empty \cite[Theorem 2]{Evansetal1996}. Moreover, there are no embedded
eigenvalues, and the spectrum in $[0,\infty)$ is purely absolutely
continuous \cite[Theorem 3.4.1]{BalinskyEvans2011}.
As in the Chandrasekhar case, the eigenvalues $\lambda_{Z,n,j,\ell}$
of $B_{c,Z}$ in channel $(j,\ell)$ satisfy the bounds
$$
-\frac{Z^2}{2(n+\ell+1)^2} \geq \lambda_{Z,n,j,\ell} \geq -\const\cdot\frac{Z^2}{(n+\ell+1)^2},
$$
where the constant in the second inequality can be chosen independently of $\gamma$.
The lower bound is due to \cite[Theorem~2.1]{Franketal2009}, while the
upper bound follows from the fact that the non-relativistic kinetic energy
dominates the relativistic one. In particular, the Brown--Ravenhall
eigenvalues are smaller than those of $D_{c,Z}^H-c^2$ due to the
min-max-principle for operators with spectral gaps,
cf.~\cite{GriesemerSiedentop1999,Griesemeretal1999}.

\subsubsection*{Furry operator}

Naturally, the projection onto the positive spectral subspace of the
free Dirac operator is not the only possibility to get rid of the
positronic part of the wave functions. Furry and Oppenheimer
\cite{FurryOppenheimer1934} proposed rather to project onto the
positive spectral subspace of the Coulomb--Dirac operator, i.e.,
the Hilbert space of admissible states is
$$
\gh_{c,Z}:=\Lambda_{c,Z}(L^2(\br^3:\bc^4)):=\one_{(0,\infty)}(D_{c,Z}^H)(L^2(\br^3:\bc^4)).
$$
Since $D_\gamma^H$ can be realized as a self-adjoint operator with
form domain $H^{1/2}(\R^3:\C^4)$ by Nenciu's method \cite{Nenciu1976}
when $Z/c=\gamma<\gamma_F$ with
\begin{align}
  \label{eq:defgammaf}
  \gamma_F := 1,
\end{align}
we have
$$
\Lambda_{c,Z}(\cs(\br^3:\bc^4))\subseteq H^{1/2}(\br^3:\bc^4)
$$
and dense in $\gh_{c,Z}$. Therefore, the quadratic form
\begin{align}
  \langle\psi,(D_{c,Z}^H-c^2)\psi\rangle,
  \quad \psi\in\Lambda_{c,Z}\cs(\R^3:\C^4)
\end{align}
is well-defined and bounded from below when $\gamma\in(0,1)$.
By Friedrichs, there is a corresponding self-adjoint operator.
The quadratic form domain of this operator is
$H^{1/2}(\br^3:\bc^4)\cap\gh_{c,Z}$.
This operator is called Furry operator and denoted by $F_{c,Z}$.
Scaling $x\mapsto x/c$ shows that the Furry operator is unitarily
equivalent to $c^2F_{1,\gamma}$ with $\gamma=Z/c$ and we write
$F_\gamma:=F_{1,\gamma}$ in the rescaled picture.

\subsubsection*{Mittleman operator}

Although, at this point, the choice of the projections seems to be
arbitrary and only justifiable by the comparison of the results
with the measured quantities this is -- according to Mittleman
\cite{Mittleman1981} not the case:
the optimal projection and optimal ground state should be obtained by
a minimax-principle, namely the infimum over the states in a class of
fermionic Hilbert spaces defined by the positive spectral subspace of
some Dirac operator $-i\ualpha\cdot\nabla +mc^2\beta - \varphi$ followed
by a supremum over a suitable class of potentials $\varphi$.
However, this has not been implemented on a mathematical level.
In fact there some elementary no-go results
\cite{Barbarouxetal2005S,Barbarouxetal2006} that a potential mathematical
implementation has to circumvent.
For a more detailed review of Mittleman's principle and references,
see, e.g., \cite[Section~4.5]{Estebanetal2008}.

For a variational principle inspired by Mittleman,
see, e.g., \cite{Bachetal1999},
and for works connecting Mittleman's principle and the Dirac--Fock
equations, see, e.g.,
\cite{Barbarouxetal2005S,Barbarouxetal2005,Barbarouxetal2006}.

\subsection{Many-particle operators}

The results for the above-discussed single-particle operators allow
to define many-particle operators.

\subsubsection{Chandrasekhar operator}

Chandrasekhar molecules with $q=2$ are described by
\begin{align}
  \label{eq:defchandramanyparticle}
  C_{N,V} := \sum_{\nu=1}^N\left(\sqrt{-c^2\Delta_\nu+c^4}-c^2-V(x_\nu)\right) + \sum_{1\leq\nu<\mu\leq N}\frac{1}{|x_\nu-x_\mu|} + U
\end{align}
for $\gamma\leq\gamma_C=2/\pi$ and $V$ and $U$ as in \eqref{eq:defv}
and \eqref{eq:defu}, respectively. (Technically, $C_{N,V}$ is defined
as the Friedrichs extension of the corresponding quadratic form with
form domain $\cs(\R^{3N}:\C^{2^N})$.)

The ground state energy of a Chandrasekhar molecule is 
\begin{align}
  \label{eq:defgsenergychandrasekhar}
  E_{c,C}(N,\uZ,\uR):=\inf\spec(C_{N,V}).
\end{align}
In the neutral, atomic case $K=1$, $\uR=0$, $N=Z$, the ground state energy
$E_{c,C}(Z):=E_{c,C}(Z,Z,0)$ is an eigenvalue \cite{Lewisetal1997}.

For an associated ground state $\psi$ of $C_{N,V}$
(or an approximate ground state on the Thomas--Fermi or Scott scale),
we define the associated three-dimensional one-particle ground state densities
\begin{align}
  \rho_C(x) := N\sum_{\sigma=1}^2\int_{\Gamma^{N-1}}|\psi(x,\sigma;y_2,...,y_N)|^2\,\dy_2...\dy_N,
  \quad x\in\R^3.
\end{align}
Similarly as in the Schr\"odinger case, we define the one-dimensional angular
momentum resolved version of $\rho_C$ for any $\ell\in\N_0$ by
\begin{align}
  \varrho_{\ell,C}(r) := Nr^2\sum_{\sigma=1}^2\sum_{m=-\ell}^\ell\int_{\Gamma^{N-1}}\left|\int_{\bs^2}\overline{Y_{\ell,m}(\omega)}\psi(r\omega,\sigma;y_2,...,y_N)\right|^2\,\dy_2...\dy_N
\end{align}
for $r>0$. Note that
\begin{align}
  \int_{\bs^2}\rho_C(r\omega)\,\rd\omega = r^{-2}\sum_{\ell\geq0}\varrho_{\ell,C}(r),
  \quad r>0.
\end{align}

\subsubsection{Brown--Ravenhall and Furry operators}

The energy of an atom in the Brown--Ravenhall or Furry pictures is
given by
\begin{align}
  \ce_{c,Z,N}[\psi] := \langle \psi,\left(\sum_{\nu=1}^N \left(D_{c,Z}-c^2\right)_\nu + \sum_{1\leq\nu<\mu\leq N}\frac{1}{|x_\nu-x_\mu|}\right)\psi\rangle,
\end{align}
whenever $\psi\in \bigwedge_{\nu=1}^N \Lambda_{c,\#} \cs(\R^3:\C^4)$
with $\#\in\{0,Z\}$.
The quadratic form is bounded from below if $\gamma\leq\gamma_B$ in
the Brown--Ravenhall case and if $\gamma\leq1$ in the Furry case
by the previous discussion of the one-particle operators.
The resulting self-adjoint operators constructed
according to Friedrichs are called Brown--Ravenhall and Furry operator,
respectively. If $N=Z$, we drop the third index in the above energy
form and write $\ce_{c,Z}[\psi]:=\ce_{c,Z,Z}$.

The ground state energy of a Brown--Ravenhall or Furry atom is
\begin{align}
  E_{c,B}(N,Z) := \inf\{\ce_{c,Z,N}[\psi]:\,\psi\in \Lambda_{c,0}\cs(\R^3:\C^4)\}, \\
  E_{c,F}(N,Z) := \inf\{\ce_{c,Z,N}[\psi]:\,\psi\in \Lambda_{c,Z}\cs(\R^3:\C^4)\}.
\end{align}
We record that Morozov and Vugalter \cite{MorozovVugalter2006}
(see also Morozov \cite{Morozov2008}, Jakuba{\ss}a-Amundsen
\cite{Jakubassa-Amundsen2005L} for HVZ theorems),
and Matte and Stockmeyer \cite{MatteStockmeyer2010} proved that
for $N=Z$, the Brown--Ravenhall and Furry ground state energies
$E_{B/F}(Z):=E_{B/F}(Z,Z)$ are eigenvalues.

For an associated ground state $\psi$ of $\ce_{c,Z,N}$
(or an approximate ground state on the Thomas--Fermi or Scott scale)
in either the Brown--Ravenhall or Furry picture, we define the
associated one-particle ground state densities
\begin{align}
  \rho_{B/F}(x) := N\sum_{\sigma=1}^4\int_{\Gamma^{N-1}}|\psi(x,\sigma;y_2,...,y_N)|^2\,\dy_2...\dy_N,
  \quad x\in\R^3.
\end{align}
As in the Chandrasekhar case, we define a spin-orbit
resolved version of $\rho_{B/F}$.
More precisely, for given spin-orbit coupling $\kappa\in\dot\Z$
(recall \eqref{l}-\eqref{zp}), we define
\begin{align}
  \rho_{\kappa,B/F}(x) := \frac{N}{4\pi} \sum_{\sigma\in\{+,-\}}\sum_{m=-j_\kappa}^{j_\kappa}\int_{\Gamma^{N-1}}\dy\,
  \left|\sum_{\tau=1}^4\int_{\mathbb{S}^2}\,\rd\omega\,
    \overline{\Phi_{\kappa,m}^\sigma(\omega,\tau)}
    \psi(|x|\omega,\tau,y)\right|^2
\end{align}
for $x\in\R^3$.
Here $\Phi^\sigma_{\kappa,m}$ are the spherical Dirac spinors
\eqref{fi}.
Note that
\begin{align}
  \frac{1}{4\pi}\int_{\bs^2}\rho_{B/F}(|x|\omega)\rd\omega = \sum_{\kappa\in\zp}\rho_{\kappa,B/F}(x),
  \quad x\in\R^3.
\end{align}

\begin{remarks}
  (1) The Brown--Ravenhall and Furry operators are examples of so-called
  ``no-pair'' operators, i.e., Schr\"odinger operators that can
  \emph{formally} be derived from
  quantum electrodynamics by \emph{neglecting the creation of
    electron-positron pairs} \cite{Sucher1980}. These operators
  are popular among quantum chemists, as they provide decent
  numerical results which are in good accordance with experimentally
  measured data.
  For instance, the Scott correction in the Furry picture (Formula
  \eqref{eq:scottfurryatom} in Theorem \ref{scottatom}) coincides
  astonishingly well with experimental data (see, e.g., \cite{NIST_ASD}),
  see \cite[Section~6]{HandrekSiedentop2015}
  and \cite{Pilkuhn2005,ReiherWolf2009} for textbook treatments.

  (2) The Scott correction is also believed to be true when the mean
  field in the sense of Mittleman \cite{Mittleman1981} is taken into
  account.
  However, this is so far only known in the Hartree--Fock approximation
  when the involved projection is given by the Dirac--Fock operator,
  see Fournais et al.~\cite{Fournaisetal2020}.
\end{remarks}

In the following subsections we summarize results concerning the
asymptotic expansion of the ground state energies and convergence
of the one-particle ground state densities for the above-introduced
models in the atomic and molecular cases.
In particular we review the elements of the proof of the
\emph{relativistic strong Scott conjecture} for Chandrasekhar atoms
in \cite{Franketal2020P,Franketal2020R}.

\subsection{Atoms without magnetic fields}
\label{ss:relatom}

As far as we know, the results below have only been proved
in the neutral case $N=Z$.
We believe that they also hold for ions.

\subsubsection{Thomas--Fermi scale}
\label{sss:relatomtf}

In all of the above relativistic models the leading order of the
asymptotic expansion of the ground state energy is non-relativistic.
These results indicate that electrons whose distances to the nucleus
are of order $Z^{-1/3}$ still behave non-relativistically, although
they are being sucked into the nucleus.  The following theorem was
proved by S\o rensen \cite{Sorensen2005} (Chandrasekhar), by
\cite{CassanasSiedentop2006} and \cite{Franketal2009}
(Brown--Ravenhall), and by \cite{HandrekSiedentop2015} (Furry, as a
consequence of the Scott correction).  Recall the critical coupling
constants $\gamma_C,\gamma_B,\gamma_F$ in \eqref{eq:defgammac},
\eqref{eq:defgammab}, \eqref{eq:defgammaf}.

\begin{theorem}
  \label{gsenergyrelativistictf}
  Let $N=Z$ and $E^{\rm TF}(Z)$ denote the Thomas--Fermi energy of a
  neutral Thomas--Fermi atom with $q=2$.
  Then
  \begin{align}
    \lim_{\substack{Z,c\to\infty}}\frac{E_{c,C}(Z)-E^{\rm TF}(Z)}{Z^{7/3}} = 0
    \quad \text{for fixed}\ \frac{Z}{c}\leq\gamma_C,\\
    \lim_{\substack{Z,c\to\infty}}\frac{E_{c,B}(Z)-E^{\rm TF}(Z)}{Z^{7/3}} = 0
    \quad \text{for fixed}\ \frac{Z}{c}\leq\gamma_B,\\
    \lim_{\substack{Z,c\to\infty}}\frac{E_{c,F}(Z)-E^{\rm TF}(Z)}{Z^{7/3}} = 0
    \quad \text{for fixed}\ \frac{Z}{c}<\gamma_F.
  \end{align}
\end{theorem}

\begin{remarks}[Elements in the proof of Theorem \ref{gsenergyrelativistictf}]
  We make some remarks on S\o rensen's proof \cite{Sorensen2005} for
  the Chandrasekhar case.
  Since $\sqrt{p^2+1}-1\leq p^2/2$, it suffices to prove
  the lower bound, whose proof is similar to Lieb's simplified proof
  of Theorem \ref{quantumtfconvnonrelatom}, see \cite[Theorem~5.1]{Lieb1981}.
  It can be split into the following steps.
  
  \begin{enumerate}
  \item Reduce the linear many-particle problem to estimating the
    non-linear one-particle quadratic form
    \begin{align}
      \sum_{\nu=1}^N\langle m_\nu,\left[\sqrt{-c^2\Delta+c^4}-c^2-\frac{Z}{|x|} + \rho_Z^{\rm TF}\ast\frac{1}{|\cdot|}\right]m_\nu\rangle
    \end{align}
    with the Thomas--Fermi density $\rho_Z^{\rm TF}$ and orthonormal
    orbitals $\{m_\nu\}_{\nu=1}^N$ from below with the help of a
    correlation inequality (e.g., by Lieb and Oxford
    \cite{Onsager1939,Lieb1979,LiebOxford1981}
    or that of \cite{Mancasetal2004}).

  \item Due to the (non-perturbative) Coulomb singularity at the origin,
    one localizes position space into the regions $|x|\lesssim Z^{-o}$
    and $|x|\gtrsim Z^{-o}$ with $o\in(1/3,2/3)$. Guided by the proof of
    Theorem \ref{quantumtfconvnonrelatom} (Lieb \cite[Theorem~5.1]{Lieb1981})
    and the intuition that electrons on distances $Z^{-o}$ behave
    non-relativistically, it is expected that the electrons in the region
    $|x|\gtrsim Z^{-o}$ lead to the TF energy, while the contribution
    from $|x|\lesssim Z^{-o}$ is $o(Z^{7/3})$.
    The localization errors can be controlled at the end of the argument
    with the help of an arbitrarily small amount ($Z^{-\epsilon}$) of
    kinetic energy.

  \item The contribution of the electrons in $\{|x|\lesssim Z^{-o}\}$,
    where the Coulomb singularity is located, can be controlled with the
    help of the following strengthening of one of Daubechies' inequalities
    \cite{Daubechies1983},
    \begin{align}
      \label{eq:hlt}
      \tr\left(|p|-\frac{2/\pi}{|x|}-V\right)_-
      \lesssim \int_{\br^3} V(x)_+^4\,\dx
    \end{align}
    by \cite{Franketal2008H} and the inequality
    $\sqrt{p^2+1}-1\geq|p|-1$.
    (Inequality \eqref{eq:hlt} is often called Hardy--Lieb--Thirring inequality
    because of the homogeneity of the ``unperturbed Hardy operator''
    $|p|-\frac{2/\pi}{|x|}$.)
    Here, $V$ is a bounded function, supported on $\{|x|\lesssim Z^{-o}\}$.
    Using \eqref{eq:hlt} one computes the contribution to the energy to be
    $\co(Z^{2+3(1-o)})$, which is more than $Z^{7/3}$.
    For this reason another localization on the length scale $Z^{-i}$ is
    necessary. The Hardy--Lieb--Thirring inequality for $V$ being supported
    on $\{|x|\lesssim Z^{-i}\}$ then produces an
    $o(Z^{7/3})$ error when $i\in(8/9,1)$.
    The additionally introduced localization error can be controlled by an
    $\epsilon$ of kinetic energy, too.

  \item The energy contribution of electrons located in the intermediate
    region $\{Z^{-i}\lesssim|x|\lesssim Z^{-o}\}$ can be controlled using
    Daubechies' inequality,
    \begin{align}
      \tr\left(\sqrt{-\Delta+1}-1-V(x)\right)_-
      \lesssim \int_{\R^3} (V_+^{5/2}(x)+V_+^4(x))\,\dx.
    \end{align}
    Note that $|x|^{-1}\notin L^4(\R^3)$.

  \item Electrons in the region $\{|x|\gtrsim Z^{-o}\}$ are expected to
    generate the TF energy. Here semiclassical analysis is used.
    Roughly speaking, one compares the quantum energy
    \begin{align}
      \sum_{\nu=1}^N\langle m_\nu,\one_{\{|x|\gtrsim Z^{-o}\}}\left[\sqrt{-c^2\Delta+c^4}-c^2-\frac{Z}{|x|}+\rho_Z^{\rm TF}\ast\frac{1}{|\cdot|}\right] \one_{\{|x|\gtrsim Z^{-o}\}}m_\nu\rangle
    \end{align}
    to the classically expected energy
    \begin{align}
      \iint\limits_{|q|\gtrsim Z^{-o}} \left(\sqrt{c^2p^2+c^4}-c^2-\frac{Z}{|q|}+\rho_Z^{\rm TF}\ast\frac{1}{|\cdot|}(q)\right)_-\,\frac{\rd p\,\rd q}{(2\pi)^{3}}.
    \end{align}
    The latter leads to the \emph{non-relativistic} TF energy, since
    $\sqrt{c^2p^2+c^4}-c^2\sim p^2/2$ for $|p|\lesssim Z^{o}\ll Z$
    (since $|x|\gtrsim Z^{-o}$ with $o\in(1/3,2/3)$).
    To that end a phase-space localization using coherent states
    \cite[Theorem~5.1]{Lieb1981} is used.
    In fact, merely the localization errors coming from the
    phase-space localization force the position localization to
    the scale $\lesssim Z^{-1/3}$.
  \end{enumerate}
\end{remarks}

Theorem \ref{gsenergyrelativistictf} is accompanied by the following
convergence results for the ground state densities.

\begin{theorem}[{\cite{MerzSiedentop2019,Merz2019D}}]
  \label{gsdensityrelativistictf}
  Let $N=Z$ and $\rho_1^{\rm TF}$ denote the hydrogenic Thomas--Fermi
  minimizer with $\int\rho_1^{\rm TF}(x)\,\dx=1$ and $q=2$. Then
  \begin{align}
    \lim_{\substack{Z,c\to\infty}}Z^{-2}\rho_{C}(Z^{-1/3}\cdot) = \rho_1^{\rm TF}
    \quad \text{for fixed}\ \frac{Z}{c}\leq\gamma_C, \\
    \lim_{\substack{Z,c\to\infty}}Z^{-2}\rho_{B}(Z^{-1/3}\cdot) = \rho_1^{\rm TF}
    \quad \text{for fixed}\ \frac{Z}{c}\leq\gamma_B, \\
    \lim_{\substack{Z,c\to\infty}}Z^{-2}\rho_{F}(Z^{-1/3}\cdot) = \rho_1^{\rm TF}
    \quad \text{for fixed}\ \frac{Z}{c}<\gamma_F.
  \end{align}
  In all three formulae the convergence holds in Coulomb norm
  (see \eqref{eq:defD}-\eqref{eq:defcoulombnorm}) with convergence rate
  $\co(Z^{-3/16})$. In the Chandrasekhar case the convergence also holds
  when both sides are integrated against any $U\in L^{5/2}\cap L^4(\R^3)$
  and in the Brown--Ravenhall case when in addition
  $U\in |\cdot|^{-1}L^\infty$ is Lipschitz.
\end{theorem}

\begin{remarks}
  (1) The proof of the convergence in Coulomb norm uses an observation of
  Fefferman and Seco \cite{FeffermanSeco1989}, together with the energetic
  expansion of the ground state energy.
  For an error term $\co(Z^{a})$ to the leading TF energy with
  $a<7/3$ the convergence rate is $\co(Z^{(a-7/3)/2})$.
  In view of the energetic results
  \cite{Franketal2008,Solovejetal2008,Franketal2009,HandrekSiedentop2015}
  (Scott correction), we have $a=47/24$.

  (2) The proof of weak convergence uses the proof of the energetic
  results (Theorem \ref{gsenergyrelativistictf}) together with a
  linear response argument. We will flesh out the details in the
  discussion of the densities on the Scott scale (Subsection
  \ref{sss:elementsfmss}).
\end{remarks}

\subsubsection{Scott scale}

As explained in the introduction, electrons in proximity of the
nucleus are expected to generate relativistic effects that should
be visible in the ground state energy and density on the spatial scale
$Z^{-1}$. In fact, the Scott correction is relativistically lowered.
The precise amount depends on the sum of the differences of the
non-relativistic and the relativistic hydrogen eigenvalues.
To that end, let $\lambda_n^S$, $\lambda_n^C$, $\lambda_n^B$, and
$\lambda_n^F$ denote the $\gamma$-dependent eigenvalues of the non-relativistic
operator $S_\gamma^{H}$ in \eqref{eq:defhydrogenoperator}, of the
Chandrasekhar operator $C_\gamma^H$ in \eqref{eq:defchandra} (with $q=2$),
of the Brown--Ravenhall operator $B_\gamma$.
and of the Furry operator $F_\gamma$.
(The Furry eigenvalues coincide of course with those
\eqref{eq:eigenvalue} of $D_\gamma^H-1$.)
We introduce the spectral shifts
\begin{align}
  [0,\gamma_C]\ni\gamma\mapsto s_C(\gamma) & := \gamma^{-2}\sum_{n\geq0}\left(\lambda_n^S-\lambda_n^C\right) \geq 0, \\
  [0,\gamma_B]\ni\gamma\mapsto s_B(\gamma) & := \gamma^{-2}\sum_{n\geq0}\left(\lambda_n^S-\lambda_n^B\right) \geq 0, \\
  [0,\gamma_F]\ni\gamma\mapsto s_F(\gamma) & := \gamma^{-2}\sum_{n\geq0}\left(\lambda_n^S-\lambda_n^F\right) \geq 0,
\end{align}
and record the following observation.

\begin{proposition}
  Let $\#\in\{C,B,F\}$. Then the functions $s_{\#}$ on their
  respective domains are continuous and monotone decreasing,
  and obey $s_{\#}(0)=0$.
\end{proposition}

\begin{proof}
  For $s_C$ this is proved in
  \cite[Theorems~1.1, 1.4, Corollary~1.6]{Solovejetal2008}.
  See also \cite[p.~552]{Franketal2008}, where it is shown that
  $s_C$ is monotone decreasing and finite.
  For $s_F$ the claim can be inferred from the explicitly known
  eigenvalues $\lambda_n^S$ and $\lambda_n^F$, respectively.
  The continuity and monotonicity of $s_B$ follow from the explicit
  knowledge of the Schr\"odinger and Coulomb--Dirac eigenvalues,
  and the inequality $\lambda_n^F\geq\lambda_n^B$.
\end{proof}

The following theorem concerning the energy asymptotics of
Chandrasekhar atoms was proved independently by
Solovej et al.~\cite{Solovejetal2008}, and
the work \cite{Franketal2008}
using different techniques.
The result for Brown--Ravenhall atoms was proved by
\cite{Franketal2009}, and that for
Furry atoms by \cite{HandrekSiedentop2015}.

\begin{theorem}
  \label{scottatom}
  Let $N=Z$ and $E^{\rm TF}(Z)$ denote the Thomas--Fermi energy of
  a neutral Thomas--Fermi atom with $q=2$.
  Then
  \begin{align}
    \label{eq:scottchandrasekharatom}
    \lim_{\substack{Z,c\to\infty}}\frac{E_{c,C}(Z)-\left[E^{\rm TF}(Z)+\left(\frac12-s_C(\gamma)\right)Z^2\right]}{Z^{2}} & = 0
    \quad \text{for fixed}\ \frac{Z}{c}\leq\gamma_C,\\
    \label{eq:scottbratom}
    \lim_{\substack{Z,c\to\infty}}\frac{E_{c,B}(Z)-\left[E^{\rm TF}(Z)+\left(\frac12-s_B(\gamma)\right)Z^2\right]}{Z^{2}} & = 0
    \quad \text{for fixed}\ \frac{Z}{c}\leq\gamma_B,\\
    \label{eq:scottfurryatom}
    \lim_{\substack{Z,c\to\infty}}\frac{E_{c,F}(Z)-\left[E^{\rm TF}(Z)+\left(\frac12-s_F(\gamma)\right)Z^2\right]}{Z^{2}} & = 0
    \quad \text{for fixed}\ \frac{Z}{c}<\gamma_F.
  \end{align}
  In all of the above limits the error term can be quantified
  and is $\co(Z^{47/24})$.
\end{theorem}

\begin{remark}
  It is believed that these results also hold for ions (at least as
  long as the ionization degree is sufficiently small), since the electrons
  on length scales $\co(Z^0)$ should not disturb the energy generated by
  electrons on length scales $\co(Z^{-1})$.
  However, a rigorous proof is lacking.
\end{remark}

The energetic results on the Scott scale are accompanied by recent
results \cite{Franketal2020P,Franketal2020R,MerzSiedentop2020}
for the density in the Chandrasekhar and Furry cases.

\begin{theorem}[{\cite{Franketal2020P,Franketal2020R,MerzSiedentop2020}}]
  \label{fmss}
  Let
  \begin{align}
    \label{eq:testfunctionsexample}
    U \in \D := & \left\{W\in L_{\rm loc}^1(\R_+): \, \forall \epsilon>0 \, \exists a>0\, \forall r>0: \right.\notag \\
    & \qquad\qquad\qquad\quad\left. |W(r)|\leq a\left(r^{-1}\one_{\{r\leq1\}} + r^{-\frac32-\epsilon}\one_{\{r\geq1\}}\right) \right\}
  \end{align}
  be arbitrary.
  Then the following statements hold.
  \begin{enumerate}
  \item (Convergence for fixed angular momentum/spin-orbit coupling)
    Let $\ell\in\N_0$ and $\kappa\in\dot\Z$ be fixed. Then
    \begin{align}
      \label{eq:fmss1}
      \lim_{\substack{Z,c\to\infty}}\int_0^\infty c^{-3}\varrho_{\ell,C}(c^{-1}r)U(r)\,\dr & = \int_0^\infty \varrho_{\ell,C}^H(r)U(r)\,\dr
      \quad \text{for fixed}\ \frac{Z}{c}<\gamma_C, \\
      \label{eq:ms1}
      \lim_{\substack{Z,c\to\infty}}\int_{\R^3}c^{-3}\rho_{\kappa,F}(c^{-1}x)U(|x|)\,\dx & = \int_{\R^3}\rho_{\kappa,D}^H(x)U(|x|)\,\dx
      \quad \text{for fixed}\ \frac{Z}{c}<\gamma_F.
    \end{align}

  \item (Convergence of total density) We have
    \begin{align}
      \label{eq:fmss2}
      \lim_{\substack{Z,c\to\infty}}\int_{\R^3}c^{-3}\rho_{C}(c^{-1}x)U(|x|)\,\dx & = \int_{\R^3}\rho_C^H(x)U(|x|)\,\dx
      \quad \text{for fixed}\ \frac{Z}{c}<\gamma_C,\\
      \label{eq:ms2}
      \lim_{\substack{Z,c\to\infty}}\int_{\R^3}c^{-3}\rho_{F}(c^{-1}x)U(|x|)\,\dx & = \int_{\R^3}\rho_D^H(x)U(|x|)\,\dx
      \quad \text{for fixed}\ \frac{Z}{c}<\gamma_F.
    \end{align}
  \end{enumerate}
\end{theorem}

\begin{remarks}
  (1) For the sake of clarity we restricted attention to the above
    class $\D$ of test functions, although the results actually hold for
    a substantially larger class.

  (2) However, the exemplary test function class \eqref{eq:testfunctionsexample}
  is believed to be optimal.
  \begin{enumerate}[(a)]
  \item Due to Kato's inequality, we cannot expect (at least not for
    $\ell=0$) \eqref{eq:fmss1}-\eqref{eq:ms2} to hold for test
    functions, whose singularity is worse than $|x|^{-1}$.
    
  \item The $|x|^{-3/2-\epsilon}$ decay seems optimal in view of the expected
    $|x|^{-3/2}$-decay of $\rho_{C/F}^H$, see Theorems
    \ref{existencerhoh} and \ref{existencerhohdirac}.
  \end{enumerate}

  (3) In view of Theorems \ref{existencerhoh} and
  \ref{existencerhohdirac}, a transition between the length scales
  $Z^{-1}$ and $Z^{-1/3}$ is again clearly visible.
\end{remarks}

\subsubsection{Elements in the proof of Theorem \ref{fmss}}
\label{sss:elementsfmss}

The rest of this subsection is concerned with explaining the key elements
of the proof of Theorem \ref{fmss} in the Chandrasekhar case, i.e.,
the limits \eqref{eq:fmss1} and \eqref{eq:fmss2}.
For simplicity we set $q=1$ here.
We begin with the argument to prove \eqref{eq:fmss1}
for a fixed angular momentum channel.

We follow the lines of Lieb and Simon \cite{LiebSimon1977},
Baumgartner \cite{Baumgartner1976}, and 
\cite{Iantchenkoetal1996} by employing a linear-response argument.
Let $\lvert\psi\rangle\langle\psi\rvert$ be a ground state density
matrix of the atomic many-particle operator $C_Z$ (see
\eqref{eq:defchandramanyparticle}), and define, for $U\in\D$ and in
slight abuse of notation, the perturbed operator
\begin{align}
  C_{Z,\lambda} := C_Z - \lambda\sum_{\nu=1}^Z c^2U(c|x_\nu|)\Pi_{\ell,\nu}
  \quad \text{in}\ \bigwedge_{\nu=1}^Z L^2(\R^3).
\end{align}
Here $\Pi_{\ell}$ is the orthogonal projection in $L^2(\R^3)$ onto
the $\ell$-th angular momentum channel defined by
$$
\Pi_{\ell} := \sum_{m=-\ell}^{\ell} |Y_{\ell,m}\rangle\langle Y_{\ell,m}|,
$$
and $\Pi_{\ell,\nu}$ acts as $\Pi_\ell$ with respect to the $\nu$-th
particle. Since the singularity of $U$ is Coulombic,
$C_{Z,\lambda}$ is realized as a
self-adjoint operator by Kato's inequality if $\gamma<2/\pi$ and
$|\lambda|$ is sufficiently small. Moreover, since \eqref{eq:fmss1}
is linear in $U$ we may assume $U\geq0$ and $\lambda>0$ without loss
of generality.

By the linear response argument,
the Scott correction (Theorem \ref{scottatom}),
and scaling $x\mapsto x/c$, we have
\begin{align}
  \label{eq:scottdensityaux1}
  \begin{split}
    & \varlimsup\limits_{Z\to\infty}\int_0^\infty c^{-3}\varrho_{\ell,C}(r/c)U(r)\,\dr \\
    & \quad = \varliminf\limits_{\lambda\searrow0}\varlimsup\limits_{Z\to\infty}\tr_{\bigwedge_{\nu=1}^Z L^2(\R^3)}\left[\frac{|\psi\rangle\langle\psi|(C_Z-C_{Z,\lambda})}{\lambda c^2}\right]\\
    & \quad \leq
    (2\ell+1) \cdot \varliminf\limits_{\lambda\searrow0} \frac{\tr_{L^2(\R_+)}(C_{\ell,\gamma}^H-\lambda U(r))_- - \tr_{L^2(\R_+)}(C_{\ell,\gamma}^H)_-}{\lambda}
  \end{split}
\end{align}
with $C_{\ell,\gamma}^H$ as in \eqref{eq:defcellh}.
To compute the right side of \eqref{eq:scottdensityaux1} we have
two options.
\begin{enumerate}
\item Find a majorant to apply the dominated convergence theorem
  to interchange $\liminf_{\lambda\searrow0}$ and $\tr$. Then apply
  standard perturbation theory, i.e., the classical Hellmann--Feynman
  theorem for a single eigenvalue.
  This lead to the shorter proof in \cite{Franketal2020R}.
  
\item Compute the derivative with respect to $\lambda$ directly.
  This lead to the longer, original proof in \cite{Franketal2020P}.
\end{enumerate}

Here we shall present the arguments of the longer proof,
as we believe that it better unearths the involved mathematics
of the relativistic strong Scott conjecture.
Besides, it allows us to popularize a generalization of the
classical Hellmann--Feynman theorem that may be of independent
interest in the analysis of many-particle problems. 

First we state this generalized Hellmann--Feynman theorem with
``natural'' assumptions on the perturbation.
However, this version is not applicable to our problem.
Afterwards we state a generalization with weaker
assumptions, which suffices for our
purposes. Recall that an operator $B$ is called \emph{relatively form
  trace class} with respect to a self-adjoint operator $A$ that
is bounded from below, if $(A+M)^{-1/2}B(A+M)^{-1/2}$ is trace class for
some (and hence any) large enough $M>0$.

\begin{theorem}[{\cite[Theorem~3.1]{Franketal2020P}}]
  \label{diff}
  Assume that $A$ is a self-adjoint operator in some Hilbert space
  with $A_-$ trace class.
  Assume that $B$ is a non-negative operator in the same Hilbert space,
  and relatively form trace class with respect to $A$.
  Then the one-sided derivatives of
  $$
  \lambda\mapsto S(\lambda) := \Tr(A-\lambda B)_-
  $$
  satisfy
\begin{equation}
\label{eq:diff}
  \Tr B \one_{(-\infty,0)}(A) = D^-S(0) 
  \leq D^+S(0) = \Tr B \one_{(-\infty,0]}(A).
\end{equation}
  In particular, $S$ is differentiable at $\lambda=0$, if and only
  if $B|_{\ker A} = 0$.
\end{theorem}

\begin{remarks}
  (1) The relative form trace class assumption implies
  that the expression on the right of \eqref{eq:diff}, and
  consequently also that on the left, is finite.
  
  (2) By the variational principle it follows that $S$ is convex.
  Thus, $S$ has left and right sided derivatives
  (cf.~\cite[Theorem 1.26]{Simon2011}).
  
  (3) If $\inf\sigma_{\mathrm{ess}}(A)>0$ or $A-\lambda B$ has only
  finitely many negative eigenvalues, then the derivative and the
  trace can be interchanged and the result follows from the classical
  Hellman--Feynman theorem.
  The point is that the formulae remain valid even when the bottom
  of the essential spectrum is zero, so that perturbation theory is
  not directly applicable.  
\end{remarks}

In our application, $A$ is the Chandrasekhar operator $C_{\ell,\gamma}^H$
(which has no zero eigenvalue) and $B=U$ in $L^2(\br_+,\dr)$.
Thus, \eqref{eq:fmss1} would follow from Theorem \ref{diff}, if one could
verify its assumptions. By Kato's inequality (using $\gamma<2/\pi$),
one can replace $A$ by the kinetic energy $C_\ell$. In this case, the relative
trace class condition can be formulated explicitly using the Fourier--Bessel
transform. But since $(k+1)^{-1} \notin L^1(\R_+,\dk)$, this shows that
$C_{\ell,\gamma}^H$ cannot satisfy the relative trace class assumption,
no matter how nice the perturbation $B=U$ is.

For this reason, we proved and used a generalization of Theorem \ref{diff},
where the relative trace class assumption is stated with respect to
$(A+M)^{2s}$ for some $s>1/2$. We will state this generalization in
Theorem~\ref{diffgen0} in a moment.
Using the following special case of \cite[Theorem~1.1]{Franketal2021},
the assumption of Theorem~\ref{diffgen0} can be recast as an assumption
involving $C_\ell$ instead of $C_{\ell,\gamma}^H$.

\begin{theorem}[{\cite[Theorem~1.1]{Franketal2021}}]
  \label{fms}
  Let $2/\pi\geq\gamma\geq0$,
  $0<s<3/2-\sigma_\gamma$ (with $\sigma_\gamma$ as in
  \eqref{eq:defsigmagamma}), and $V\in L_{\rm loc}^1(\R^3)$ satisfying
  $-\gamma/|x|\leq V\leq0$. Then we have the quadratic
  form inequality
  \begin{align}
    |p|^{2s} \lesssim (|p|+V)^{2s}.
  \end{align}
\end{theorem}

We are now ready to state the generalization of Theorem \ref{diff}.
\begin{theorem}[{\cite[Theorem~3.2]{Franketal2020P}}]
  \label{diffgen0}
  Assume that $A$ is self-adjoint with $A_-$ trace class. Assume that $B$ is
  non-negative and relatively form bounded with respect to $A$. Assume that
  there are $1/2< s\leq 1$ such that for some $M>-\inf\spec A$,
    \begin{equation}
    \label{eq:traceclassdelta0}
    (A+M)^{-s} B(A+M)^{-s} \qquad\text{is trace class}
  \end{equation}
  and
  \begin{equation}
    \label{eq:relbounddelta}
    \limsup_{\lambda\to 0} \left\| (A+M)^{s} (A-\lambda B+M)^{-s} \right\| <\infty.
  \end{equation}
  Then the conclusions in Theorem \ref{diff} are valid.
\end{theorem}

Thus, to conclude the proof of \eqref{eq:fmss1}, we are left with showing
the assumptions of Theorem \ref{diffgen0}. Since the test functions
$U$ in the above formulation of Theorem \ref{fmss} are bounded by a
multiple of the Coulomb potential,
the condition \eqref{eq:relbounddelta} can be verified easily using
Theorem~\ref{fms}.
We now verify \eqref{eq:traceclassdelta0} with $A=C_{\ell,\gamma}^H$
and $B=U$. Letting $\|\cdot\|_2$ denote the Hilbert--Schmidt norm,
we have, for any $1/2<s\leq1$,
\begin{align*}
  & \|U^{1/2}(C_{\ell,\gamma}^H+M)^{-s}\|_2^2
  \lesssim \|U^{1/2}(C_{\ell}+M)^{-s}\|_2^2\\
  & \quad = \int_0^\infty \dr\ U(r)\int_0^\infty \dk\ \frac{kr J_{\ell+1/2}(kr)^2}{(\sqrt{k^2+1}-1+M)^{2s}}
    \lesssim \|U\|_{\ck_{s}^{(0)}},
\end{align*}
where we used Theorem \ref{fms} and the Fourier--Bessel transform.
Here, 
\begin{align*}
  \|W\|_{\ck_{s}^{(0)}} := \sup_{R\geq1/2} \left[\int_0^R \left(\frac{r}{R}\right)^{2s-1}|W(r)|\,\dr + \int_R^\infty |W(r)|\,\dr\right].
\end{align*}
Clearly, functions $U\in\D$, the test function space
\eqref{eq:testfunctionsexample}, satisfy
$\|U\|_{\ck_{s}^{(0)}}<\infty$.
This concludes the sketch of the arguments in the proof of \eqref{eq:fmss1}.


\medskip
To prove the convergence of the total density $c^{-3}\rho_C(x/c)$,
we also need to interchange the $\ell$-summation with the
limits $Z\to\infty$ and $\lambda\to0$. This is done using the
dominated convergence theorem. To apply it,
it suffices to prove that there is an $\epsilon>0$ such that
\begin{align}
  \label{eq:scottdensityaux2}
  \tr(C_{\ell}-V-\lambda U(r))_- - \tr(C_{\ell}-V)_-
  \lesssim \lambda(\ell+1/2)^{-2-\epsilon},
\end{align}
whenever $0\leq V\leq\gamma/r$.
In our application, $V$ is closely related to the Thomas--Fermi potential.
For the proof of \eqref{eq:scottdensityaux2}, see
\cite[Section~5]{Franketal2020P}.

\subsection{Molecules without magnetic fields}
\label{ss:relmolecule}

We immediately present the energetic result for Chandrasekhar
molecules on the Scott scale, which was proved by
Solovej et al.~\cite{Solovejetal2008}.

\begin{theorem}
  \label{scottchandrasekharmolecule}
  Let $q\in\N$, $\uz=(z_1,...,z_K)\in(0,1)^K$ with
  $\sum_{\kappa=1}^K z_\kappa=1$ and $\ur=(r_1,...,r_K)\in\R^{3K}$ with
  $\min_{k\neq\ell}|r_k-r_\ell|>r_0$ for some $r_0>0$. Define
  $\uZ=(Z_1,...,Z_K)=|\uZ|z$ and $R=|\uZ|^{-1/3}r$ for $|\uZ|>0$.
  Let $E^{\rm TF}(\uz,\ur)$ be the Thomas--Fermi energy of the
  unconstrained problem \eqref{eq:tfunconstrained} and
  $E_{c,C}(N,\uZ,\uR)$ be the ground state energy with nuclear
  configuration $\uZ$ and $\uR$.
  Then the following statements hold:
  \begin{enumerate}
  \item The function
    \begin{align}
      \left[0,\frac2\pi\right]\ni\gamma\mapsto \cs_C(\gamma):=
      \lim_{\kappa\to0}\left(\iint \left[\frac{p^2}{2}-\frac{1}{|v|}+\kappa\right]_-\,\frac{\rd p\,\rd v}{(2\pi)^3} - \tr[H_C(\gamma) + \kappa]_-\right)
    \end{align}
    with
    \begin{align}
      H_C(\gamma) :=
      \begin{cases}
        \sqrt{-\gamma^{-2}\Delta+\gamma^{-4}}-\gamma^{-2}-1/|x| & \quad \text{if}\ \gamma\in(0,2/\pi],\\
        -\frac12\Delta-1/|x| & \quad \text{if}\ \gamma=0
      \end{cases}
    \end{align}
    is continuous, monotone decreasing,
    and satisfies $\cs_C(0)=1/4$.

  \item As $|\uZ|=\sum_{\kappa=1}^K Z_\kappa\to\infty$ and $c\to\infty$
    with $\max_\kappa\{Z_\kappa/c\}\leq2/\pi$, one has
    \begin{align}
      E_{c,C}(|\uZ|,\uZ,\uR) = E^{\rm TF}(\uz,\ur)|\uZ|^{7/3} + q\sum_{\kappa=1}^K Z_\kappa^2 \cdot \cs_C\left(\frac{Z_\kappa}{c}\right) + \co(|\uZ|^{2-1/30}).
    \end{align}
    The error term $\co(|\uZ|^{2-1/30})$ means that
    $|\co(|\uZ|^{2-1/30})|\lesssim |\uZ|^{2-1/30}$, where the implicit
    constant depends only on $r_0$ and $K$.
  \end{enumerate}
\end{theorem}

\begin{remark}
  The convergence of the one-particle ground state density
  on the Thomas--Fermi scale in Coulomb norm can be proved using the
  argument in the proof of Theorem \ref{gsdensityrelativistictf} 
  together with Theorem \ref{scottchandrasekharmolecule}.
  The convergence of the density on the Scott scale has not been
  worked out so far.
\end{remark}

\subsection{Molecules with self-generated magnetic fields}
\label{ss:relmagnetic}

Suppose that the kinetic energy of the electrons is described by
the Chandrasekhar in presence of a magnetic vector potential, i.e.,
\begin{align}
  \ct^{(c)}(\uA) := \sqrt{c^{2}T(\uA) + c^{4}} - c^{2},
\end{align}
where $T(\uA)$ is either the magnetic Schr\"odinger or Pauli operator
as in \eqref{eq:magnetickinetic}. We consider $\uA$-fields in
\begin{align}
  \tilde\gA := \{\uA\in L^6(\R^3:\R^3):\,{\rm div}(\uA)=0,\,|\nabla\otimes \uA|\in L^2(\R^3)\}.
\end{align}

Let
\begin{align}
  C_{N,V,\uA} = \sum_{\nu=1}^N\left(\ct_\nu^{(c)}(\uA) - V(x_\nu)\right) + \sum_{1\leq\nu<\mu\leq N}\frac{1}{|x_\nu-x_\mu|} + U \quad \text{in}\ \bigwedge_{\nu=1}^N L^2(\R^3:\C^2)
\end{align}
be a Hamilton operator for a relativistic molecule with given
vector potential $\uA$, and $V$ and $U$ as in \eqref{eq:defv}
and \eqref{eq:defu}, respectively.
(Technically, $C_{N,V,\uA}$ is defined as the Friedrichs extension of the
corresponding quadratic form with form domain $\cs(\R^3:\C^2)$,
whenever $\max_\kappa Z_\kappa/c\leq\gamma_C$.)

For admissible vector potentials $\uA\in\tilde\gA$, given nuclear
positions $\uR$, and $N=|\uZ|$ the ground state energy is
\begin{align}
  E_{c,C,{\rm mag}}(\uZ,\uR,\uA) := \inf\spec(C_{N,V,\uA}),
\end{align}
and the total energy including the magnetic field energy is
\begin{align}
  E_{c,C,{\rm mag}}(\uZ,\uR) := \inf_{\uA\in\tilde\gA}\left(E_{c,C,{\rm mag}}(\uZ,\uR,\uA) + \frac{c^2}{8\pi^2}\int_{\R^3}|\nabla\times \uA|^2\right).
\end{align}

The following energetic result for the Scott scale was proved by
Erd\H os et al.~\cite{Erdosetal2012}. 
\begin{theorem}
  \label{scottchandrasekharmoleculemagnetic}
  Let the notations and assumptions be as in Theorem
  \ref{scottchandrasekharmolecule}. In particular, fix $K$,
  $\uz\in(0,1)^K$, and $\ur\in\R^{3K}$. Assume furthermore that there
  is $\gamma_0<2/\pi$ such that $\max_{\kappa}Z_\kappa/c<\gamma_0$.
  Then
  \begin{align}
    E_{c,C,{\rm mag}}(\uZ,\uR) = E^{\rm TF}(\uz,\ur)|\uZ|^{7/3} + \sum_{\kappa=1}^K Z_\kappa^2\left(\frac12-s_C(\gamma)\right) + o(|\uZ|^2)
  \end{align}
  in the limits $|\uZ|,c\to\infty$.
\end{theorem}

\begin{remarks}
  (1) Since $c^{-1}|\uZ|$ is bounded, the prefactor $c^2/(8\pi^2)$ of the
  magnetic energy in the relativistic case is of order $|\uZ|^2$ (at least
  if we assume additionally that $|\uZ|/c$ is kept constant when $|\uZ|,c\to\infty$).
  This is much larger than in the non-relativistic case
  (Theorem~\ref{scottnonrelmagn}), where $c^{-2}|\uZ|$ was bounded.
  Thus, the self-generated magnetic field has to
  be much smaller in the relativistic case, which explains why it
  does not alter the Scott coefficient (contrary to the
  non-relativistic situation). In fact, the prefactor $8\pi$ in
  front of the field energy is irrelevant and can be replaced by
  any other fixed constant in the relativistic situation.
  
  (2) The convergence of the density on the Scott scale has not been
  worked out so far.
\end{remarks}

\section{Open questions}
\label{s:questions}

We collect some questions that -- at least from our perspective -- are
interesting both from physical and mathematical points of view. 

\begin{enumerate}
\item
  For $N\in\N$ and $\uZ\in(0,\infty)^K$, we set
  \begin{equation}
    \label{eq:gs}
    \mathfrak{E}_S(N,\uZ):=\inf \{E_S(N,\uZ,\uR): \uR\in \br^{3K}\}.
  \end{equation}
  This is the ground state energy of a non-relativistic molecule in static
  approximation, provided there is a state $\psi$ in the form domain of
  $H_{N,V}$ and nuclear positions $R\in\R^{3K}$ such that
  $\langle \psi,H_{N,V}\psi\rangle=\mathfrak{E}_S(N,\uZ)$.
  
  Question: Can one prove the Scott conjecture for $\mathfrak{E}_S(N,\uZ)$,
  i.e., when one minimizes over the nuclear positions?
  We expect additivity of the energy up to $o(|\uZ|^{\frac53})$, i.e.,
  \begin{equation}
    \label{eq:add}
    \mathfrak{E}_S(N,\uZ)= \sum_{\kappa=1}^K\left(E^\mathrm{TF}(Z_\kappa) +  \frac q4Z_\kappa^2 -C_\mathrm{DS}Z^\frac53_\kappa\right)+o(|\uZ|^{\frac53}).
  \end{equation}
  Even more, one might ask whether the strong Scott conjecture would
  hold around each of the nuclei if \eqref{eq:add} is true.
  A step in this direction was taken by
  Iantchenko et al.~\cite[Theorem~3]{Iantchenkoetal1996} under the
  additional hypothesis that the minimal nuclear distance is bounded
  from below, e.g., by $\const |\uZ|^{-\frac14}$.

  The analogue of $\mathfrak{E}_S$ can be defined for other many-particle
  models discussed in this review, e.g., those in Section \ref{s:relativistic}.
  As far as we know, the Scott conjecture is also open for these problems.
  Again, we expect additivity of the energy up to order $o(|\uZ|^{5/3})$.
  
  One could also consider a variant of this problem, where the kinetic
  energy of the nuclei is taken into account.
  
\item It is folklore in quantum chemistry that chemical accuracy is
  achieved without taking a self-generated magnetic field into account.
  Therefore, we ask whether in physical models at least the
  Scott conjecture does not depend on the self-generated magnetic field.

\item The Scott conjectures for the energy and the density are open for
  no-pair operators where a self-consistent mean field
  is taken into account. In this case
  the Hilbert space of admissible one-particle states (in the $x\mapsto x/c$
  rescaled picture) is $\Lambda_\chi(L^2(\R^3:\C^4))$ with
  $\Lambda_\chi:=\one_{(0,\infty)}(D_\gamma^H+\chi)$
  with a (not necessary local) mean field $\chi$.
  The resulting no-pair operator is sometimes called Fuzzy operator.
  Possible choices for $\chi$ are the Thomas--Fermi potential (the
  right sides of \eqref{eq:tfeq}-\eqref{eq:tfeq2}), or the Hartree--Fock
  potential
  $$
  \sum_i\left(|\phi_i|^2\ast\frac{1}{|\cdot|}(x) - \frac{\phi_i(x)\overline{\phi_i(\cdot)}}{|x-\cdot|}\right)
  $$
  generated by a set of appropriately chosen orbitals
  $\{\phi_i\}_{i}$. The latter choice is especially popular in
  quantum chemistry, see, e.g.,
  \cite{Johnson1993,Johnson1998,Sapirstein1998,Johnsonetal2004}.
  It turns out, though, that numerically computed values of the ground
  state energy for Coulomb systems with $Z\gtrsim90$ in the Fuzzy picture
  are quite close to those in the Furry picture; see also
  \cite{Mittleman1981,ReiherWolf2009,Saue2011,HandrekSiedentop2015}.
  
  Nevertheless, from a mathematical point of view it is interesting to
  investigate the precise value of the Scott correction for the Fuzzy
  operator.
  The following tasks seem natural.
  \begin{enumerate}[(a)]
  \item Show that for any (reasonable?) choice of the mean field $\chi$,
    the leading order of the ground state energy is still the
    Thomas--Fermi energy. 

  \item Show the Scott conjecture for the Fuzzy model defined in the spirit
    of Mittleman (minimization of the electronic degrees followed by a
    maximization of the splitting).
    
  \item Show that, to within the order of accuracy of Scott's correction,
    the maximization is attained in the Furry picture.
    (Recall that the $Z^2$-correction in the Furry picture is exclusively
    generated by the effective one-particle problem involving the
    hydrogenic operator
    $D_\gamma^H$.
    Thus, our intuition is supported by the variational principle for operators
    with spectral gaps \cite{GriesemerSiedentop1999,MorozovMueller2015},
    which leads to the largest eigenvalues of
    $\Lambda_\chi D_\gamma^H\Lambda_\chi$, when one chooses $\Lambda_\chi$
    to be the Furry projection. This is also the underlying spirit of
    Schwinger's derivation \cite{Schwinger1980} of the relativistic
    Scott correction.)
  \end{enumerate}

\item Can one show first and second order asymptotics for the Lieb--Loss
  model \cite{LiebLoss1999}? See Bach and Hach \cite{BachHach2022} for the
  first order in the non-relativistic setting.

\item Let us formulate some questions regarding the hydrogenic
  densities $\rho_{\#}^H$ with $\#\in\{C,B,F\}$.
  For $\#=F$ one might be able to exploit the explicitly known eigenfunctions
  of the Coulomb--Dirac operator
  (cf.~\cite{Gordon1928,Darwin1928,Pidduck1929,Bethe1933,Thaller1992}) as in
  Theorem \ref{heilmannlieb} by Heilmann and Lieb.
  \begin{enumerate}
  \item Inspired by Theorems \ref{existencerhoh} and \ref{existencerhohdirac},
    we ask whether the densities $\rho_{\#}^H$ satisfy a power law at the origin,
    i.e., whether there are model-dependent constants
    $a_{\gamma,\#},b_{\gamma,\#}>0$ such that
    \begin{align}
      \lim_{|x|\to0}\rho_{\#}^H(x)|x|^{2a_{\gamma,\#}} = b_{\gamma,\#}.
    \end{align}
    In view of Theorems \ref{existencerhoh} and \ref{existencerhohdirac} and
    the fact $|\psi_{n=0,\kappa=\pm1,m}^D(x)| \sim |x|^{-\Sigma_\gamma}$ for
    $m=-j_\kappa,...,j_\kappa$ and $|x|\ll1$
    (cf.~\cite[p.~316]{Bethe1933} or \cite[p.~427]{Thaller2005}),
    it seems natural to believe that
    $a_{\sigma,C}=\sigma_\gamma$ (cf.~\eqref{eq:defsigmagamma}) and
    $a_{\sigma,F}=\Sigma_\gamma$ (cf.~\eqref{eq:defsigmanu}).
    
  \item Electrons far away from the nucleus are expected to behave
    non-re\-la\-ti\-vis\-ti\-cal\-ly. Can one show
    \begin{align}
      \lim_{|x|\to\infty}|x|^{3/2}\rho_{\#}^H(x)=\lim_{|x|\to0}|x|^{3/2}\rho_{z=1}^{\rm TF}(x)  
    \end{align}    
    as in Theorem \ref{heilmannlieb}?
    (Here $\rho_{z=1}^{\rm TF}$ is the hydrogenic TF density with $q=2$.)
    In case $\#=F$ one might be able to derive an asymptotic expansion
    as in \eqref{eq:heilmannlieb1} using the explicitly known eigenfunctions.
    
  \item The non-relativistic density $\rho_S^H(x)$ decreases monotonically
    in $|x|>0$. Is this also true for the relativistic hydrogenic densities?
\end{enumerate}

\item As we have seen in Theorem~\ref{tfwz2theorem} and the ensuing
  discussion, Weizs\"acker's parameter $A$ in the TFW functional
  \eqref{eq:deftfw} can be tuned to achieve \emph{either} energy
  agreement, \emph{or} agreement of the TFW density near the
  origin. This leads to the following (vague) question of whether one
  can construct a modification of Weizs\"acker's gradient term that
  gives \emph{simultaneously} the energy and the density at the origin
  correctly.

\end{enumerate}

\section*{Acknowledgments}

Partial support through U.S. National Science Foundation grant DMS-1954995 (R.L.F.) and by the Deutsche Forschungsgemeinschaft through Germany's Excellence Strategy, grant EXC-2111-390814868 (R.L.F.\& H.S) is acknowledged.


\def\cprime{$'$}

\end{document}